\documentclass[a4paper,12pt]{article}
\usepackage{a4wide}
\usepackage{amsfonts,amsmath,latexsym,amssymb,bbm}    % need for subequations
\usepackage{graphicx,float}
\usepackage[colorlinks=true,linkcolor=blue,citecolor=blue,urlcolor=blue]{hyperref}
\usepackage{graphicx,float}
\usepackage{enumerate}
\usepackage{comment}
\usepackage{tikz}
\usepackage{longtable}
\usepackage{natbib}
\usepackage{tabularx}

\usepackage[normalem]{ulem}

\newcommand{\qed}{{$\square$}}

\newcommand{\calC}{{\cal C}}

\newcommand{\calP}{{\cal P}}

\newcommand{\calU}{{\cal U}}

\newcommand{\calX}{{\cal X}}
\newcommand{\calY}{{\cal Y}}

\newcommand{\veps}{{\varepsilon}}
\newcommand{\eps}{{\epsilon}}

\DeclareMathOperator*{\argmin}{arg\,min}

\numberwithin{equation}{section}

\def\E{{\mathbb{E}}}

\def\eps{\epsilon}
\newcommand{\half}{\frac{1}{2}\:}
\newcommand{\diag}{{\rm diag}}

\newcommand{\real}{{\mathbb{R}}}
\newcommand{\pr}{{\mathbb{P}}}
\newtheorem{theorem}{Theorem}[section]

\newtheorem{lemma}[theorem]{Lemma}

\newtheorem{proposition}[theorem]{Proposition}

\newtheorem{properties}[theorem]{Properties}
\newtheorem{definition}[theorem]{Definition}

\newtheorem{rem}[theorem]{Remark}

\begin{document}
\title{
Comparing Multivariate Distributions: A Novel Approach Using Optimal Transport-based Plots}

\author{Sibsankar SINGHA\textsuperscript{(a)} \\[.5ex]
Supervisors: Marie KRATZ\textsuperscript{(b)} and Sreekar VADLAMANI\textsuperscript{(c)}%\thanks{Corresponding author: xx}
\\[1ex]
\small
\textsuperscript{(a)} TIFR-CAM, Bangalore, India \& ESSEC CREAR, France; Email: sibsankar@tifrbng.res.in \\
\small
\textsuperscript{(b)} ESSEC Business School, CREAR, Cergy-Pontoise, France; Email: kratz@essec.edu \\
\small
\textsuperscript{(c)} TIFR-CAM, Bangalore, India; Email: sreekar@tifrbng.res.in
}

\maketitle
%\tableofcontents{}

\vspace{.5cm}

\begin{abstract}
    
 Quantile-Quantile (Q-Q) plots are widely used for assessing the distributional similarity between two datasets. Traditionally, Q-Q plots are constructed for univariate distributions, making them less effective in capturing complex dependencies present in multivariate data. In this paper, we propose a novel approach for constructing multivariate Q-Q plots, which extend the traditional Q-Q plot methodology to handle high-dimensional data. Our approach utilizes optimal transport (OT) and entropy-regularized optimal transport (EOT) to align the empirical quantiles of the two datasets. Additionally, we introduce another technique based on OT and EOT potentials which can effectively compare two multivariate datasets. Through extensive simulations and real data examples, we demonstrate the effectiveness of our proposed approach in capturing multivariate dependencies and identifying distributional differences such as tail behaviour. We also propose two test statistics based on the Q-Q and potential plots to compare two distributions rigorously.

\end{abstract}

{\it Keywords:} Q-Q plots; multivariate quantile; optimal transport; entropy regularisation; hypothesis testing; geometric quantile; tail behavior

\newpage

\section{Introduction}

Univariate Quantile-Quantile (Q-Q) plot is a graphical tool used in statistics to assess whether two samples follow a similar distribution by comparing their respective quantiles. It is a simple yet powerful visualization technique that provides valuable insights into the shape, location, and scale of two samples.

%{\color{red} (Not needed) Sorting the data in ascending order to evaluate sample quantiles, the corresponding values are then paired up and plotted against each other. If the two sets of samples come from the same distribution, then the Q-Q plot will resemble a straight line, and {\it vice versa}.}

While the univariate Q-Q plot is well defined, there is not much literature on multivariate Q-Q plots. This perhaps can be attributed to the absence of natural ordering. 
\cite{Easton1990} introduced a permutation based multivariate Q-Q plot, where samples of equal length are matched against each other such that the total $\ell^2$ distance between them is minimised. Subsequently, the corresponding components of matched samples are plotted against one another in different bivariate plots. These plots then show visual evidence of similarity of underlying distributions of the two sets of samples. In \cite{Dhar2014}, the authors used the crucial property of unique characterisation of underlying distributions by geometric quantile to develope componentwise Q-Q plots. Specifically, samples are matched according to their geometric rank, and then a componentwise plotting of matched samples provides insight into the underlying distributions.

In this context, we aim at extending this graphical tool using the optimal transport (OT) map, to which we refer as OT Q-Q plot. %and optimal transport potential (OT potential). Throughout this paper, we refer to them as OT Q-Q plot and OT potential plots, respectively. 
%In addition to proposing novel (Q-Q, potential) plots based on the theory of optimal transport, 
We prove that, as the sample size increases, the Q-Q plots are concentrating around the straight line passing through the origin and with slope $1$, if and only if the two sets of samples are drawn from the same distribution. This characteristic enables us to visually compare two distributions.

%{\color{red}[We need to emphasize the benefit of having both potential and Q-Q plots. In part, the potential function as a tool to discriminate between distributions and also as model validation. See the intro of He Eihnmal whenever tehy spoke about scenarios]}

Moreover, in the context of very large dimensions, the componentwise plot may not be ideal due to a large number of bivariate plots. So, we suggest an alternative approach based on OT potential function, which allows for the comparison of two sets of multivariate samples in a single bivariate plot, named OT potential plot. Similar to the componentwise plots, we establish that, as the sample size increases, the potential plot becomes arbitrarily close to the straight line with slope $1$ and passing through the origin if and only if the underlying distributions of the two samples are identical. 

Computing empirical OT maps can be costly, especially when sample sizes are large; see \cite{Huetter2021}.
Therefore, practical solutions have been proposed in the literature, one of the most popular approaches being the entropy regularisation, introduced in \cite{Cuturi2013}. By selecting the regularisation parameter to be sufficiently small, the entropy regularised OT (EOT) map and EOT potential can closely approximate the OT map and OT potential, respectively. Moreover, we prove 
%in Proposition \ref{thm:unique characterisation by EOT} 
that the EOT map and EOT potential can also uniquely characterize the distribution. 

These significant properties of EOT map and EOT potential serve as our driving force to also build Q-Q plots (we call it Q-Q, although EOT maps are not quantiles) and potential plots based on them. Finally, as given in \cite{Shapiro1965, Dhar2014}, we propose test statistics to assess the relevance of the proposed Q-Q plots and potential plots.

We apply this (E)OT approach on examples from simulated and real data, then compare (E)OT Q-Q plots with geometric ones.

%%%
Our contributions are the following:
\begin{itemize}
    \item[(i)] We construct OT Q-Q plots for comparing multivariate distributions, as we can prove that these plots are concentrating, as the sample size increases, around the straight line passing through the origin and with slope $1$, if and only if the two sets of samples are drawn from the same distribution. 
    \item[(ii)] We also propose OT potential plot, which is shown to share the same unique characterization property as for OT maps. This tool is quite interesting as it gives a single bivariate plot, whatever the dimension of the distribution.
    \item[(iii)] We develop the same tools and characteristic properties as for OT when considering the EOT approach, widely used because less computationally challenging than the OT approach.  
    \item[(iv)] We propose test statistics to assess the relevance of the  Q-Q and potential plots as tools for comparing multivariate distributions, for the EOT approach.
    \item[(v)] We question the benefit of using (E)OT Q-Q or potential plots, to retrieve specific features of the distributions considered, specifically in the tail. We do so via an extensive simulation study.
    %\item[(vi)] Potential plots allow one to compare two $d$-dimensional distributions, for any $d\ge 2$, in a bivariate plot, which is certainly a strong advantage. Nevertheless, it is less obvious to detect visually some of the features of the distributions than with (componentwise) Q-Q plots. Reason why the pair of these two types of plots constitutes a very good graphical tool for distributions comparison.    
    \item[(vi)] We compare (E)OT Q-Q plots with geometric Q-Q plots developed in \cite{Dhar2014}, on simulated and real data. We observe that the former is better performing visually in general (on our examples) than the latter to characterize tail distributions. 
\end{itemize}
\vspace{1ex}
{\it Notation:}  $\|\cdot\|$ denotes the $\ell^2$ norm, $P^{conv}_d$  the set of all probability measures on $\real^d$ with convex support, $\langle \cdot \rangle$ the usual inner product on $\real^d$, and $X  \sim \nu$ means that $X$ has  distribution $\nu$.  $\pi\ll P$ implies that the measure $\pi$ is absolutely continuous with respect to the measure $P$. The measure $\mu$ is assumed to be uniform on the open unit ball $B^d$ unless otherwise specified. $\Pi(\mu,\nu)$ is the set of all couplings with first marginal $\mu$ and second marginal $\nu$, and $\nabla \equiv (\frac{\partial}{\partial x_1},\cdots,\frac{\partial}{\partial x_n})$. Let $\nu$ be a probability measure and $K$ be a $\nu-$measurable set with $\nu(K)>0$. Then define the probability measure $\nu|_K(\cdot)$ such that $\nu|_K(A)=\frac{\nu(A \cap K)}{\nu(K)}$ for all $\nu-$measurable sets $A$.

%$\Phi_0$ is the set of all convex pairs $(\phi,\psi)$ such that $\phi$ and $\psi$ are convex conjugates and $\phi(0)=\psi(0)=0$.

\vspace{1ex}
\textbf{\it Structure of the paper:}
In Section \ref{section:OT and EOT review}, we briefly recall the concepts related to OT, EOT maps, and potentials, together with their connections to the multivariate quantile function. Section \ref{sec:mult Q-Q plot} includes discussions and results related to constructions of multivariate Q-Q plots and potential plots based on OT and EOT. In Section \ref{sec: Test statistics}, we propose two test statistics and study their asymptotics. In Section~\ref{sec: Illustration}, we illustrate our (E)OT approach on simulated samples to compare multivariate distributions on various scenarios and experiments. In Section~\ref{sec:ex_real_data}, we apply the method on real data. In Section~\ref{sec:comparisonGeom}, we compare results obtained with geometric and (E)OT Q-Q plots, respectively, first on simulated data, then on real data.

\section{Background and preliminaries on multivariate quantile and potential function} 
\label{section:OT and EOT review}
 \subsection{Optimal transport approach}

 The motivation for defining multivariate quantile through optimal transport arises from certain inherent properties of the univariate quantile functions. One key property of the univariate quantile function is its nondecreasing nature. Another compelling aspect is that any probability measure is the pushforward of the uniform measure on the interval $(0,1)$ by its quantile function. These two properties together uniquely characterize a quantile function; any function meeting these criteria must be a quantile. This property of transporting measure in high dimensions is at the heart of the theory of optimal transport and has a very rich literature (see \cite{villani2003topics, villani2009oldnew}, \cite{Santambrogio2015} for an exhaustive review). In particular, the following theorem due to \cite{Brenier1991PolarFA} and \cite{McCann1995Existence}, proves the existence and uniqueness of a map $T$ that pushes forward a measure $\mu$ to $\nu$ if $\mu$ is absolutely continuous with respect to the Lebesgue measure. Moreover, the map $T$ is the gradient of a convex function, i.e. $T=\nabla \phi$. The property that $T$ can be written as the gradient of a convex function, is analogous to the monotonicity property of univariate quantile function.

\begin{theorem}{\bf (Brenier--McCann)}
\label{Thm:Brennier-McCann}
    Let $\mu$ and $\nu$ be two probability measures on $\real^d$. Assume that $\mu$ is absolutely continuous with respect to the Lebesgue measure. Then, there exists a convex function $\varphi$ on $\real^d$, which is $\mu$ almost everywhere unique up to an additive constant, such that $\nabla \varphi$ pushes forward the measure $\mu$ to the measure $\nu$, notationally $\nabla \varphi \sharp \mu =\nu$. If $\nu$ is also absolutely continuous with respect to the Lebesgue measure, then $\nabla \varphi$ is invertible and $(\nabla\varphi)^{-1}=\nabla \varphi^*$, where $\varphi^*$ is the Fenchel--Legendre transform of $\varphi$. 
   Additionally, if both $\mu$ and $\nu$ have finite second moment, then $\nabla\varphi$ uniquely solves the {\it Monge problem}:
   \begin{equation}
    \label{eqn:Monge-problem}
    \argmin_{T:T\sharp \mu =\nu} \int_{\real^d}\|x-T(x)\|^2 \mu(dx).
 \end{equation}
\end{theorem}
Let $\mu$ be the uniform probability measure on the unit ball $B^d$ and $\nu$ be any probabilty measure on $\real^d$, then Theorem \ref{Thm:Brennier-McCann} ensures the existence of a unique monotone map (the gradient of a convex function), which pushes forward $\mu$ to $\nu$. 
%Therefore it generalizes all the properties of univariate quantiles listed in Proposition \ref{properties: univ quantile}. 
The following definition of multivariate quantile function is from \cite{Chernozhukov2017} with minor modification, as explained in Remark~\ref{rk:modification_dfn_quantile}.
\begin{definition}{\bf (OT quantile; \cite{Chernozhukov2017})} Let $\mu$ be the uniform probability measure on $B^d$. The function  defined by $T_{\nu}=\nabla \varphi |_{B^d}$ such that  $\nabla \varphi\sharp \mu=\nu$ for a convex function $\varphi$, is defined as the OT quantile function of the probability measure $\nu$.
\end{definition}
Although the OT quantile is almost everywhere unique, the convex functions in Theorem \ref{Thm:Brennier-McCann} are unique up to additive constants. Therefore setting its infimum to $0$ will ensure their uniqueness.
\begin{definition}{\bf (OT Potential; \cite{Chernozhukov2017})} \label{dfn:Monge potential}
    Let $\mu$ be the uniform probability measure on the unit ball $B^d$. A convex function $\varphi_{\nu}$ satisfying $\nabla \varphi_{\nu}\sharp \mu=\nu$  and $\displaystyle\inf_{u \in B^d}\varphi_{\nu}(u)=0$ is called the OT potential of $\nu$.
\end{definition}
\begin{rem}
    \label{rk:modification_dfn_quantile}
    \begin{enumerate}
        \item Since any convex function on bounded open set is bounded below by a finite number we can subtract the infimum to satisfy the condition $\displaystyle\inf_{u \in B^d}\varphi_{\nu}(u)=0$. This will ensure that the potential function is non-negative everywhere.
        \item In the definitions above, the fixed measure $\mu$ is assumed to be the uniform probability measure on the unit ball. However, one can choose a different reference distribution; for instance, in \cite{Chernozhukov2017}, \cite{Hallin2021}, $\mu$ is considered to be the spherically uniform measure. We refer the reader to \cite[Remark 3.11]{Ghosal2022} for discussion on the effect of choosing different reference measures.
    \end{enumerate}
\end{rem} 
Later in this paper, we will extensively use the following properties of both OT quantile and OT potential, which are a straightforward consequence of Theorem~\ref{Thm:Brennier-McCann}.
\begin{properties}{\bf(Unique characterisation of measure)} 
\label{Prop:unique characterisation of the measure}
Let $\nu_1$ and $\nu_2$ be two probability measures. Then,
    \begin{enumerate}[(i)]
        \item $\nu_1=\nu_2$ if and only if $\varphi_{\nu_1}(u)=\varphi_{\nu_2}(u),\,\,\, \text{ for }\mu-\text{almost every } u \in B^d$, where $\varphi_{\nu_1},\varphi_{\nu_2}$ are OT potentials.
        \item $\nu_1=\nu_2$ if and only if $T_{\nu_1}(u)=T_{\nu_2}(u), \,\, \text{ for }\mu-\text{almost every } u \in B^d$, where $T_{\nu_1},T_{\nu_2}$ are the OT quantiles.
    \end{enumerate}
\end{properties}

\paragraph{Empirical OT quantile - }
When a distribution on $\mathbb{R}^d$ is supported on finite points (i.i.d. samples), the empirical quantile function can be defined in either of the following ways: 
\begin{itemize}
    \item Discrete to discrete: In this case, the empirical quantile function is an optimal transport map between two sets of samples. Thus, clearly, the quantile function is supported on the finite sample, generated uniformly on $B^d$. 
    \item Continuous to discrete or semi-discrete: In this case, the quantile function is the transport map between the uniform measure on $B^d$ and a finite sample. Therefore, the map is defined on the whole unit ball. See \cite{Chernozhukov2017} for more details.
\end{itemize}
  For simplicity, we consider the discrete to discrete version of empirical OT quantiles; and refer the reader to \cite{Chernozhukov2017,Hallin2021,Ghosal2022} for more properties on empirical OT quantiles.

For the purpose of this work, we shall consider i.i.d. samples ${\cal X}^n=\{X_1,X_2, ...,X_n\}$ with $X_i \sim \nu$ and ${\cal U}^n =\{U_1,U_2,\cdots,U_n\}$ with $U_i \sim \mu$, $\mu$ being the uniform measure on $B^d$. Then, we define the empirical measures corresponding to ${\cal X}^n$ and ${\cal U}^n$ as

\begin{equation}
    \label{eq:enp_measure}
    \nu_n(\cdot)=\frac{1}{n}\sum_{i=1}^n\delta_{X_i}(\cdot)  \qquad \text{and} \qquad  \mu_n(\cdot)=\frac{1}{n}\sum_{i=1}^n\delta_{U_i}(\cdot).
\end{equation}

\begin{definition}{\bf (Empirical OT quantile; \cite{Hallin2021})}
The empirical quantile function of the empirical measure $\nu_n$ with respect to the reference measure $\mu_n$ is a function $T_{\nu_n}:{\cal U}^n \to {\cal X}^n$ such that
\begin{equation}
\label{eqn:dfn of emp quantile}
    \sum_{i=1}^n\|U_i-T_{\nu_n}(U_i)\|^2=\min_{\sigma_n\in {\cal S}_n} \sum_{i=1}^n\|X_i-U_{\sigma_n(i)}\|^2 ,
\end{equation}
where ${\cal S}_n$ denotes the set of permutations of the set $\{1,\ldots,n\}$.
\end{definition}
\begin{rem}
    The minimization problem \eqref{eqn:dfn of emp quantile} resembles the {\it Monge problem} in discrete setting. Any solution to this problem satisfies the cyclical monotonicity property, which is a unique feature of the OT map (see \cite{villani2003topics,villani2009oldnew}).
\end{rem}
\begin{rem}
     Observe that the function $T_{\nu_n}$ is random, and is supported on the random set ${\cal U}^n$. 
%Even for a fixed realisation of ${\cal X}^n$, the empirical quantile function $T_{\nu_n}$ remains random due to the inherent randomness in ${\cal U}^n$. Notice that the definition of empirical OT quantile of $\nu_n$ depends upon the particular observation of the random set ${\calU}^n$. 
Although in our notation the empirical OT quantile $T_{\nu_n}$ seems to depend only on the measure $\nu_n$, note that it is defined with respect to a reference $\calU^n$. In \cite{Hallin2021}, instead of taking a random sample ${\cal U}^n$, the authors construct a deterministic grid on $B^d$ to define empical OT quantiles.
\end{rem}
\begin{rem}
    The map $T_{\nu_n}$ may not be unique. For example, let  $\,{\cal U}^2=\{(0,0),(.5,.5)\}$ and ${\cal X}^2=\{(0,.5),(.5,0)\}$, and note that, there are only two possible bijective maps from $\calU^2$ to $\calX^2$ and the cost remains the same for each of them. Therefore there are two possibilities for the map $T_{\nu_2}$. However, since ${\cal U}^n$ is a sample drawn from an absolutely continuous distribution $\mu$, the event that $T_{\nu_n}$ is not unique has probability zero; see \cite{Hallin2021}.
\end{rem}
\begin{comment}
\begin{rem}
    {\color{red}In the definition of the empirical OT map, it is necessary to keep the sample sizes of $\calX^n$ and $\calU^n$ the same. For different sample sizes, it can be defined as the barycentric projection of the optimal plan. However, this is beyond the scope of this work. }
\end{rem}
\end{comment}
For the empirical OT potential, we will adopt the definition from \cite{Chernozhukov2017}, where it is defined such that its gradient coincides with the empirical quantile function on its finite support. The definition is based on the dual formulation of the {\it Monge problem} \eqref{eqn:Monge-problem}. %\eqref{eqn:Monge prob gen d}.
%Similarly as the population version, we want to define a convex function such that its gradient coincides with empirical OT quantile function. 
\begin{definition}{\bf (Empirical OT potential; \cite{Chernozhukov2017})}  
    Let $\mu_n$ and $\nu_n$ be as defined in \eqref{eq:enp_measure}.  The empirical OT potential of $\nu_n$ with respect to  $\mu_n$ (the reference measure) is a convex function $\varphi_n$ such that $(\varphi_n, \varphi_n^{*}) \in \Phi_0$ solves the following dual problem 
\begin{equation}
    \int \varphi_n \,d\mu_n + \int \varphi_n^* \,d\nu_n=\inf _{(\psi, \psi^*) \in \Phi_0} \left(\int \psi \,d\mu_n + \int \psi^* \,d\nu_n\right),
\end{equation}
where $\Phi_0$ denotes the set of all pairs of convex conjugates $(\psi,\psi^*)$ such that \\$\displaystyle \inf_{u \in B^d}\psi(u)=0$.
\end{definition}

Note that $\nabla \varphi_n$, when restricted to the set ${\cal U}^n$, coincides with the empirical quantile function $T_{{\cal X}^n}$.

    \begin{lemma}
    \label{lemma:quantile convergence}
        Let $\nu$ be an absolutely continuous probability measure with convex support. Let $T_{\nu}$ and $T_{\nu_n}$ be population and empirical quantile functions with respect to $\mu$ and $\mu_n$, respectively. Then, for any compact set $K \subset B^d$,
        \begin{equation*}
            \max_{\substack{U_i \in K \\{1\leq i \leq n}}}\left|T_{\nu_n}(U_i)-T_{\nu}(U_i)\right| \stackrel{a.s.}{ \longrightarrow} 0, \,\,\,\,\,\, \text{ as } n \to +\infty.
        \end{equation*}
    \end{lemma}
\begin{proof}
     Since $\nu$ is absolutely continuous, the quantile function $T_{\nu}$ is bijective. The inverse of the quantile is defined as the multivariate distribution function or OT distribution function and denoted by $F_{\nu}$, i.e. $T_{\nu}=F^{-1}_{\nu}$. Similarly, the empirical distribution function $F_{\nu_n}$ is the inverse of the empirical quantile, i.e. $T_{\nu_n}=(F_{\nu_n})^{-1}$. 
     
     Note that the support of $\nu$ is convex, $T_{\nu}:B^d\to \real^d$ is a homeomorphism between $B^d$ and interior of $supp(\nu)$. Let $T_{\nu_n}(U_i)=X_{\sigma_n(i)}$. Then we can write, 
        \begin{equation*}
            \label{eqn:inv-identity}
            \left|T_{\nu_n}(U_i)-T_{\nu}(U_i)\right|=  \left|X_{\sigma_n(i)}-T_{\nu}(U_i)\right|
             = \left|T_{\nu}(F_{\nu}(X_{\sigma_n(i)}))-T_{\nu}(F_{\nu_n}(X_{\sigma_n(i)}))\right|,
        \end{equation*}
    which, combined with Proposition 2.4 in \cite{Hallin2021} and the fact that $T_{\nu}$ is uniformly continuous on compact sets, concludes the proof of Lemma~\ref{lemma:quantile convergence}.\hfill\qed
\end{proof}
\begin{lemma}
 Assume that the support of $\nu$ $(supp(\nu))$ is compact and the optimal transport map $T_{\nu}$ is a homeomorphism from the open unit ball $B^d$ to the interior of $supp(\nu)$. Then, for any compact $K \subset B^d$, 
\begin{equation*}
            \sup_{u \in K }\left|\varphi_{\nu_n}(u)-\varphi_{\nu}(u)\right| \stackrel{\pr^*}{ \longrightarrow} 0, \,\,\,\,\,\, \text{ as } n \to +\infty,
\end{equation*}
where $\stackrel{\pr^*}{ \longrightarrow}$ denotes the convergence in outer probability in the sense of \cite{Vaart1996}.
\end{lemma}
\begin{proof}
    It is a straightforward adaptation of the proof of \cite[Theorem 3.1]{Chernozhukov2017} developed for OT maps, to OT potentials.
\end{proof}

\subsection{Entropy regularised optimal transport}
While the optimal transport map provides a meaningful definition of the multivariate quantile function, it comes along with practical challenges. Estimating the optimal transport for large samples is costly, with $O(n^3\log n)$ complexity for a sample size $n$. \cite{Cuturi2013} proposed an entropy regularised optimal transport (EOT), which closely approximates the optimal transport, and reduces the complexity to $O(n^2)$. 
%Although one of the the main motivation behind studying EOT in the literature was to speed up the estimation of OT, here we view it as an independent object of interest, rather than just approximating the optimal map. 
In the following, we shall briefly review the basic concepts and properties of EOT.
\\
Given two measures $\mu$ and $\nu$, a weak version of the {\it Monge problem}, known as the {\it Kantorovich relaxation}, is given by
\begin{equation}
\label{prblm:Kantorovich relaxation}
    \argmin_{\pi\in \Pi(\mu,\nu)} \int_{\real^d \times \real^d} \half \|x-y\|^2 d\pi(x,y),
\end{equation}
where $\Pi(\mu,\nu)$ is the set of all couplings between $\mu$ and $\nu$. Since $\Pi(\mu,\nu)$ is compact in the weak topology and the cost function is lower semicontinuous, therefore the existence of a minimiser to the {\it Kantorovich problem} is guaranteed (see \cite[Theorem $4.1$]{villani2009oldnew}), whereas the minimiser does not always exist in the {\it Monge problem}. Additionally, the {\it Kantorovich problem} has a unique minimiser, denoted by $\pi^*$, whenever both $\mu$ and $\nu$ have finite second moment and $\mu$ is absolutely continuous with respect to the Lebesgue measure. Furthermore, the unique minimiser $\pi^*$ is supported on the graph of the OT map (Brenier map) between $\mu$ and $\nu$. In particular, $\pi^*=(\nabla \varphi, I)_{\sharp}\mu$, where $\nabla \varphi$ is the OT map and the convex function $\varphi$ is the OT potential. Alternatively, we can write
\begin{equation}
    \label{eqn:OT map through conditional expectation}
    \nabla \varphi(u)=E_{\pi^*}[X|U=u], \text{ for $\mu$ - almost every $u$, }
\end{equation}
where $(U,X) \sim \pi^*$. This establishes link between OT map ($\nabla \varphi$) and the solution $(\pi^*)$ of the {\it Kantorovich problem}. %
%An alternative way of defining $\pi^*$ is as follows: For any two borel sets $A,B$ in $\real^d$ 
%\begin{equation}
    %\pi^*(A\times B)=\mu(A \cap (\nabla \varphi)^{-1}(B)).
%\end{equation}
%For a more detailed discussions on the statements made above, refer to \cite{Santambrogio2015}. 

%So far we have discussed about {\it Kantorovich problem} and 
%Since $\pi^*$ is a coupling between two measures $(\mu,\nu)$ and supported on the graph of $\nabla \varphi$, the function $\nabla \varphi$ can be recovered as a conditional expectation of the second marginal given the first, also known as barycentric projection.\\
We will now discuss the entropy regularised version of the {\it Kantorovich problem}, which is defined by adding an extra penalty in the cost function $\int_{\real^d \times \real^d} \half \|x-y\|^2 d\pi(x,y)$. Specifically, for $\veps>0$ and probability measures $\mu$ and $\nu$ which have finite second moment, the entropy regularised {\it Kantorovich problem} is defined as
\begin{equation}
\label{EOT}
\argmin_{\pi \in \Pi(\mu,\nu)} \left(\int_{\real^d \times \real^d} \half\|u-x\|^2 d\pi(u,x) + \veps KL(\pi|P)\right),
\end{equation}
where $P=\mu \otimes \nu$ (the product measure) and $KL(\pi|P)$ is the the Kullback–Leibler divergence given by:
$$KL(\pi|P) = \left\{\begin{array}{cc}\E_{\pi}[\log\frac{d\pi}{dP}], & \text{if } \pi\ll P, \\ \infty, & \text{otherwise.} \end{array}\right.$$\\
This problem is known to have a unique minimiser $\pi^*_{\veps}\in \Pi(\mu,\nu)$, as proven in \cite[Theorem $4.2$]{Nutznote2022}.\\
Now, we recall the relation between OT map and the {\it Kantorovich solution} in \eqref{eqn:OT map through conditional expectation}, and analogously define the EOT map. In particular, the EOT map is defined as the conditional expectation with respect to $\pi^*_{\veps}$.
\begin{definition}{\bf (EOT map; \cite{pooladian2022entropic})}\label{dfn:EOT map}
    Let $\nu$ be any probability measure defined on $\real^d$ with finite second moment. Let $\mu$ be the uniform probability measure on $B^d$. Consider $(U,X) \sim \pi^*_{\veps}$, then the EOT map $T^{\veps}_{\nu}(u)$ is defined as:
    \begin{equation*}
        \label{eqn:EOT map}
            T^{\veps}_{\nu}(u)=E_{\pi_{\veps}^*}[X|U=u],\, \,\, \text{for }\mu-\text{almost every } u.
        \end{equation*}
\end{definition}
We shall now explore an alternative representation of the EOT map. Specifically, analogous to OT map, we observe that EOT map can also be represented as the gradient of some potential function (not convex in general). This is achieved in duality. The EOT problem \eqref{EOT} admits strong duality, shown in \cite[Theorem 7]{Genevay2019}, in the following way: let $C_{\veps}$ denote the cost of the regularised version of the {\it Kantorovich problem}, i.e. 
\begin{equation}
    C_{\veps}= \min_{\pi \in \Pi(\mu,\nu)} \left(\int \half\|u-x\|^2 d\pi(u,x) + \veps KL(\pi|P)\right),
\end{equation}
then, 
\begin{equation}
    \label{eqn:EOT dual}
    C_{\veps}=\max_{(\varphi^{\veps},\psi^{\veps})\in L^1(\mu) \times L^1(\nu)} \left(\int \varphi^{\veps} d\mu + \int \psi^{\veps} d\nu + \veps\int e^{\frac{\varphi^{\veps}(u) + \psi^{\veps}(x)-\half \|u-x\|^2}{\veps}}d(\mu\otimes \nu) -\veps\right).
\end{equation}
If $\mu$ and $\nu$ have finite second moment, the maximiser of \eqref{eqn:EOT dual} is unique, $\mu\otimes \nu$--almost surely, up to an additive constant. More precisely, if $(\varphi^{\veps}_1, \psi^{\veps}_1)$ is another pair of solutions to the dual Kantorovich problem, then there exists a constant $a$ such that $\varphi^{\veps}_1=\varphi^{\veps}+a$ and $\psi^{\veps}_1=\psi^{\veps}-a$.
Importantly, the EOT map $T^{\veps}_{\nu}$ can be expressed as the gradient of an EOT potential function (see Section 3 in \cite{pooladian2022entropic}), i.e.
\begin{equation}
    \label{eqn:EOT map as grad of potential}
    T^{\veps}_{\nu}(u)=u-\nabla \varphi^{\veps}(u)=\nabla \left(\frac{\|u\|^2}{2}-\varphi^{\veps}(u)\right).
\end{equation}
\begin{definition}{\bf(EOT potential)}\label{dfn:EOT potential}
    Let $(\varphi^{\veps},\psi^{\veps})$ be a maximiser of \eqref{eqn:EOT dual}, where $\mu$ is the uniform probability measure on the unit ball $B^d$ and $\nu$ is any probability measure defined on $\real^d$ with finite second moment. Then the function defined by $$\varphi^{\veps}_{\nu}(u)=\frac{\|u\|^2}{2}-\varphi^{\veps}(u)-\left[\varphi^{\veps}(u_0)-\frac{\|u_0\|^2}{2}\right],\,\, u \in B^d,$$ where $u_0 \in B^d$ is a fixed point, is called the EOT potential corresponding to $\nu$.
\end{definition}
\begin{rem}
    Whenever minimisers exist, we will set $\displaystyle u_0=\argmin_{u \in B^d} \left[\varphi^{\veps}(u)-\frac{\|u\|^2}{2}\right]$, which will ensure that the EOT potential is non-negative everywhere.
\end{rem}
Now, we move to the unique characterisation property of EOT maps and potentials. 
%Similar properties for OT ones are given in Proposition~\ref{Prop:unique characterisation of the measure}.
\begin{proposition}
    \label{thm:unique characterisation by EOT}
        Let $\mu$ be the uniform measure on the unit ball $B^d$ and $\nu_1$, $\nu_2$ be two probability measures with finite second moment. Let $\veps>0$ be fixed.
        \begin{itemize}
            \item[(i)] Let $\varphi_{\nu_1}^{\veps}$ and $\varphi_{\nu_2}^{\veps}$ be EOT potentials corresponding to $\nu_1$ and $\nu_2$, respectively. Then, $\nu_1=\nu_2$ if and only if $\varphi_{\nu_1}^{\veps}(u)=\varphi_{\nu_2}^{\veps}(u),$ for $\mu$-almost every $u$.
            \item[(ii)] Consider two EOT maps $T^{\veps}_{\nu_1}$ and $T^{\veps}_{\nu_2}$ corresponding to the measures $\nu_1$ and $\nu_2$. Then, $\nu_1=\nu_2$ if and only if $T^{\veps}_{\nu_1}(u)=T^{\veps}_{\nu_2}(u)$, for $\mu$-almost every $u$.
        \end{itemize}
\end{proposition}
\noindent{\bf Proof:} (i)  If $\nu_1=\nu_2$, the proof follows from the fact that the optimisation problem in \eqref{eqn:EOT dual} has a unique minimiser up to the addition of a constant.\\
    Now let us prove the reverse implication by contradiction. Suppose $\varphi^{\veps}_{\nu_1}(u) = \varphi^{\veps}_{\nu_2}(u)$ for $\mu$ almost every $u$ and for some $\veps>0$, but $\nu_1 \neq \nu_2$. Since any pair of EOT potentials $(\varphi^{\eps}_{\nu},\psi^{\eps}_{\nu})$ satisfies the Schr\"odinger system (see \cite[Section $1$]{Nutznote2022}), 
    \begin{equation*}
        \int e^{\frac{\varphi^{\veps}_{\nu}(u) + \psi^{\veps}_{\nu}(x)-\half \|u-x\|^2}{\veps}}d\mu(u) =1, \qquad\text{for } \nu-\text{ almost every }x,
    \end{equation*}

    \begin{equation*}
        \int e^{\frac{\varphi^{\veps}_{\nu}(u) + \psi^{\veps}_{\nu}(x)-\half \|u-x\|^2}{\veps}}d\nu(x) =1, \qquad\text{for } \mu-\text{ almost every }u,
    \end{equation*}
    
    we have
    \begin{equation}
        \begin{aligned}
        \psi^{\veps}_{\nu_1}(x)=-\veps \log \int e^{\frac{\varphi^{\veps}_{\nu_1}(u) -\half \|u-x\|^2}{\veps}}d\mu(u)
        =\psi^{\veps}_{\nu_2}(x)
        \end{aligned}
    \end{equation} 
where the second equality follows by replacing $\varphi^{\veps}_{\nu_1}(u)$ with $\varphi^{\veps}_{\nu_2}(u)$. Then, the proof follows from \cite[Theorem $4.2$]{Nutznote2022}.\\
%{\color{red}Hence the joint measure in \eqref{eqn:EOT plan} for the pairs $(\mu,\nu_1)$ and $(\mu,\nu_2)$ are the same,  which then implies that marginals are identical i.e., $\nu_1=\nu_2$.}\\
(ii)  The proof is  similar to (i) after observing that the EOT map can be written in terms  of potential function, i.e. $T^{\veps}_{\nu}(u)=u-\nabla \varphi^{\veps}_{\nu}(u)$ ( see \eqref{eqn:EOT map as grad of potential}).\hfill\qed
\\

Let us now consider the empirical counterpart of EOT map and potential, which will be later used to develop Q-Q and potential plots.
\paragraph{Empirical  EOT map and its convergence} 
\label{sec:empirical EOT}

Let ${{\cal X}_n}=\{X_1, \cdots , X_n\}$  and ${\cal U}^n=\{U_1,U_2,\cdots,U_n \}$ be two sets of i.i.d. samples drawn from $\nu$ and $\mu$ (the uniform measure on $B^d$), respectively. Let $\nu_n$ and $\mu_n$ be empirical measures induced from $\calX^n$ and $\calU^n$, as defined in \eqref{eq:enp_measure}. 

%$\displaystyle\nu_n(\cdot)=\frac{1}{n}\sum_{i=1}^n\delta_{X_i}(\cdot)$, where $\delta_{X_i}$ is the Dirac measure concentrated at the point $X_i$. Similarly define $\mu_n(\cdot)=\frac{1}{n}\sum_{i=1}^n\delta_{U_i}(\cdot)$, where ${\cal U}^n=\{U_1,U_2,\cdots,U_n \}$ is an i.i.d. sample from the uniform distribution $\mu$ on $B^d$. 
%%%%%%%%%%%%%%%%%%TO LOOK LATER%%%%%%%%%%
%{\color{red} paragraph style does not look good}
\begin{definition}{\bf (Empirical EOT map and potential)}\label{dfn:emp EOT map and potential}
    The empirical EOT map and potential are obtained by replacing $\nu,\mu$ with their empirical counterparts $\nu_n,\mu_n$ in Definitions \ref{dfn:EOT map} and \ref{dfn:EOT potential}, respectively. We denote the empirical EOT map by $T^{\veps}_{{\cal X}_n}$ and the potential by $\varphi^{\veps}_{{\cal X}^n}$.
\end{definition}
\begin{rem}
    \label{rk:empirical EOT}
    \begin{enumerate}
        \item \label{rk:supp eot map} Notice that, although the measure $\mu_n$ is supported on a finite set $\mathcal{U}^n$, the empirical EOT map $T^{\varepsilon}_{{\mathcal{X}}^n}$ is by default defined on $B^d$. However, recall that the empirical OT quantile is defined only on sample points $\mathcal{U}^n$.
        \item If $\nu$ is supported on a compact set, then, for a fixed $\veps>0$, the empirical function $(T^{\veps}_{{\calX}^n},\varphi^{\veps}_{{\cal X}^n})$ converges almost surely to $(T^{\veps}_{\nu},\varphi^{\veps}_{\nu})$ as $n \to \infty$; see \cite[Lemma $1(iii)$]{goldfeld2023limit}.
        \item Note that in Definition~\ref{dfn:emp EOT map and potential}, the number of observations in $\calX^n$ and $\calU^n$ is the same. In the case of two different sample sizes, i.e. $\calX^n$ and $\calU^m$ with $n \neq m$, we can similary define emprical EOT map and potential by replacing $\mu_n$ with $\mu_m$ in the definition.
        %\item non compact {\color{red}(It should be defined under the definition of population EOT map)}.
    \end{enumerate}
\end{rem}

 \section{Construction of multivariate Q-Q plots}
 \label{sec:mult Q-Q plot}

 Let ${\cal X}^n=\{X_1,\cdots,X_n\}$ and ${\cal Y}^n=\{Y_1, \cdots, Y_n\}$ be two sets of i.i.d. samples generated from $\nu_X$ and $\nu_Y$, respectively, and ${\cal U}^n$ generated from the uniform measure on $B^d$. Consider $T_{{\cal X}^n}$ and $T_{{\cal Y}^n}$ the empirical quantile functions with respect to the reference sample ${\cal U}^n$. For each $i\in \{1,\cdots,d\}$ and $K \subset B^d$, define the set
 \begin{equation}\label{eqn:construction of delta one sample}
 \Delta_{i}({\cal X}^n,{\cal Y}^n:{\cal U}^n \cap K)=\left\{\left( \left\langle T_{{\cal X}^n}(U_j),e_i\right\rangle, \langle T_{{\cal Y}^n}(U_j),e_i\rangle\right): U_j \in K, j=1, \cdots,n\right\},
\end{equation} 
 where $e_i$ is the $i-$th canonical basis vector. Therefore $\Delta_{i}({\cal X}^n,{\cal Y}^n:{\cal U}^n \cap K)$ is a subset of $\real^2$, where each element is a pair of the $i-$th components of the OT quantile functions $T_{{\cal X}^n}$ and $T_{{\cal Y}^n}$. 
 Now consider the empirical optimal potentials $\varphi_{{\cal X}_n}$ and $\varphi_{{\cal Y}_n}$. Then, for $K\subset B^d$, define
\begin{equation}\label{eqn:construction of Gamma}
    \Gamma({\cal X}^n,{\cal Y}^n:{\cal U}^n \cap K)=\left\{ (\varphi_{{\cal X}_n}(U_j),\varphi_{{\cal Y}_n}(U_j)): U_j \in K, \,\,j=1,\cdots,n \right\}.
\end{equation}
In Theorem \ref{thm: distinguish two sample OT}, we show that, as $n$ increases, $\Delta_{i}({\cal X}^n,{\cal Y}^n:{\cal U}^n \cap K)$ and $\Gamma({\cal X}^n,{\cal Y}^n:{\cal U}^n \cap K)$ become arbitrarily close to the straight line with slope $1$ and intercept $0$  if and only if $\nu_X(\cdot)=\nu_Y(\cdot)$. With this property, we can define the OT Q-Q plots and OT potential plots as follows: 
 \begin{definition}
    \label{dfn:OT plots}
    %%TO LOOK AFTER%%% EOT Q-Q PLOT DEFINITION MISSING%%%%%%
     Let $K\subset B^d$ be a compact set. An OT Q-Q plot is a collection of $d$ individual scatter plots, where the $i$--th one displays the set $\Delta_{i}({\cal X}^n,{\cal Y}^n:{\cal U}^n \cap K)$ defined in \eqref{eqn:construction of delta one sample}. Similarly, a scatter plot of the set $\Gamma({\cal X}^n,{\cal Y}^n:{\cal U}^n \cap K)$ defined in \eqref{eqn:construction of Gamma}, is called OT potential plot.
 \end{definition}
  Let $L_{\eta}$ be the $\eta$-neighbourhood of the diagonal line $L=\{(u,u):u \in \real\}$, i.e.
  \begin{equation}
      L_{\eta}=\{(x,y) \in \real^2: (x-u)^2+(y-u)^2<\eta^2 , \text{ for some }u \in \real\}.
  \end{equation}
\begin{theorem}
\label{thm: distinguish two sample OT}
Let $\nu_X,\nu_Y$ be probability measures, defined in $\real^d$, with convex support, i.e. $\nu_X,\nu_Y\in \calP^{conv}_d$. Let ${\cal X}^n=\{X_1,\cdots,X_n\}$ and ${\cal Y}^n=\{Y_1, \cdots, Y_n\}$ be two i.i.d. samples drawn from $\nu_X$ and $\nu_Y$, respectively. Let ${\cal U}^n=\{U_1, \cdots,U_n\}$, where $U_i$'s are i.i.d. with common distribution $\mu$.
    Then, $\nu_X=\nu_Y$ if and only if one of the following holds
    \begin{enumerate}[(i)]
        \item 
        $\displaystyle
    \label{probability:component-i-Q-Q}
        \pr\left(\liminf_{n \to \infty}\bigcap_{i=1}^d\left\{\Delta_{i}({\cal X}^n,{\cal Y}^n:{\cal U}^n \cap K)\subseteq L_{\eta}\right\}\right)=1
    $, for all compact $K \subset B^d$ and all $\eta>0$.
    \item $\displaystyle
    \lim_{n \to \infty}\pr^*\left(\left\{\Gamma(\mathcal{X}^n,\mathcal{Y}^{n}:{\cal U}^n \cap K)\subseteq L_{\eta}\right\}\right)=1$, for all compact $K \subset B^d$ and all $\eta>0$, where $\pr^*$ is the outer probability, in the sense of \cite{Vaart1996}.
    \label{probability:potential closeness}
    \end{enumerate}
\end{theorem}
\begin{proof}
    (i) First assume $\nu_X=\nu_Y$. Then the corresponding quantile functions are also equal, $T_{\nu_X}=T_{\nu_Y}$. By the construction of the set $\Delta_{i}({\cal X}^n,{\cal Y}^n:{\cal U}^n \cap K)$ (see \eqref{eqn:construction of delta one sample}), it follows that, for any fixed $\veps>0$,
    \begin{equation}\label{eqn:delta and quantile function inequality}
        \left\{\Delta_{i}({\cal X}^n,{\cal Y}^n:{\cal U}^n \cap K) \subseteq L_{\eta}, \forall i=1,\cdots,d\right\} \supseteq \left\{ \max_{\substack{1 \leq j \leq n \\ {U_j \in K}}} \|T_{{\cal X}_n}(U_j)-T_{{\cal Y}_n}(U_j)\|^2 \leq 2 \eta^2 \right\},
    \end{equation}
    where $T_{{\cal X}_n}$ and $T_{{\cal Y}_n}$ are discrete to discrete optimal transport maps from  ${{\cal U}_n}$ to ${{\cal X}_n}$ and ${{\cal U}_n}$ to ${{\cal Y}_n}$, respectively. To prove the first part of the theorem, it is enough to show that the probability of the event in the right side of \eqref{eqn:delta and quantile function inequality} converges to $1$ as $n$ increases to infinity. Now the sample quantile functions $T_{{\cal X}_n}$, $T_{{\cal Y}_n}$ converge uniformly to $T_{\nu_X}$ and $T_{\nu_Y}$, respectively (by Lemma \ref{lemma:quantile convergence}), and $\displaystyle\lim_{n\to\infty} T_{{\cal X}_n}=\lim_{n\to \infty} T_{{\cal Y}_n}$ by assumption, hence the result.

Conversely, assume that \eqref{probability:component-i-Q-Q} holds for all  $\eta>0$ and for any compact subset $K \subset B^d$. Therefore, almost surely, $\|T_{{\cal X}_n}(U_j)-T_{{\cal Y}_n}(U_j)\|$ converges to $0$ as $n$ goes to infinity. Now, the convergence of empirical OT quantile to its population counterpart (see Lemma \ref{lemma:quantile convergence}) implies that $T_{\nu_X}(U_j)=T_{\nu_Y}(U_j), \forall j \geq 1$. Since the random set $\{U_n: n\geq 1\}$ is almost surely dense in the unit ball $B^d$, and $T_{\nu_X}, T_{\nu_Y}$ are continuous, thus $T_{\nu_X}(u)=T_{\nu_Y}(u),\,\,\, \forall u \in B^d$. Now from Proposition \ref{Prop:unique characterisation of the measure}, we conclude that $\nu_X=\nu_Y$. 
\end{proof}\\\\
(ii) If $\nu_X=\nu_Y$, the proof is similar to the first part of \eqref{probability:component-i-Q-Q}. \\
Conversely, assume that \eqref{probability:potential closeness} holds for all compact subsets $K$ and $\eta>0$. Then we have, 
\begin{equation}
    \label{prob:potential distance}
    \pr^*\left(\max_{\substack{\\ 1 \leq j \leq n\\ \\ U_j \in K}}|\varphi_{\nu_X}(U_j)-\varphi_{\nu_Y}(U_j)|>\eta\right)\underset{n \to \infty}{\longrightarrow} 0.
\end{equation}
Let $A_{\eta}=\{x:|\varphi_{\nu_X}(x)-\varphi_{\nu_Y}(x)|> 2 \eta\}$. If $\mu(A_{\eta})=0,\,\, \forall \eta>0$, then $\varphi_{\nu_X}=\varphi_{\nu_Y}$.

If there exists some $\eta>0$ such that $\mu(A_{\eta})>0$, choose a compact subset $K$ such that $\mu(K \cap A_{\eta})>0$. Now,
\begin{equation}
    \begin{aligned}
    \pr^*\left(\max_{\substack{\\ 1 \leq j \leq n\\ \\ U_j \in K}}|\varphi_{\nu_X}(U_j)-\varphi_{\nu_Y}(U_j)|>\eta\right)
     &> \pr^*\left(\max_{\substack{\\ 1 \leq j \leq n\\ \\ U_j \in K \cap A_{\eta}}}|\varphi_{\nu_X}(U_j)-\varphi_{\nu_Y}(U_j)|>\eta\right)\\
    %&=\pr^*\left(\max_{ 1 \leq j \leq n}|\varphi_{\nu_X}(U_j)-\varphi_{\nu_Y}(U_j)|>\eta, U_j \in K \cap A_{\eta}\,\, \text{ for at least one } j \right)\\
    &=\pr^* \left( U_j \in K \cap A_{\eta}\,\, \text{ for at least one } j \right).\\
    %& \stackrel{n \to \infty}{\longrightarrow} 1
    \end{aligned}
\end{equation}
The last identity follows by the definition of the set $A_{\eta}$. Since $\mu(A_{\eta})>0$, it follows that $\pr^* \left( U_j \in K \cap A_{\eta}\,\, \text{ for at least one } j \right) \underset{n \to \infty}{\longrightarrow} 1$.
This contradicts \eqref{prob:potential distance}, hence $\varphi_{\nu_X}=\varphi_{\nu_Y}$. Now the proof follows from Proposition \ref{Prop:unique characterisation of the measure} $(i)$ .
\hfill \qed
\\\\
Let $\veps>0$ be fixed and $K \subset \real^d$ be such that ${\cal X}_n \cap K$ and ${\cal Y}_n \cap K$ are non empty. Recall that $\calU^n$ is a reference sample drawn from the uniform distribution $\mu$ on $B^d$. Let $T^{\veps}_{{\cal X}_n \cap K}$ and $T^{\veps}_{{\cal Y}_n \cap K}$ be two empirical EOT maps (for the subsamples ${\cal X}_n \cap K$ and ${\cal Y}_n \cap K$, respectively), and $\varphi^{\veps}_{{\cal X}^n \cap K}$ and $\varphi^{\veps}_{{\cal Y}^n \cap K}$ be two EOT potentials, as defined in Section \ref{sec:empirical EOT}. Then, for $i=1,\cdots,n$, define the sets,
\begin{equation}
    \Delta^{\veps}_i({\cal X}^n \cap K,{\cal Y}^n \cap K:{\cal U}^n)=\left\{\left( \left\langle T^{\veps}_{{\cal X}_n \cap K}(U_j),e_i\right\rangle, \langle T^{\veps}_{{\cal Y}_n \cap K}(U_j),e_i\rangle\right): j=1, \cdots,n \right\}
\end{equation}
and
\begin{equation}
    \Gamma^{\veps}({\cal X}^n \cap K,{\cal Y}^n \cap K:{\cal U}^n)=\left\{ (\varphi^{\veps}_{{\cal X}_n \cap K}(U_j),\varphi^{\veps}_{{\cal Y}_n \cap K}(U_j)): \,\,j=1,\cdots,n \right\}.
\end{equation}

\begin{definition}
    \label{dfn:EOT plots}
     Let $K\subset \real^d$ be a compact set. An EOT Q-Q plot is a collection of $d$ individual scatter plots, where the $i$--th one displays the set $\Delta^{\veps}_i({\cal X}^n \cap K,{\cal Y}^n \cap K:{\cal U}^n)$. Similarly, a scatter plot of the set $\Gamma^{\veps}({\cal X}^n \cap K,{\cal Y}^n \cap K:{\cal U}^n)$ is called EOT potential plot.
 \end{definition}

The next theorem proves similar results as in Theorem \ref{thm: distinguish two sample OT} but for EOT.
\begin{theorem}
    \label{thm: distinguish two sample EOT}
    Let $\nu_X$ and $\nu_Y$ be two probability measures with finite second moment. Let ${\cal X}^n=\{X_1,\cdots,X_n\}$ and ${\cal Y}_n=\{Y_1, \cdots, Y_n\}$ be two i.i.d. samples drawn from $\nu_X$ and $\nu_Y$, respectively. Let ${\cal U}^n=\{U_1, \cdots,U_n\}$, where $U_i$'s are i.i.d. with common distribution $\mu$. Then $\nu_X=\nu_Y$ if and only if one of the following holds.
    \begin{enumerate}[(i)]
        \item  
    $\displaystyle
    \label{probability:EOT potential closeness}
        \pr\left(\liminf_{n \to \infty}\left\{\Gamma^{\veps}(\mathcal{X}^n \cap K,\mathcal{Y}^{n} \cap K :{\cal U}^n)\subseteq L_{\eta}\right\}\right)=1
    $, for all compact $K \subset B^d$ and all $\eta>0$.
    \item $\displaystyle
    \label{probabilityEOT map closeness}
        \pr\left(\liminf_{n \to \infty}\bigcap_{i=1}^d\left\{\Delta^{\veps}_i(\mathcal{X}^n \cap K,\mathcal{Y}^{n} \cap K:{\cal U}^n)\subseteq L_{\eta}\right\}\right)=1
    $, for all compact $K \subset B^d$ and all $\eta>0$.
    
    \end{enumerate}
\end{theorem}
\begin{proof}
    $(i) $ First assume that $\nu_X=\nu_Y$. Let $\nu_X|_K$ be a probability measure obtained by restricting $\nu_X$ on the compact set $K$. Similarly, $\nu_Y|_K$ be the restriction $\nu_Y$ on $K$. Then the potentials for the restricted measures are equal, i.e. $\varphi^{\eps}_{\nu_X|_K}=\varphi^{\eps}_{\nu_Y|_K}$. 
    Since both (random) empirical measures $\nu_{{\cal X}_n \cap K}$ and $\nu_{{\cal Y}_n \cap K}$ converge (weakly) to the limit $\nu_X|_K ( =\nu_Y|_K)$
     almost surely, therefore Lemma $1(iii)$ of \cite{goldfeld2023limit} implies that $\|\varphi^{\eps}_{{{\cal X}_n} \cap K}-\varphi^{\eps}_{{{\cal Y}_n}\cap K}\|_{{\cal C}^s}$ converges to $0$ almost surely, for any integer $s \geq 1$. Here, $\|\cdot\|_{\calC^s}$
      defined as the usual H\"older norm with exponent $s$. The proof now follows through a similar argument as employed in Theorem \ref{thm: distinguish two sample OT}. \\\\
    $(ii)$  The proof is similar as for the first part, after observing that the map $\varphi \to \nabla \varphi$ is continuous and linear from ${\cal C}^{s}$ to ${\cal C}^{s-1}$.
    \hfill\qed
\end{proof}

\vspace{1cm}
\begin{comment}
\begin{remark}
    { \color{red} hange the remark}
    Theorem~\ref{thm: distinguish two sample OT} and \ref{thm: distinguish two sample EOT} are crucial for experimenting with (E)OT Q-Q and potential plots. However, note that Theorem~\ref{thm: distinguish two sample EOT} is only applicable for compactly supported distributions. In case of unbounded support or if the support is not known then one can proceed with restricting the distribution on a big compact set. In otherwords, for a compact set $K$ the conclusions of Theorem~\ref{thm: distinguish two sample EOT} holds for the sets $\Gamma^{\veps}(\mathcal{X}^n\cap K,\mathcal{Y}^{n}\cap K:{\cal U}^n)$ and $\Delta^{\veps}_i(\mathcal{X}^n\cap K,\mathcal{Y}^{n}\cap K:{\cal U}^n)$. Notice that the cardinality of the sets $\mathcal{X}^n\cap K$, $\mathcal{Y}^n\cap K$, and $\calU^n$ could be different. However, this poses no problem, as the empirical EOT map and potential can be defined between two samples of different sizes (see Remark~\ref{rk:empirical EOT}).  In practice, we will consider a large enough compact set $K$ so that it contais all the samples $\calX^n$ and $\calY^n$, see Section~\ref{sec: Illustration} for more details.
\end{remark}
\end{comment}
\section{Test Statistics for comparing two  distributions} \label{sec: Test statistics}

As in \cite{Shapiro1965, Dhar2014}, we propose two test statistics which are motivated from our Q-Q plots. Consider the  i.i.d random variables $\calX^n= \{X_1, \cdots ,X_n\}$ and $\calY^n= \{Y_1, \cdots ,Y_n\}$ with common distributions ${\nu_X}$ and ${\nu_Y}$, respectively, and assume the corresponding EOT maps $T^{\veps}_{\calX^n},T^{\veps}_{\calY^n}$ as defined in Section \ref{sec:empirical EOT}. We can measure the total deviation of the EOT Q-Q plots from the straight line $L$ by the quantity $\displaystyle \int_{B^d} \|T^{\veps}_{\calX^n}(u) - T^{\veps}_{\calY^n}(u)\|^2 d\mu(u)$. Similarly, the total deviation in the EOT potential plot can be measured by the quantity $\displaystyle\int_{B^d} \|\varphi^{\veps}_{\calX^n}(u) - \varphi^{\veps}_{\calY^n}(u)\|^2 d\mu(u)$.
 Therefore, we define the following two test statistics
 \begin{equation}
    \label{dfn:teststat}
    E_n=n\int_{B^d} \|T^{\veps}_{\calX^n}(u) - T^{\veps}_{\calY^n}(u)\|^2 d\mu(u) \quad\text{and}\quad 
    F_n=n\int_{B^d} |\varphi^{\veps}_{\calX^n}(u) - \varphi^{\veps}_{\calY^n}(u)|^2 d\mu(u).
 \end{equation}
Since the supports of $\nu_X$ and $\nu_Y$ are compact, it follows from  \cite[Lemma $1$]{goldfeld2023limit} that for any $s \in \mathbb{N}$,  $\|\varphi^{\veps}_{\calX^n}\|_{\calC^s(B^d)} \lor \|\varphi^{\veps}_{\calY^n}\|_{\calC^s(B^d)} \leq R_s,\,\, \forall n \geq 1$, where $R_s$ is a constant depending on $\veps$ and the support of the measures $\nu_X$, $\nu_Y$. In other words, the potentials and their derivatives are uniformly bounded on $B^d$. This implies that $E_n$ and $F_n$ are almost surely finite for all $n\geq1$.

\begin{proposition}\label{prop:cv-test-stat} Let $\nu_X$ and $\nu_Y$ be two probability measures supported on compact subsets of $\real^d$. Let $\nu_X=\nu_Y$, and $\calX^n, \calY^n$ be two independent sets of samples drawn from $\nu_X$ and $\nu_Y$, respectively. Consider $F_n$ and $E_n$ as in \eqref{dfn:teststat}, then we have $$F_n \underset{n\to\infty}{\overset{d}{\longrightarrow}} F \qquad \text{and} \qquad E_n \underset{n\to\infty}{\overset{d}{\longrightarrow}} E,$$
where $E$ and $F$ are non-negative random variables with 
%%%%%%%%%%%%%%%% TO LOOK AFTER LATER%%%%%%%%%%%{\color{red}positive} 
 finite variance.
\end{proposition} 
\begin{proof}
    The samples $\calX^n$ and $\calY^n$ are independent of each other, hence $\varphi^{\veps}_{\calX^n}-\varphi^{\veps}_{\nu_X}$ and $\varphi^{\veps}_{\nu_Y}-\varphi^{\veps}_{\calY^n}$ are also independent. Therefore, it follows from \cite[Theorem $1$]{goldfeld2023limit} that, 
\begin{equation}
    \sqrt{n}\left(\varphi^{\veps}_{\calX^n} - \varphi^{\veps}_{\nu_X}, \varphi^{\veps}_{\calY^n}-\varphi^{\veps}_{\nu_Y} \right) \underset{n\to\infty}{\overset{d}{\longrightarrow}}(Z_X, Z_Y),
\end{equation}
where $Z_X$ and $Z_Y$ are i.i.d mean zero Gaussian random elements, taking values in the space $\calC^s(B^d)$. Observe that the map $f \mapsto \int_{B^d} |f(u)|^2du$ is continuous from $\calC^s(B^d)$ to $\real$, for any $s \in \mathbb{N}$. Hence, $F_n \underset{n\to\infty}{\overset{d}{\longrightarrow}} F:=\int_{B^d} \|(Z_X+Z_Y)(u)\|^2d\mu(u)$.  Therefore, by Fernique's theorem, we have $var(F)<\infty$. 

%The random variable $Z_X + Z_Y$ is again zero mean Gaussian random variable on $\calC^s(B^d)$,

%%%%%%%%%%%%%%%%%%%% TO LOOK AT LATER %%%%%%%%%%%%%%%%%%%%%%%%%%%
% Since the random variables $Z_X$ and $Z_Y$ are non degenerate, it follows that {\color{red}$var(F)>0$}.
%

The result $E_n \underset{n\to\infty}{\overset{d}{\longrightarrow}} E$ follows by an application of the functional delta method, as proven in \cite[Corollary $1$]{goldfeld2023limit}. A similar line of argument holds for $F_n$. \hfill\qed
\end{proof}

\vspace{1cm}
\noindent We can then deduce the following theorem from Proposition~\ref{prop:cv-test-stat}.
\begin{theorem}
\label{thm:p-values}
Assume that $\nu_X$ and $\nu_Y$ are two probability measures with compact support. Consider the null hypothesis $H_0:\nu_X=\nu_Y$ and the alternate hypothesis $H_1:\nu_X \neq \nu_Y $. Let $E_n, F_n$ and their limits $E, F$ be defined as above, with $E$ and $F$ non-degenerate. Then the following holds: 
\begin{enumerate}[(i)]
    \item For $0<\alpha<1$ and $q_{\alpha}$ the $\alpha $-th quantile of $E$, we have,
    
$\displaystyle \lim_{n \to \infty}\pr(E_n>q_{\alpha}\,\,|H_0)=1-\alpha$ \;and \;$\displaystyle \lim_{n \to \infty}\pr(E_n<q_{\alpha}\,\,|H_1)=0.$
\item For $0<\alpha<1$ and $q_{\alpha}$ the $\alpha$-th quantile of $F$, we have,

$\displaystyle\lim_{n \to \infty}\pr(F_{n}>q_{\alpha}\,\,|H_0)=1-\alpha$ \; and \; $\displaystyle\lim_{n \to \infty}\pr(F_{n}<q_{\alpha}\,\,|H_1)=0.$
\end{enumerate}
\end{theorem}
\begin{proof}
\begin{enumerate}[(i)]
    \item 
    Corollary $1$ in \cite{goldfeld2023limit} implies that, under the  null hypothesis, $E_n$ converges in distribution to $E$. Therefore, we have $\displaystyle\lim_{n \to \infty}\pr(E_n>q_{\alpha}|H_0)=\pr(E>q_{\alpha}|H_0)=1-\alpha$.
    
    Moving to the second part of the statement, first note that
    \begin{eqnarray*}
        \frac1n E_n&=& \int_{B^d} \|T^{\veps}_{\calX^n}(u)-T^{\veps}_{{\nu}_X}(u)+T^{\veps}_{{\nu}_X}(u) - T^{\veps}_{\nu_Y}(u)+ T^{\veps}_{\nu_Y}(u)-T^{\veps}_{\calY^n}(u)\|^2 d\mu(u)\\
        &\ge& \int_{B^d}\Big[\|T^{\veps}_{\nu_X}(u) - T^{\veps}_{\nu_Y}(u)\|^2 + 2 \langle T^{\veps}_{\calX^n}(u)-T^{\veps}_{{\nu}_X}(u), T^{\veps}_{{\nu}_X}(u) - T^{\veps}_{\nu_Y}(u) \rangle\\
        && +  2 \langle T^{\veps}_{{\nu}_X}(u) - T^{\veps}_{\nu_Y}(u), T^{\veps}_{{\nu}_Y}(u) - T^{\veps}_{\calY^n}(u) \rangle
         +  2 \langle T^{\veps}_{\calX^n}(u)-T^{\veps}_{{\nu}_X}(u), T^{\veps}_{{\nu}_Y}(u) - T^{\veps}_{\calY^n}(u) \rangle\Big] \,d\mu(u).
    \end{eqnarray*}
    Therefore,
    \begin{eqnarray*}
        \{E_n>q_{\alpha}\} & \supseteq & \Bigg\{ \int_{B^d} \Big[\sqrt{n}\|T^{\veps}_{\nu_X}(u) - T^{\veps}_{\nu_Y}(u)\|^2 + 2 \sqrt{n}\langle T^{\veps}_{\calX^n}(u)-T^{\veps}_{{\nu}_X}(u), T^{\veps}_{{\nu}_X}(u) - T^{\veps}_{\nu_Y}(u) \rangle \\
        && + 2 \sqrt{n}\langle T^{\veps}_{{\nu}_X}(u) - T^{\veps}_{\nu_Y}(u), T^{\veps}_{{\nu}_Y}(u) - T^{\veps}_{\calY^n}(u) \rangle \\
        &&+ 2 \sqrt{n}\langle T^{\veps}_{\calX^n}(u)-T^{\veps}_{{\nu}_X}(u), T^{\veps}_{{\nu}_Y}(u) - T^{\veps}_{\calY^n}(u) \rangle\Big]d\mu(u) > \frac{q_{\alpha}}{\sqrt{n}}\Bigg\}.
    \end{eqnarray*}
    Now, notice that, under $H_1$ ,we have
    %Since $T^{\veps}_{{\nu}_X} \neq T^{\veps}_{{\nu}_Y}$, therefore
     $\displaystyle \sqrt{n}\int \|T^{\veps}_{{\nu}_X}(u) - T^{\veps}_{{\nu}_Y}(u)\|^2 d\mu(u) \underset{n\to\infty}{\to} +\infty$. Further, it follows from Corollary $1$ in \cite{goldfeld2023limit} that: 
     \begin{itemize}
        \item $\sqrt{n} \int \left\langle T^{\veps}_{\calX^n}(u)-T^{\veps}_{{\nu_X}}(u), T^{\veps}_{{\nu}_X}(u) - T^{\veps}_{{\nu}_Y}(u)\right\rangle d\mu(u)$ and \\ $\sqrt{n} \int \left\langle T^{\veps}_{\calY^n}(u)-T^{\veps}_{{\nu_Y}}(u), T^{\veps}_{{\nu}_X}(u) - T^{\veps}_{{\nu}_Y}(u)\right\rangle d\mu(u)$ converge in distribution to some random variables with finite variance;
        \item $\sqrt{n} \int \left\langle T^{\veps}_{\calX^n}(u)-T^{\veps}_{{\nu_X}}, T^{\veps}_{\calY^n}(u)-T^{\veps}_{{\nu_Y}}(u)\right\rangle d\mu(u)$ converges to zero almost surely.
     \end{itemize}
    Combining all these observations, we conclude that $\displaystyle \pr\left(E_n>q_{\alpha}\,\,|H_1\right) \underset{n\to\infty}{\to}1$, from which we deduce the second statement.
    \item The proof of this part follows similarly as in part (i).
    \hfill \qed 
\end{enumerate}
\end{proof}

\begin{remark}~
    \begin{enumerate}
        \item  For the OT maps and OT potentials, the asymptotic limit theorems are not known in full generality. Whereas, in case of EOT, the limit theorems are proven in \cite{goldfeld2023limit}. This is why we provided the test statistics only through EOT maps and EOT potentials. Nevertheless, once limit theorems known in the OT case, test statistics can be defined for OT maps and potentials in the same way as we did for EOT.
        \item Note that the variances of the limiting random variables $E$ and $F$ are finite, 
        %%%%%%%%%%%%%%%%%%% TO LOOK AT LATER %%%%%%%%%%%%%%%%%%%%%%%%%%%
        % and {\color{red}nonzero}, 
        %%%%
        but not known explicitly. This is why we evaluate them numerically in the examples given in Section \ref{sec: Illustration}.
    \end{enumerate}

\end{remark}

\section{Experimental design on simulated data}\label{sec: Illustration}

In the preceding sections, we defined Q-Q plots and potential plots, for both OT and EOT approaches, and established the requisite theory to assess if two given sets of multivariate samples originate from a same distribution. In this section, we test how these theoretical results apply in practice, and check the effectiveness of the visual tools offered by OT and EOT plots on various simulated data. To this aim, we perform four first experiments, considering two samples drawn from distributions such that: (I) they are identical; (II) they have different dependence structures; (III) they are related by a scaling map; (IV) one of them exhibits outliers. 
Thereafter, we consider the same question, but now focusing on the tail of the distribution. Note that the potential function which gradient pushes forward a regularly varying (RV) probability measure into another regularly varying (RV) probability measure, is also (under some conditions) regularly varying; see \cite[Theorem 5.1]{Valk2018}.  Can we clearly distinguish between light and heavy tails with these (OT, EOT) Q-Q and potential plots, like in the case of geometric quantiles (or univariate quantiles)? To answer this question, we proceed to two additional experiments, where we compare multivariate Gaussian distribution with two different heavy-tailed distributions. All scenarios in the following examples are developed considering i.i.d and non i.i.d distributions.

More formally, using our previous notation, for each experiment, we proceed as follows: 
\begin{enumerate}[(i)]
    \item Consider two samples, $\mathcal{X}^n$ and $\mathcal{Y}^n$ of size $n$ each and a compact set $K_1$. In all the experiments we choose $K_1$ such that $\calX^n \cup \calY^n \subset K_1$; 
    \item Draw another sample $\calU^n$ from the uniform distribution on $B^d$, which serves as a reference for comparison. Let $K_2\subset B^d$ be a compact set. In all the experiments we choose $K_2$ such that $\calU^n \subset K_2$;
    \item  Q-Q plots: build the sets $\Delta_{i}({\cal X}^n,{\cal Y}^n:{\cal U}^n \cap K_2)$, $i=1, \cdots, d$, for the OT Q-Q plot and $\Delta^{\veps}_i(\mathcal{X}^n\cap K_1,\mathcal{Y}^{n}\cap K_1:{\cal U}^n)$ for the EOT Q-Q plots;
    \item  Potential plots: build  $\Gamma({\cal X}^n,{\cal Y}^n:{\cal U}^n \cap K_2)$ for the OT potential plot and $\Gamma^{\veps}(\mathcal{X}^n\cap K_1,\mathcal{Y}^{n}\cap K_1:{\cal U}^n)$ for the EOT potential plot;
    \item  Examine if the Q-Q and potential plots concentrate around the straight line denoted by $L$ with slope $1$ and intercept $0$;
    \item  If they do, infer from Theorems~\ref{thm: distinguish two sample OT} and ~\ref{thm: distinguish two sample EOT} that the samples share the same distribution. Otherwise, look for any discernible pattern present in the scatter plots and whether this pattern may unveil any distinct features. For instance, observe if the points appear to cluster around any specific nonlinear curve (see Figure~\ref{fig:normal vs mul_t});
    \item  Estimate the values of the test statistics \(E_n\) and \(F_n\), and estimate the corresponding p-values, to statistically assess the conclusions drawn from the plots. Note that the p-values are meaningful only when comparing distributions supported on compact sets, as Theorem~\ref{thm:p-values} is applicable to compactly supported distributions. However, when dealing with fully supported distributions, a partial analysis can be conducted by restricting them to a large enough compact set. Two Borel probability measures are considered identical if and only if their restrictions across all compact sets are identical.
    \item Finally we also study numerically the role of the regularisation parameter $\veps$ for EOT.
\end{enumerate}
%Experiment (i) is developed in Example~\ref{ssec:ex1}, choosing a standard three-dimensional Gaussian distribution. Then, 

\vspace{.5cm}
In the following examples, we present two kinds of plots: Q-Q plot and potential plot. In each Q-Q plot example, we display the OT and EOT Q-Q plots together in a single frame: the first row for OT and the second for EOT. In each plot, the $X$-axis corresponds to the first sample $\calX^n$ and the $Y$-axis to the second sample $\calY^n$. Also, the potential plots are presented together in a single frame: the left for OT and the right for EOT. Please note that all computations are performed using the POT (Python Optimal Transport) package; for detailed documentation, we refer to \cite{flamary2021pot}.

\subsection{Example: Comparison between two samples drawn from a multivariate (non i.i.d.) Gaussian distribution}\label{ssec:ex1}

Starting with step (i), we consider two sets of i.i.d. samples, $\calX^n$ and $\calY^n$,  drawn from distributions $\nu_X$ and $\nu_Y$, respectively, where $\nu_X = \nu_Y$ is a trivariate normal distribution with mean zero and covariance matrix $(\sigma_{ij})$, where, for instance, $\sigma_{11}=\sigma_{22}=\sigma_{33}=1$,  $\sigma_{12}=0.5$, $\sigma_{13}=0.2$ and $\sigma_{23}=0$. The number of observations in each sample is $1000$. 
Following steps (ii) to (iii), we obtain the OT and EOT Q-Q plots for these two samples, as displayed in Figure~\ref{fig:comp same Gaussian}. 
\begin{figure}[H]
    \centering
    \includegraphics[width=.8\linewidth]{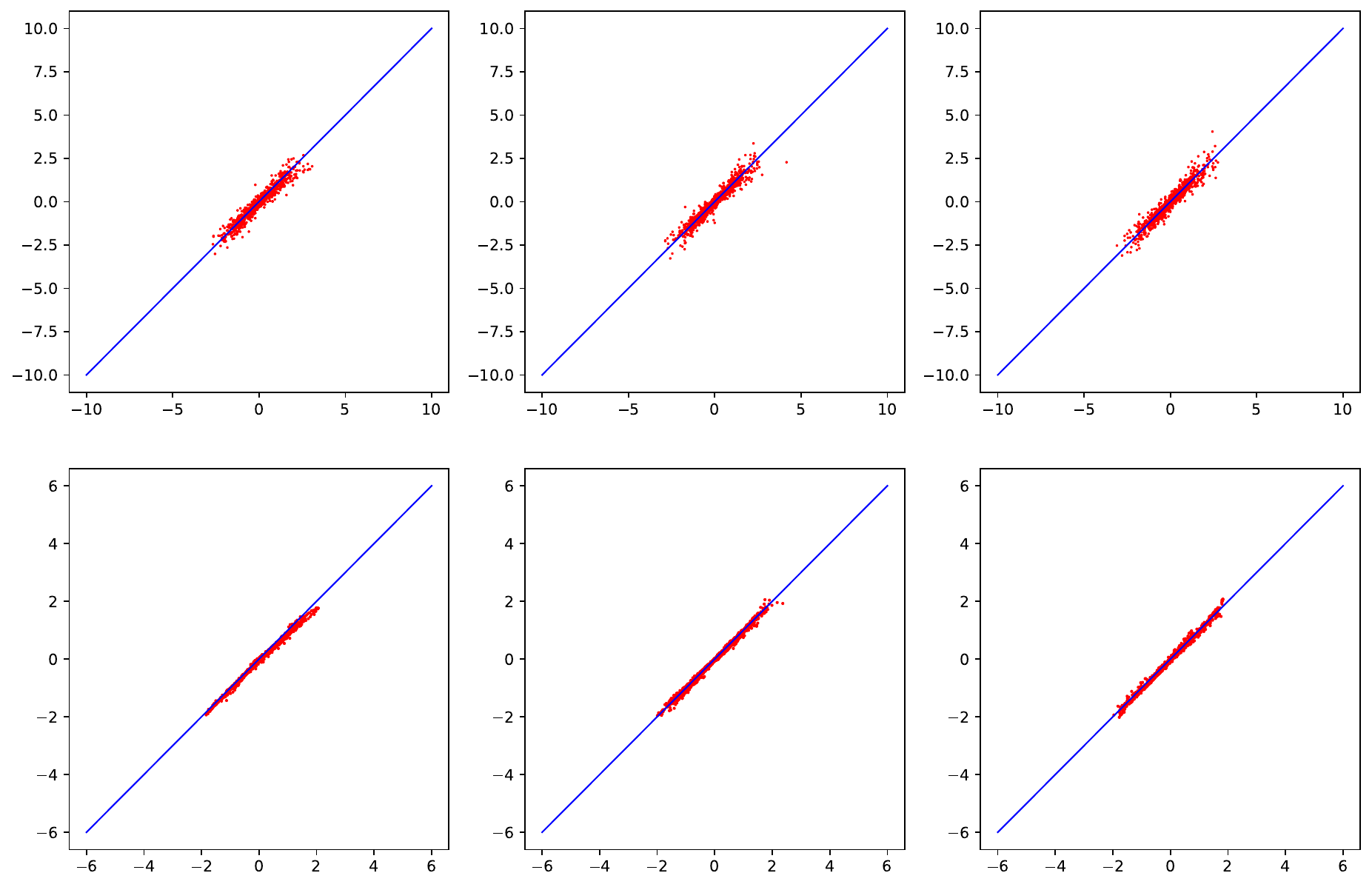}
    \caption{\small \sf Q-Q plots for two samples, drawn from a trivariate normal distribution with mean zero and covariance matrix $(\sigma_{ij})$, where $\sigma_{11}=\sigma_{22}=\sigma_{33}=1$,  $\sigma_{12}=0.5$, $\sigma_{13}=0.2$, $\sigma_{23}=0$. First row: OT Q-Q plot; second row: EOT Q-Q plot, with $\veps=10^{-2}$. The straight (blue) lines represent the line $L$. 
    }
    \label{fig:comp same Gaussian}
\end{figure}
Next, we provide the potential plots for this example following the step (iv); see Figure~\ref{fig:sG_vs_sG_potential}. Note that for EOT, we choose $\veps=10^{-2}$. 
\begin{figure}[H]
    \centering
    \includegraphics[width=.8\linewidth]{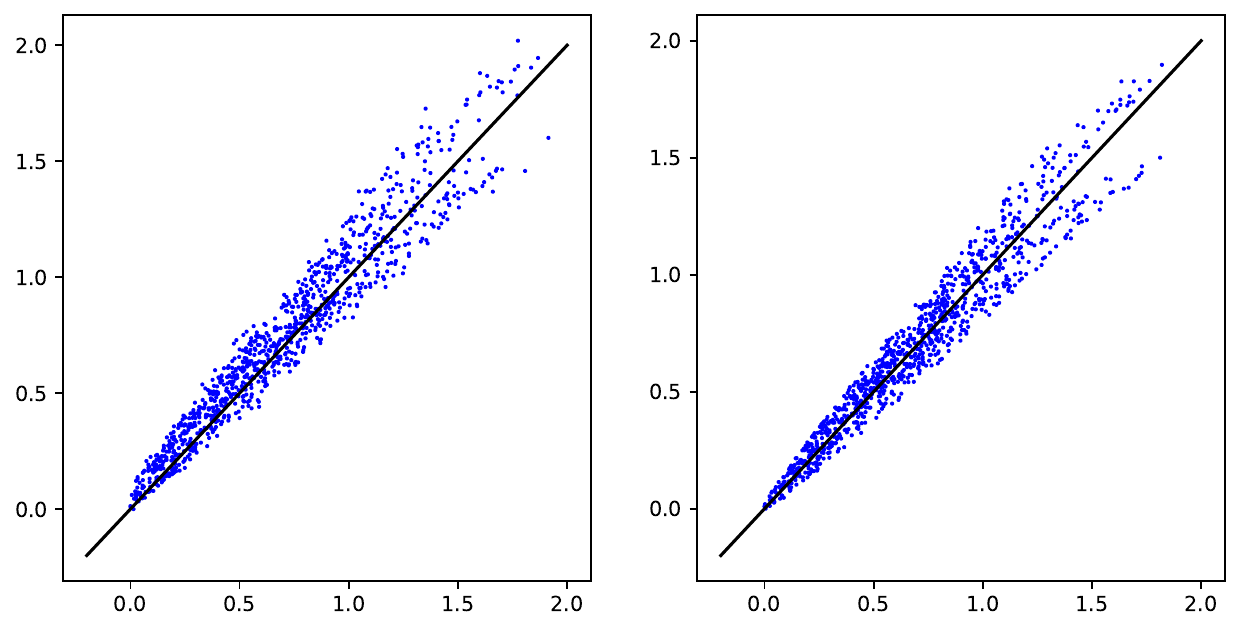}
    \caption{\small \sf Potential plots for two samples,  drawn from a trivariate normal distribution with mean zero and covariance matrix $(\sigma_{ij})$, where $\sigma_{11}=\sigma_{22}=\sigma_{33}=1$,  $\sigma_{12}=0.5$, $\sigma_{13}=0.2$, $\sigma_{23}=0$. The left plot compares the OT potentials, the right one EOT potentials taking $\veps=10^{-2}$. The straight (black) lines represent the line $L$.}
    \label{fig:sG_vs_sG_potential}
\end{figure}
As expected, we observe that, for both types of comparisons, Q-Q and potential ones, the scatter plot is concentrated along the line $L$. Therefore, according to step (vi), we can infer that the two samples are drawn from the same distribution. Although the spread around the straight line $L$ appears larger in the case of potential plots than in Q-Q plots, notice that the two sets of plots are shown on different scales, and that the dispersion in absolute terms is comparable in both cases. 

%{\color{red}Another stark feature observed clearly in (all) the potential plot(s) is the increased concentration of points in the vicinity of $0$. This may be attributed to the normalisation of potential functions by which all potential functions are pinned to zero at the origin, and, as continuous functions, will stay close to zero in a small neighbourhod around the origin.}
%, and hence the sharper concentration of potential plots around the line $L$, near $0$. 
A relatively larger spread is observed in the extremal regions of the potential plots, which may be ascribed to poorer approximation of the potential function in the extreme region when having fewer observations.

Compared to the OT plots, we also observe sharper clustering around the straight line in case of EOT. A possible reason for this could be that EOT maps and potentials are more regular than their OT counterparts (\cite{Nutznote2022}). 
%{\color{red}[check about the ref]}

Next, we perform statistical test checking the similarity of the underlying distributions of the two samples ${\cal X}^n$ and ${\cal Y}^n$.
%Now, assume that the samples ${\cal X}^n$ and ${\cal Y}^n$ are drawn from (unknown) distributions $\mu_X$ and $\mu_Y$, respectively. 
We set the null hypothesis as $H_0: \nu_X|_{K_2}=\nu_Y|_{K_2}$, and empirically estimate the values of the test statistics $E_n$ and $F_n$ defined in \eqref{dfn:teststat}. Additionally, the limiting distributions of $E_n$ and $F_n$ are empirically estimated, from which we deduce the $p$-values under $H_0$. We will proceed in the same way in all the illustrations. The values of the test statistics and $p$-values are reported in Table~\ref{table:pvalue_two_Gaussian}, varying the sample size $n$. Since the $p$-values are relatively high, whatever the sample size, it supports the assertion that the two samples have originated from the same distribution. 
\vspace{.3cm}
\begin{table}[h!]
    \begin{tabularx}{1\textwidth} { 
        | >{\centering\arraybackslash}X 
        | >{\centering\arraybackslash}X 
        | >{\centering\arraybackslash}X 
        | >{\centering\arraybackslash}X 
        | >{\centering\arraybackslash}X|}
       \hline
        $n$& $E_n$ & $p$-value of $E_n$ & $F_n$ & $p$-value of $F_n$\\
       \hline
       $250$  & $45.19$  & $0.1137$ & $ 14.19$ & $0.1559$ \\
          \hline
           $500$ & $32.17$ & $0.6578$ & $6.84$ & $0.6723$\\
          \hline
          $1000$  & $33.27$  & $0.5941$ & $ 6.84$ & $0.6723$ \\
      \hline
      \end{tabularx}  
      \caption{\small\sf Test statistics and $p$-values under the null hypothesis $H_0: \nu_X|_{K_2}=\nu_Y|_{K_2}$.}
  \label{table:pvalue_two_Gaussian}
\end{table}

\subsection{Example: Characterising dependency}

In multivariate setting, assessing dependency among marginals is an essential feature of multivariate statistical analysis. In this experiment, we compare two samples, one drawn from a distribution with independent marginals against one drawn from a distribution having dependent marginals, to see if there is any specific feature that can be observed from the Q-Q and potential plots. We consider two samples; the first one is drawn from ($\nu_X$), a trivariate standard normal distribution, and the second from ($\nu_Y$), a normal distribution with zero mean and covariance matrix $(\sigma_{ij})$, where $\sigma_{11}=\sigma_{22}=\sigma_{33}=1$,  $\sigma_{12}=0.9$, $\sigma_{13}=0$ and $\sigma_{23}=0$. The Q-Q and potential plots are displayed in Figures~\ref{fig:sG_vs_corrG_QQ} and \ref{fig:sG_vs_corrG_potential}, respectively.

In both Q-Q plots (OT and EOT), we observe that, although the points are concentrated along the line $L$, the concentration is stronger for the third component (especially for EOT), as compared to the plots for the first and second component.
Since the second distribution exhibits a high correlation ($0.9$) between the first and second marginals, would this pattern indicate some dependence? Answering this question will require further research and experiments, considering various dependence structures.
\begin{figure}[H]
    \centering
    \includegraphics[width=.8\linewidth]{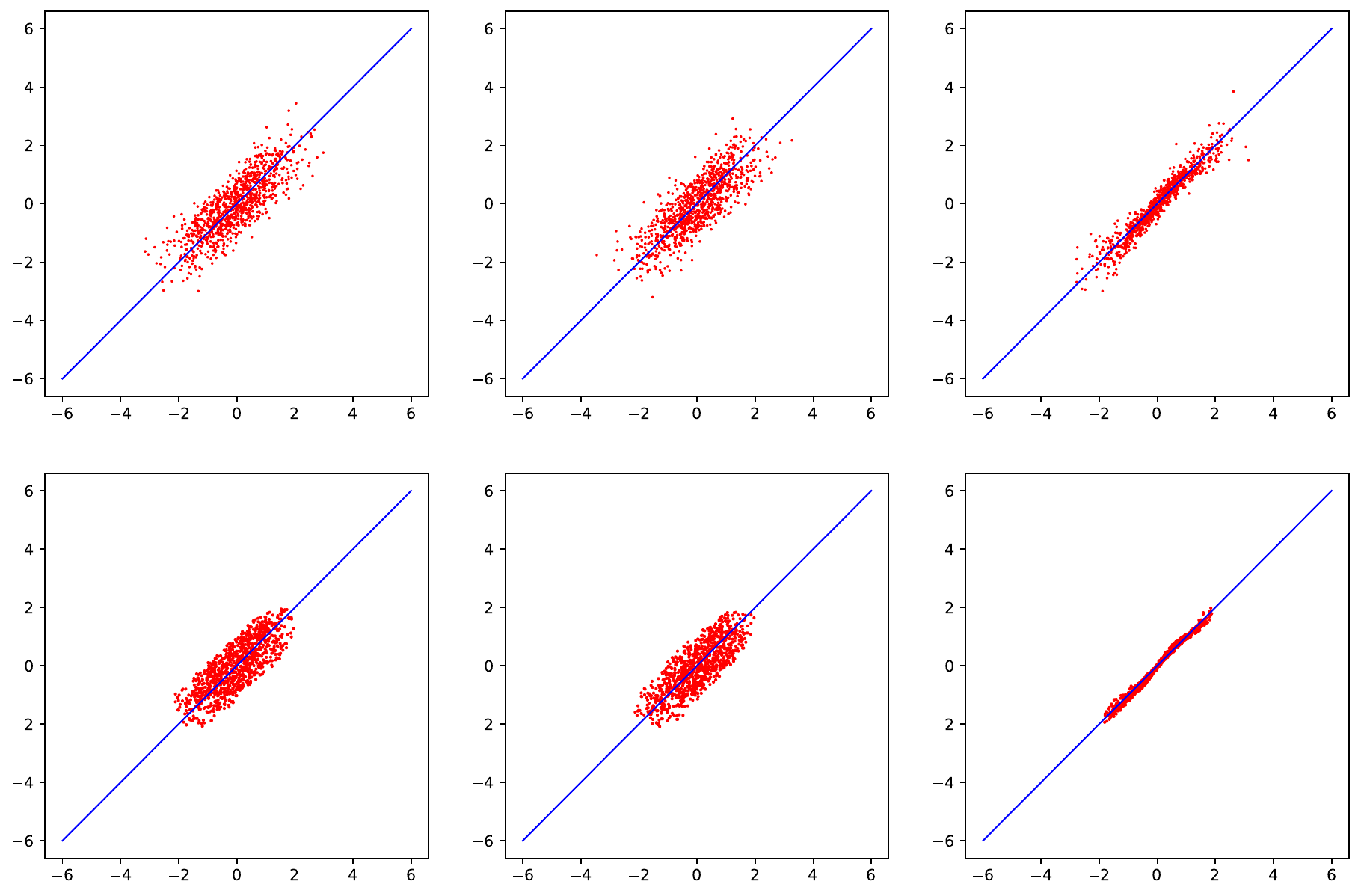}
    \caption{\small\sf Q-Q plots for two samples, each of size $1000$; the first sample is drawn from a trivariate standard normal distribution, the second from a normal distribution with zero mean and covariance matrix $(\sigma_{ij})$, where $\sigma_{11}=\sigma_{22}=\sigma_{33}=1$,  $\sigma_{12}=0.9$, $\sigma_{13}=0$, $\sigma_{23}=0$. The first row displays OT Q-Q plots, the second EOT Q-Q plots with $\veps=10^{-2}$.  The straight (blue) lines represent the line $L$.} 
    \label{fig:sG_vs_corrG_QQ}
\end{figure}
\begin{figure}[H]
    \centering
    \includegraphics[width=.8\linewidth]{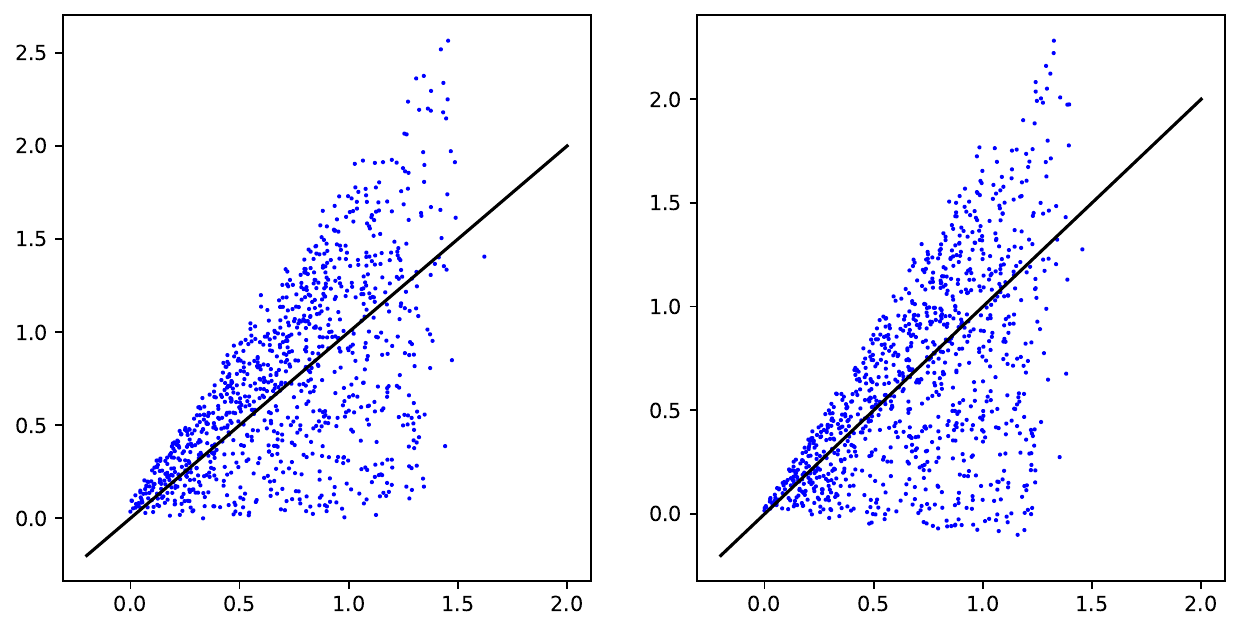}
    \caption{\small\sf Potential plots for 2 multivariate samples, each of size $1000$. The first sample is drawn from a trivariate standard normal distribution, the second from a normal distribution with mean 0 and covariance matrix $(\sigma_{ij})$, \!where $\sigma_{11}=\sigma_{22}=\sigma_{33}=1$,  $\sigma_{12}=0.9$, $\sigma_{13}=0$, $\sigma_{23}=0$. Left: OT potential plot; Right: EOT potential plot with $\veps=10^{-2}$.\!\! The straight (black) lines represent the line\! $L$.}
    \label{fig:sG_vs_corrG_potential}
\end{figure}
Turning to potential plots displayed in Figure~\ref{fig:sG_vs_corrG_potential}, we can assess that the samples considered for the experiment have originated from different distributions. However, we cannot deduce any specific pattern in the plots that would suggest a high dependence among some components. Here too, further research needs to be developed.
\\[1ex]
Finally, we empirically estimate the values of the test statistics $E_n$ and $F_n$ defined in \eqref{dfn:teststat}, as well as the corresponding $p$-values under $H_0: \nu_X|_{K_2}=\nu_Y|_{K_2}$. The values are reported in Table~\ref{table:pvalue_two_Gaussian-correlated}. The small $p$-values lead to the statement that the two samples have originated from two different distributions.
\vspace{.3cm}
\begin{table}[h!]
\begin{tabularx}{1\textwidth} { 
    | >{\centering\arraybackslash}X 
    | >{\centering\arraybackslash}X 
    | >{\centering\arraybackslash}X 
    | >{\centering\arraybackslash}X 
    | >{\centering\arraybackslash}X|}
   \hline
    $n$& $E_n$ & $p$-value of $E_n$ & $F_n$ & $p$-value of $F_n$\\
   \hline
   $250$  & $277.86$  & $0$ & $ 62.46$ & $0$ \\
      \hline
       $500$ & $469.56$ & $0$ & $166.46$ & $0$\\
      \hline
      $1000$  & $877.14$  & $0$ & $304.53$ & $0$ \\
  \hline
  \end{tabularx}
  \caption{\small\sf Test statistics and $p$-values under the null hypothesis $H_0: \nu_X|_{K_2}=\nu_Y|_{K_2}$.}
  \label{table:pvalue_two_Gaussian-correlated}
\end{table}

 Observe that the values of $E_n$ and $F_n$ are getting larger as the sample size $n$ increases. This may indicate that the two underlying distributions are not the same, since  we know that $E_n$ and $F_n$ tend to infinity in such a case. Also note that the $p$-values are exactly zero, which looks too good to be true. Although the limiting distributions of $E_n$ and $F_n$ are fully supported on the positive real line, our approximations of those are supported on bounded sets. Therefore, the large values of $E_n$ and $F_n$, which correspond to small probability regions of the limiting distributions, easily fall beyond the support of the approximated limiting distributions. To obtain non-zero $p$-values, either the approximations have to be more accurate so that their supports are large enough to contain $E_n$ and $F_n$, or the limiting distributions have to be known precisely.
 %{\color{red} Comment on the fact that p-values being exactly $0$}

\subsection{Example: Gaussian versus scaled Gaussian}\label{ssec:ex2}

We draw samples ${\cal X}^n$ from $\mu_X$, a standard normal distribution and ${\cal Y}^n$ from $\mu_Y$, a centered normal distribution with covariance matrix $\diag (1,4,1)$. The Q-Q plots are displayed in Figure \ref{fig:Gauss vs Scaled Gauss}. We observe that in the first and third columns in both rows, the points are concentrated around the line with slope $1$ (blue line) whereas in the second column, the points are concentrated around the line with slope $2$ (black line). This observation suggests that $T_{{\cal Y}n} = \text{diag}(1, 2, 1) T_{{\cal X}_n}$, indicating that the second distribution differs from the first by a scaling factor represented by the matrix $\text{diag}(1, 2, 1)$. Since $\eps$ is set very small $(10^{-2})$, the entropy regularised quantiles $T^{\eps}_{{\cal X}^n}$ closely approximate $T_{{\cal X}^n}$, and as a result, we see that in the second plot of the second row, the points are also approximately concentrated around the straight line with slope $2$ (black line). However, it is important to recall that the EOT map is not necessarily scale equivariant in general. 
\begin{figure}[H]
    \centering
    \includegraphics[width=.8\linewidth]{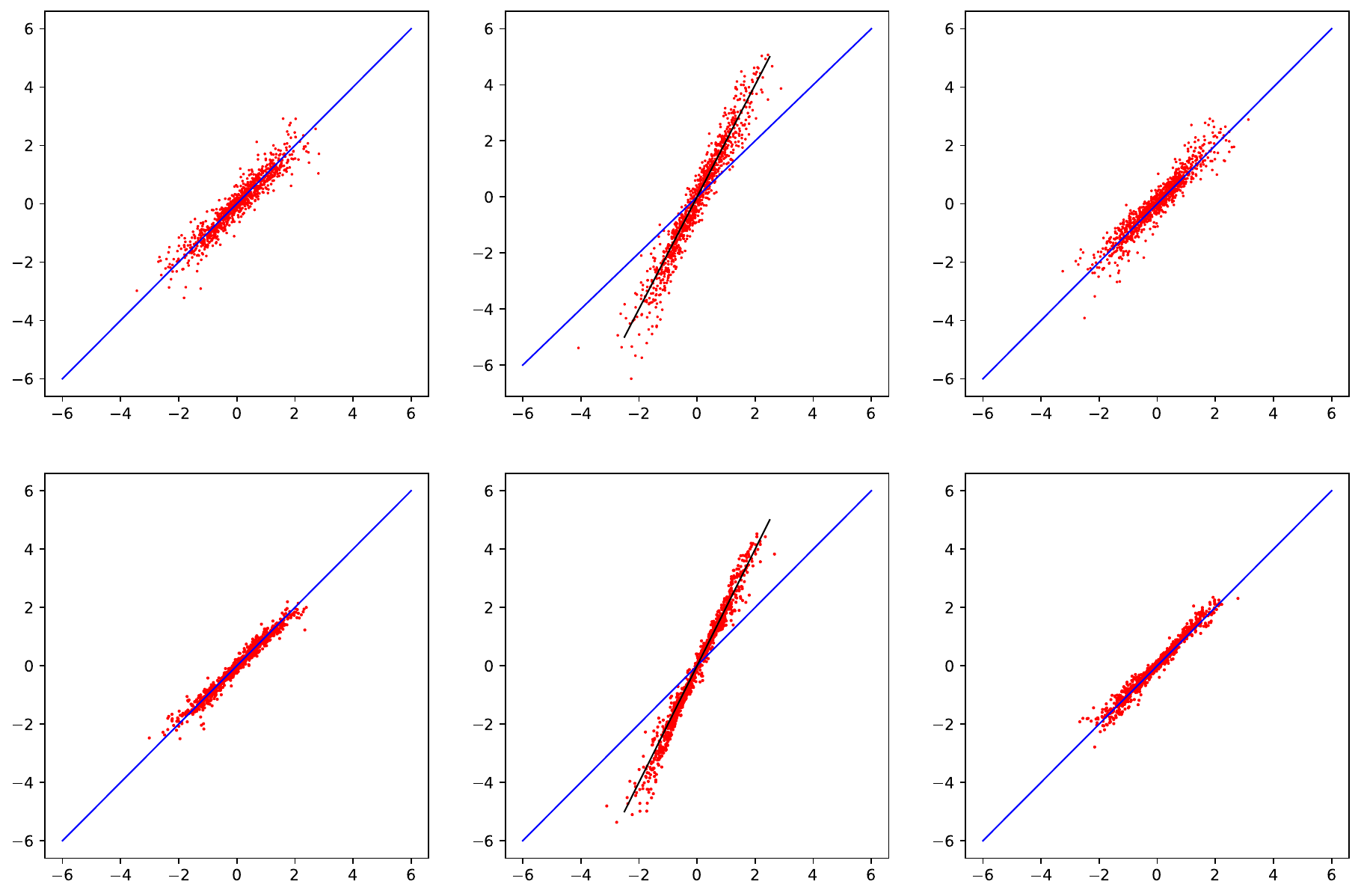}
    \caption{\small \sf Q-Q plots for two samples, the first one is drawn from a trivariate standard normal distribution, the second one from a trivariate normal distribution with covariance matrix $=diag(1,4,1)$. The first row displays OT Q-Q plot and the second one EOT Q-Q plot with $\veps=10^{-2}$. The straight blue lines represent the line $L$ and the black line in the second column has slope $2$ and intercept $0$.}
    \label{fig:Gauss vs Scaled Gauss}
\end{figure}
We present potential plots (left for OT and right for EOT) for the same pair of samples ${\cal X}^n$ and ${\cal Y}^n$ in Figure \ref{fig:G_vs_SG_pot}. Clearly, the scatter plots are spread out and not concentrated along the straight line $L$. This observation suggests that the underlying distributions corresponding to the two samples are not identical. But, unlike the Q-Q plots (shown in Figure \ref{fig:Gauss vs Scaled Gauss}), the potential plots do not offer additional information about the underlying distributions, such as the scaling in the sample $\calY^n$.
\begin{figure}[H]
    \centering
    \includegraphics[width=.8\linewidth]{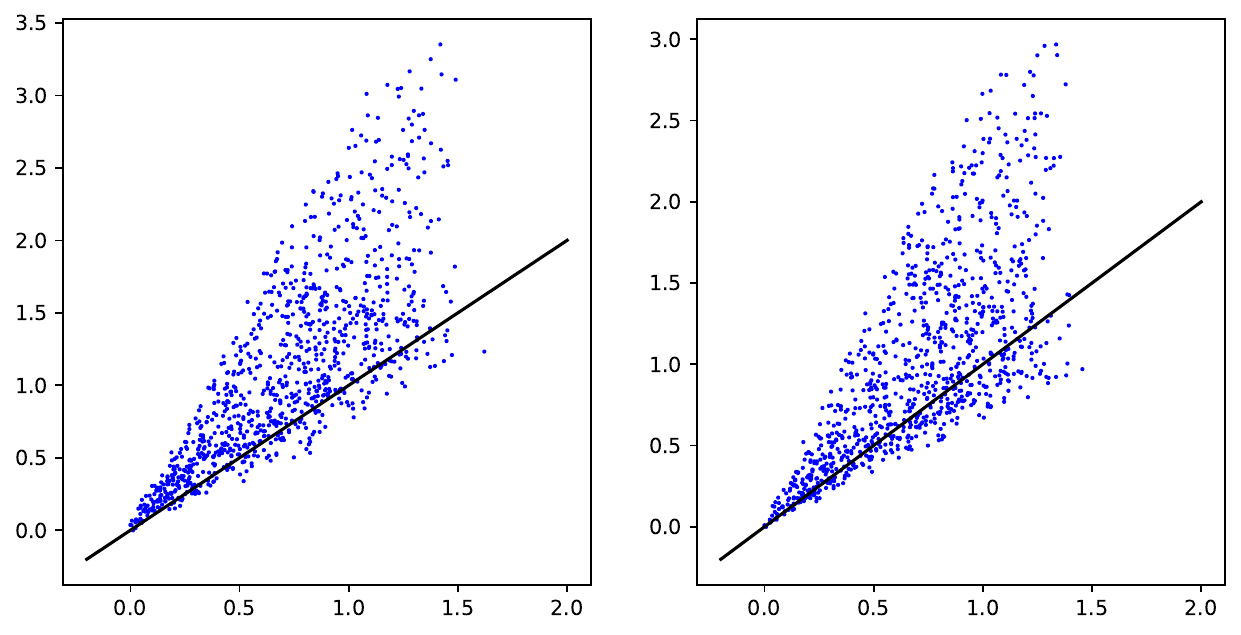}
    \caption{\small \sf Potential plots for two samples. The first sample is drawn from a trivariate standard normal distribution and the second one from a trivariate normal distribution with mean zero and covariance matrix $=diag(1,4,1)$. The left plot is the OT potential plot and the right one EOT potential plot with $\veps=10^{-2}$. The straight (black) lines represent the line $L$.}
    \label{fig:G_vs_SG_pot}
\end{figure}

Finally, we empirically estimate the values of the test statistics $E_n$ and $F_n$ defined in \eqref{dfn:teststat}, as well as the corresponding $p$-values under $H_0: \nu_X|_{K_2}=\nu_Y|_{K_2}$. The values are reported in Table~\ref{table:pvalue_two_Gaussian-scaled}. Since the $p$-values are relatively small, this supports the assertion that the two samples have originated from distinct distributions.
\vspace{.3cm}
\begin{table}[h!]
    \begin{tabularx}{1\textwidth} { 
        | >{\centering\arraybackslash}X 
        | >{\centering\arraybackslash}X 
        | >{\centering\arraybackslash}X 
        | >{\centering\arraybackslash}X 
        | >{\centering\arraybackslash}X|}
       \hline
        $n$& $E_n$ & $p$-value of $E_n$ & $F_n$ & $p$-value of $F_n$\\
       \hline
       $250$  & $222.95$  & $0$ & $ 63.58$ & $0$ \\
      \hline
       $500$ & $394.71$ & $0$ & $121.40$ & $0$\\
      \hline
      $1000$  & $796.47$  & $0$ & $267.15$ & $0$ \\
      \hline
      \end{tabularx}
      \caption{\small\sf Test statistics and $p$-values under the null hypothesis $H_0: \nu_X|_{K_2}=\nu_Y|_{K_2}$.}
        \label{table:pvalue_two_Gaussian-scaled}
\end{table}

\subsection{Example: Outlier Detection}\label{ssec:ex_outlier}

We show that outliers can be detected using OT and EOT Q-Q plots. We illustrate this with two simulated samples, ${\cal X}^n$ and ${\cal Y}^n$, drawn from trivariate standard normal distributions, each containing 1000 observations. Subsequently, we replace three observations in ${\cal Y}^n$ with outlier points $(8, 8, 8)$, $(9, 9, 9)$, and $(10, 10, 10)$, resulting in the transformed set ${\bar{\cal Y}}^n$. We follow steps (ii)--(iii) to build the Q-Q plots for ${\cal X}^n$ and ${\bar{\cal Y}}^n$, and display them in Figure \ref{fig:sG_vs_sGoutlier}. The three black points, which are far from the rest of the observations, clearly reveal the presence of outliers. Next, we present the potential plots (step (iv)) for these two samples in Figure~\ref{fig:sG_vs_sGoutlier_potential}. 

Notice that, although there are $3$ outliers present in the second sample, we only see one outlying point in the potential plot (compare this with Q-Q plots in Figure~\ref{fig:sG_vs_sGoutlier}). It appears that potential plots may be less informative in this case compared to the Q-Q plots.

It is worth noticing that we set the value of the regularisation parameter as $\varepsilon = 10^{-3}$ in this example, whereas in previous examples, we chose a larger value of $\varepsilon = 10^{-2}$. We observe that, with a higher value of $\varepsilon$, the EOT Q-Q plots do not distinctly separate the outliers. We shall study empirically the role of the regularisation parameter in visual analysis in Subsection \ref{ssec:regularisation parameter}.
\begin{figure}[H]
    \centering
    \includegraphics[width=.8\linewidth]{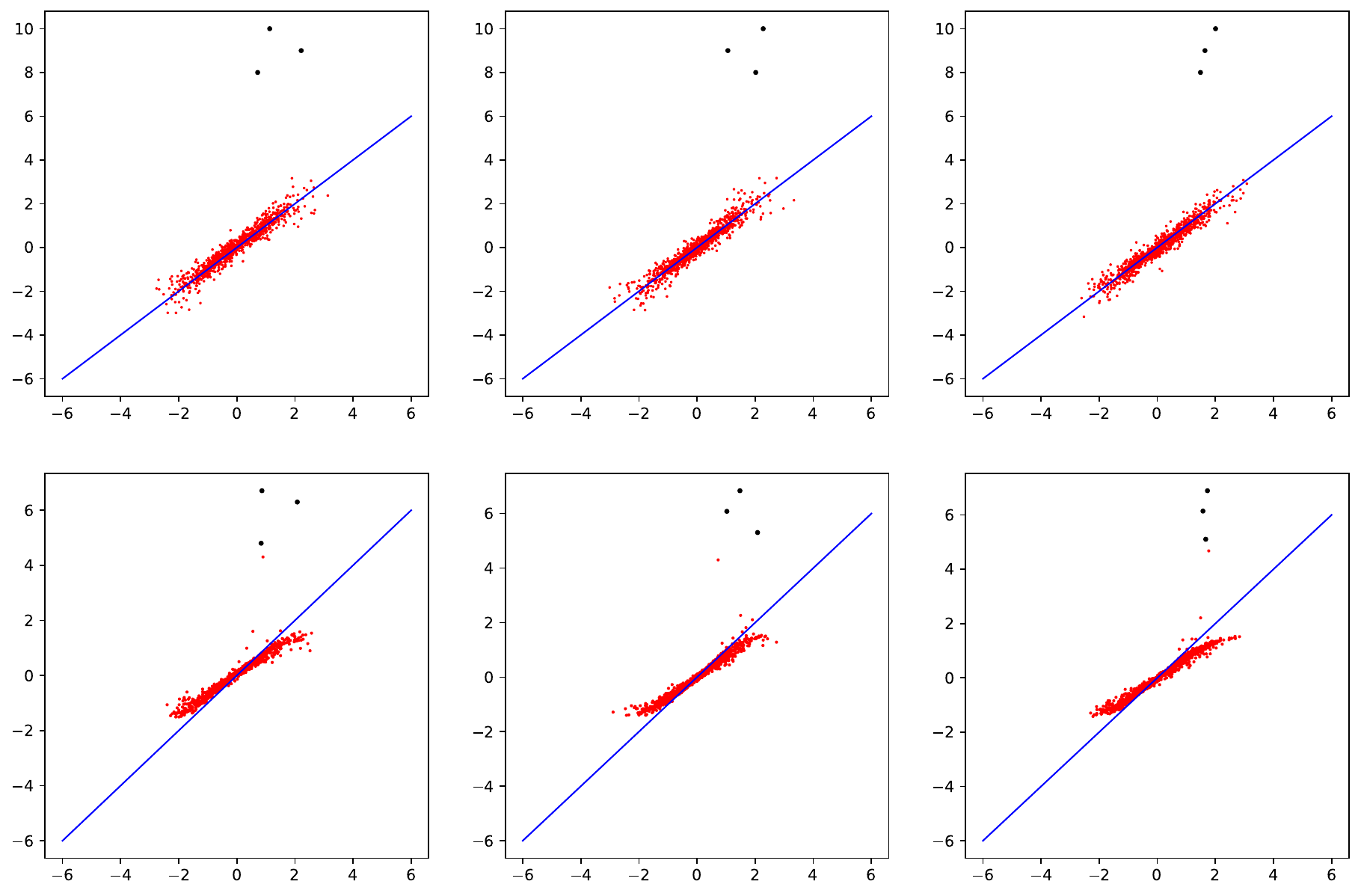}
    \caption{\small \sf Q-Q plots for two multivariate samples, each of size $1000$. Both samples are generated from a multivariate standard normal, while three observations in the second sample are replaced by three outliers. First and second rows display the OT and EOT Q-Q plots, respectively. The regularisation parameter for the EOT is chosen as $\veps=10^{-3}$. The straight (blue) lines represent the line $L$.}
    \label{fig:sG_vs_sGoutlier}
\end{figure}
\begin{figure}[H]
    \centering
    \includegraphics[width=.8\linewidth]{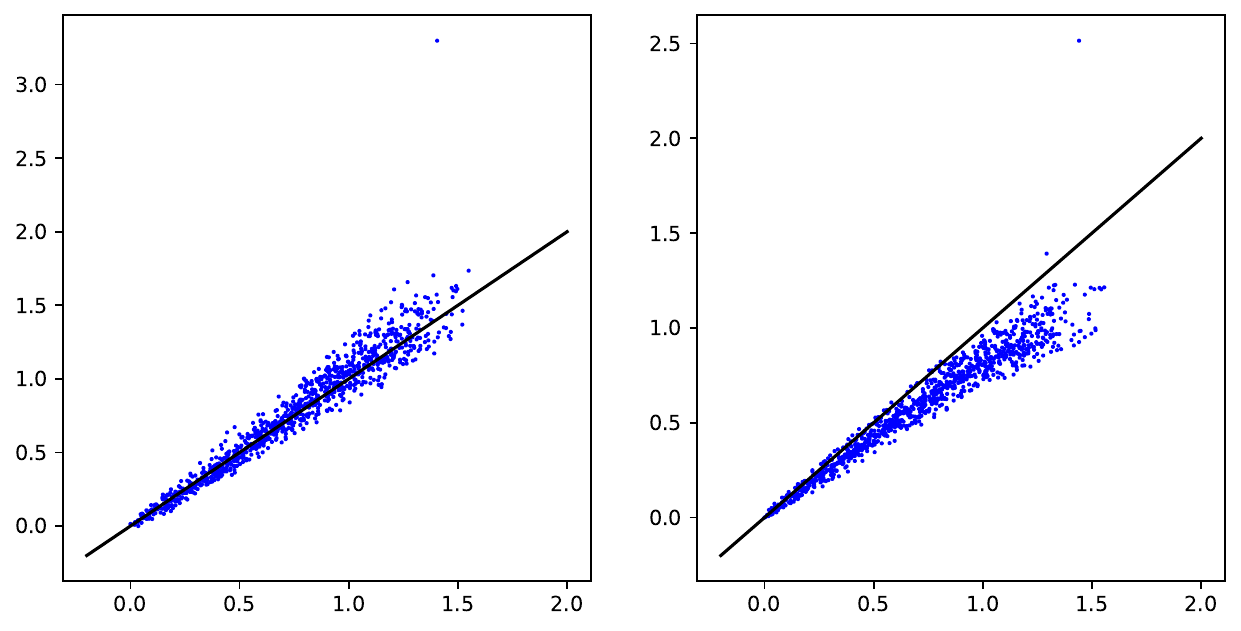}
    \caption{\small\sf Potential plots (left for OT, right for EOT with $\veps=10^{-3}$) for two multivariate samples, each of size $1000$. Both samples are generated from multivariate standard normal, while only the second sample is contaminated by some outliers. The straight (black) lines represent the line $L$.}
    \label{fig:sG_vs_sGoutlier_potential}
\end{figure}

\subsection{Example: Light tail vs Heavy tail}\label{ssec:ex_light_vs_heavy}

%\paragraph{$(b)$ Gaussian vs Student-t --}\label{ssec: Gauss vs Student-t}
%
To compare the heaviness of tail distributions, we provide a first example in dimension $3$ considering two distributions: the standard normal distribution (for the light tail) and the Student's  $t$--distribution (for the heavy one). Another example comparing the trivariate standard normal distribution with i.i.d. Pareto($3$) marginals can be found in Appendix~\ref{App:LightHT}.

Recall that the density of the  Student's $t$--distribution in dimension $d$ is given by
\begin{equation}
    f(x)=\frac{\Gamma[(r+d)/2]}{\Gamma(r/2)r^{d/2}\pi^{d/2}|\Sigma|^{1/2}}\left[ 1+\frac{1}{r}(x-m)^T\Sigma^{-1}(x-m)\right]^{-(r+d)/2},
\end{equation}
where $m\in \real^d$, $\Sigma$ is a $d \times d$ symmetric positive semidefinite matrix and $\Gamma(\cdot)$ is the gamma function. The parameter $r>0$ determines the heaviness of the distribution in the sense that moments of order greater than $r$ are infinite. We follow step (i) and draw a sample $\calX^n$ of size $1000$ from the trivariate standard normal distribution $(\nu_X)$. Similarly we draw $\calY^n$  of the same size from a trivariate Student's $t$--distribution, with parameters $m=0$, $\Sigma=I_3$ (identity in $\real^3$) and  $r=3.2$ $(\nu_Y)$. We go through steps (ii)--(iii), and display the Q-Q plots in Figure~\ref{fig:normal vs mul_t}. 
\begin{figure}[H]
    \centering
    \includegraphics[width=1\linewidth]{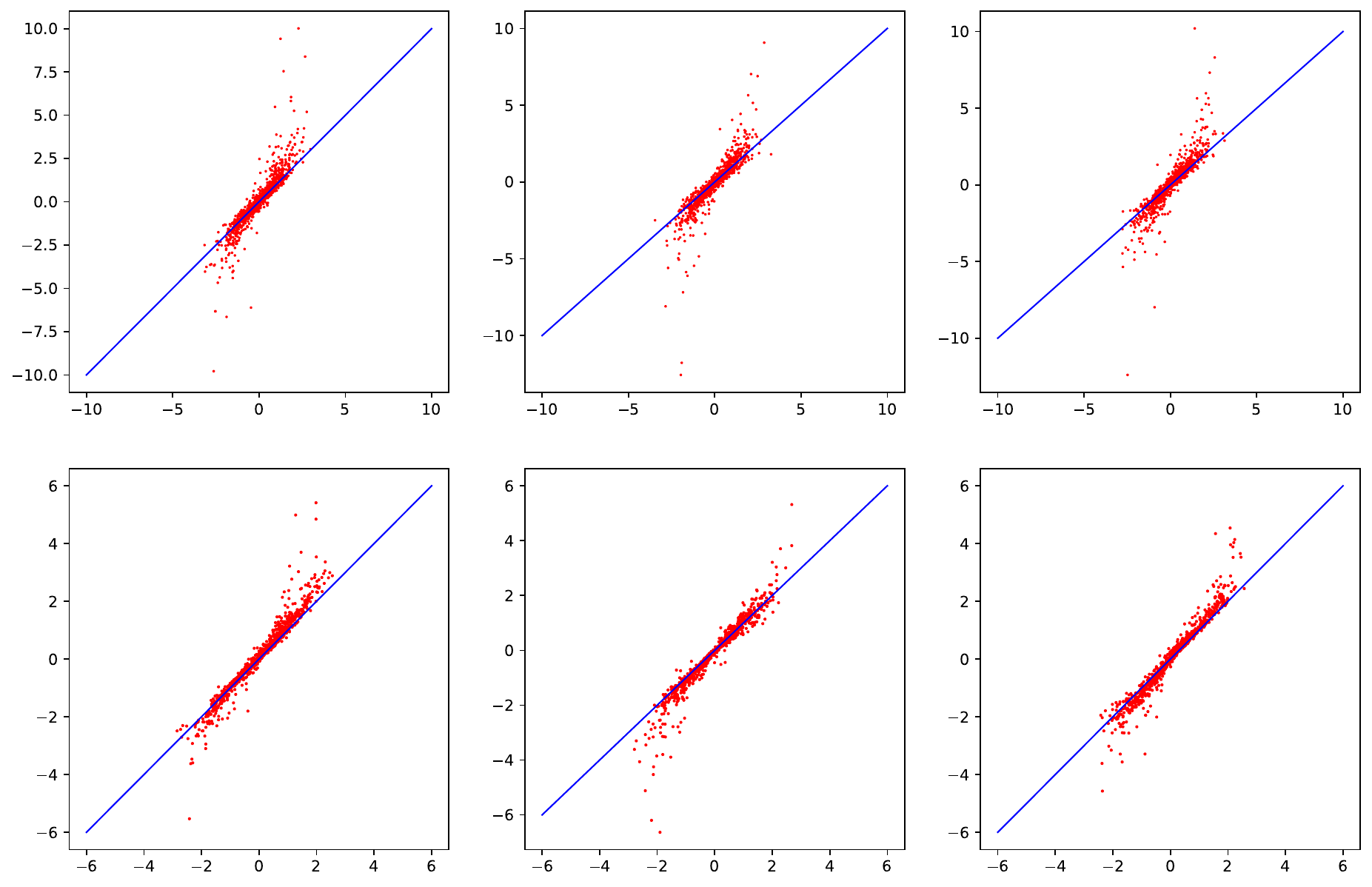}
    \caption{\small \sf Q-Q plots for two samples each of size $1000$, the first drawn from a trivariate standard normal distribution, the second from a multivariate Student's t-distribution. The first row display OT Q-Q plot and second EOT Q-Q plot, taking $\veps=10^{-3}$.  The straight (blue) lines represent the line $L$.}
    \label{fig:normal vs mul_t}
\end{figure}
Here, we observe that the OT and EOT quantiles for $\calY^n$ grow at a faster rate than those of $\calX^n$. The peculiar shape (S-shape) of the scatter plot seen in Figure~\ref{fig:normal vs mul_t} is very familiar in the univariate Q-Q plots involving heavy vs light tail comparisons. This hints at a similar behaviour in the OT plots. This feature, which is very useful in the univariate analysis, needs further exploration in the (E)OT setting.

\begin{table}[h!]
    \begin{tabularx}{1\textwidth} { 
        | >{\centering\arraybackslash}X 
        | >{\centering\arraybackslash}X 
        | >{\centering\arraybackslash}X 
        | >{\centering\arraybackslash}X 
        | >{\centering\arraybackslash}X|}
       \hline
        $n$& $E_n$ & $p$-value of $E_n$ & $F_n$ & $p$-value of $F_n$\\
       \hline
       $250$  & $670.64$  & $0$ & $ 13.30$ & $0.33325$ \\
      \hline
       $500$ & $1410.92$ & $0$ & $33.52$ & $0.003375$\\
      \hline
      $1000$  & $2771.30$  & $0$ & $92.38$ & $0.000125$ \\
      \hline
      \end{tabularx}      
      \caption{\small\sf Test statistics and $p$-values under the null hypothesis $H_0: \nu_X|_{K_2}=\nu_Y|_{K_2}$.}
          \label{table:pvalue_Gaussian-t3,2}
        \end{table}
We then present potential plots (step (iv)) for this particular example in Figure~\ref{fig:normal vs mul_t potential}. It can be seen from the plots that the points are scattered around a nonlinear (increasing) curve above the line $L$. This implies that the potential function associated with the second sample has a higher growth rate. Since the (E)OT maps are gradient of (E)OT potentials, a higher growth rate of potential implies that the corresponding distribution has a heavier tail than that of first one.

\begin{figure}[H]
    \centering
    \includegraphics[width=.8\linewidth]{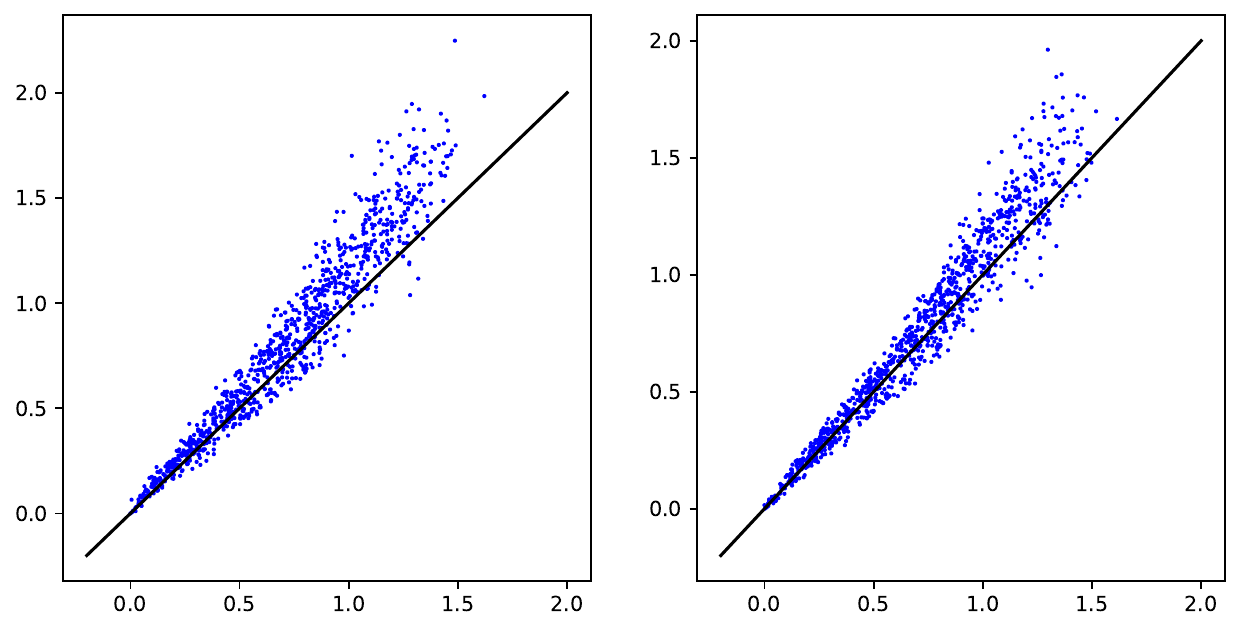}
    \caption{\small \sf Potential plots (left for OT, right for EOT with $\veps=10^{-3}$) for two multivariate samples, each of size $1000$; the first one is drawn from a trivariate standard normal distribution and the second from a multivariate  Student’s t-distribution. The straight (black) lines represent the line $L$.}
    \label{fig:normal vs mul_t potential}
\end{figure}
The $p$-values corresponding to $E_n$ and $F_n$ are provided in Table~\ref{table:pvalue_Gaussian-t3,2}, although not so informative given the question on tail behaviour.

%The same conclusion can be drawn from the EOT potential plot as well, since it closely approximates the OT potential.
%\vspace{.3cm}
\subsection{The effect of the regularisation parameter in EOT} 
\label{ssec:regularisation parameter}
 
As stated in Proposition \ref{thm:unique characterisation by EOT}, both EOT map and potential uniquely characterise distributions, regardless of the value of the regularisation parameter $\veps$. However, the regularisation parameter plays an important role in the visual analysis. 
For a reasonably small value of $\veps$ (depending upon the distribution), the EOT map and potential closely approximate the OT map and potential, respectively. As a consequence, the EOT plots (Q-Q and potentials) look very similar to the OT plots (Q-Q and potentials); see for e.g. Figures~\ref{fig:comp same Gaussian}--\ref{fig:G_vs_SG_pot}. 

On the other hand, the EOT map (resp. potential), with a large value of $\veps$, is far from being a close approximation of the OT map (resp. potential). 

Taking back the setup of Example~\ref{ssec:ex1}, we compute and display in Figures~\ref{fig:sG_vs_sG_diff_parameters} and \ref{fig:sG_vs_sG_potentials_diffparameters}, respectively, the EOT Q-Q plots and potential plots when varying $\veps$. We consider three values of $\veps$: $10^{-3}, 10^{-2}$ and $10^{-1}$, respectively. We observe that points are not only scattered around the line $L$, they also begin to concentrate around a point on the line as the value of $\veps$ increases.

\begin{figure}[H]
    \centering
    \includegraphics[width=.8\linewidth]{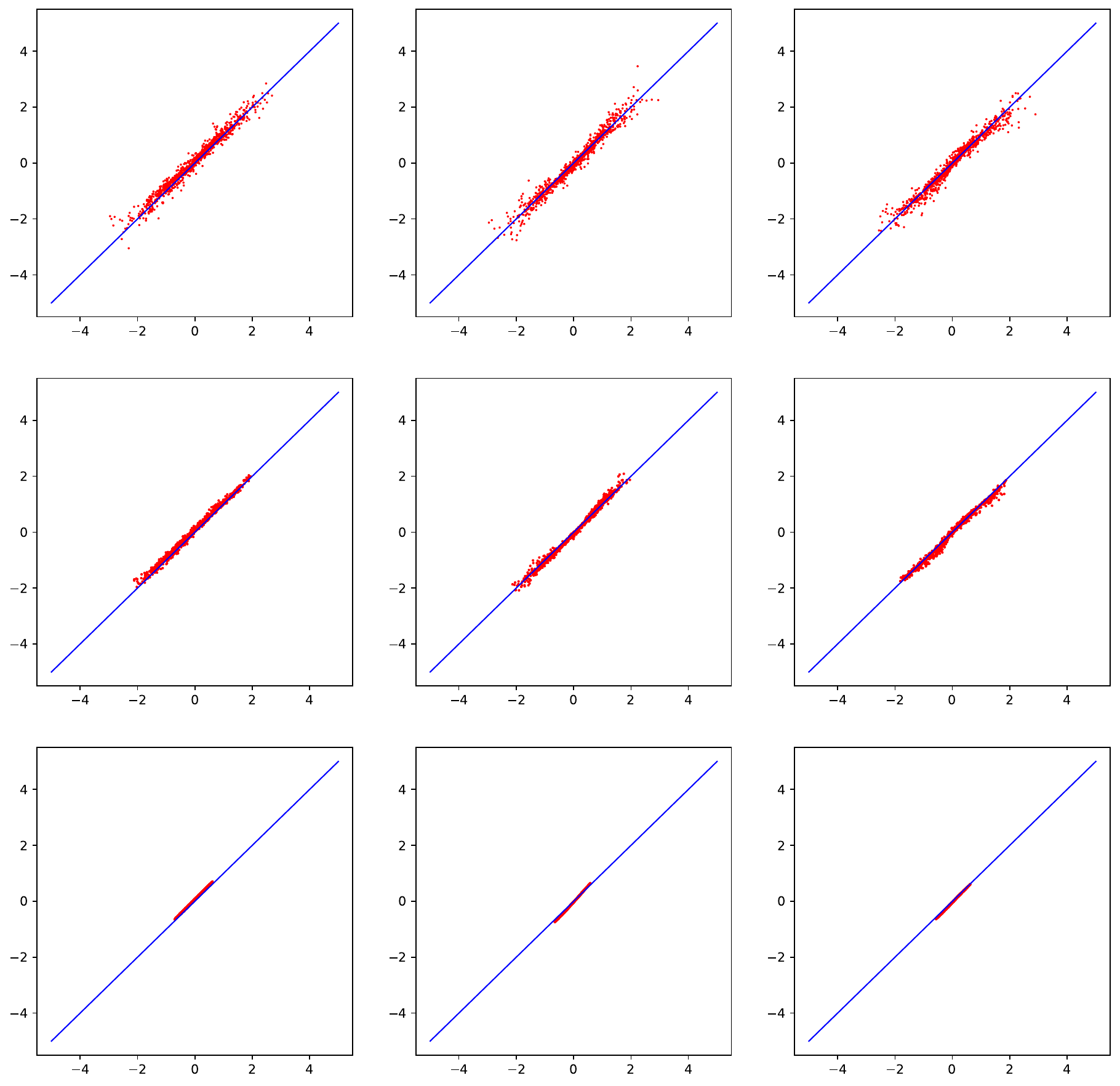}
    \caption{\small \sf EOT Q-Q plots for two samples, both generated from a trivariate standard normal distribution. The value of the regularisation parameter $(\veps)$ is $10^{-3}$ for the first row, $10^{-2}$ for the second row and $10^{-1}$ for the third row. The straight (blue) lines represent $L$.}
    \label{fig:sG_vs_sG_diff_parameters}
\end{figure}
Recall that, as $\varepsilon$ increases to $+\infty$, the EOT map converges to a constant map, with the constant being the mean of the target distribution. Therefore, for larger $\varepsilon$ (e.g., in this case $\veps=10^{-1}$), the EOT map is close to a constant, resulting in concentration around a point in the Q-Q plot. 
We next consider the same samples as in Example~\ref{ssec:ex_light_vs_heavy}. We display the EOT Q-Q plots in Figure~\ref{fig:EOTQQ_sG_vs_mul_t_diff_par}, for three different values of $\veps$, $10^{-3}, 10^{-2}$ and $10^{-1}$, respectively. Observe that, although the points are deviating from the line $L$, which indicates that the two samples are non similar, they do not reveal the tail heaviness (compared with Example \ref{ssec:ex1}).
\begin{figure}[H]
    \centering
    \includegraphics[width=.8\linewidth]{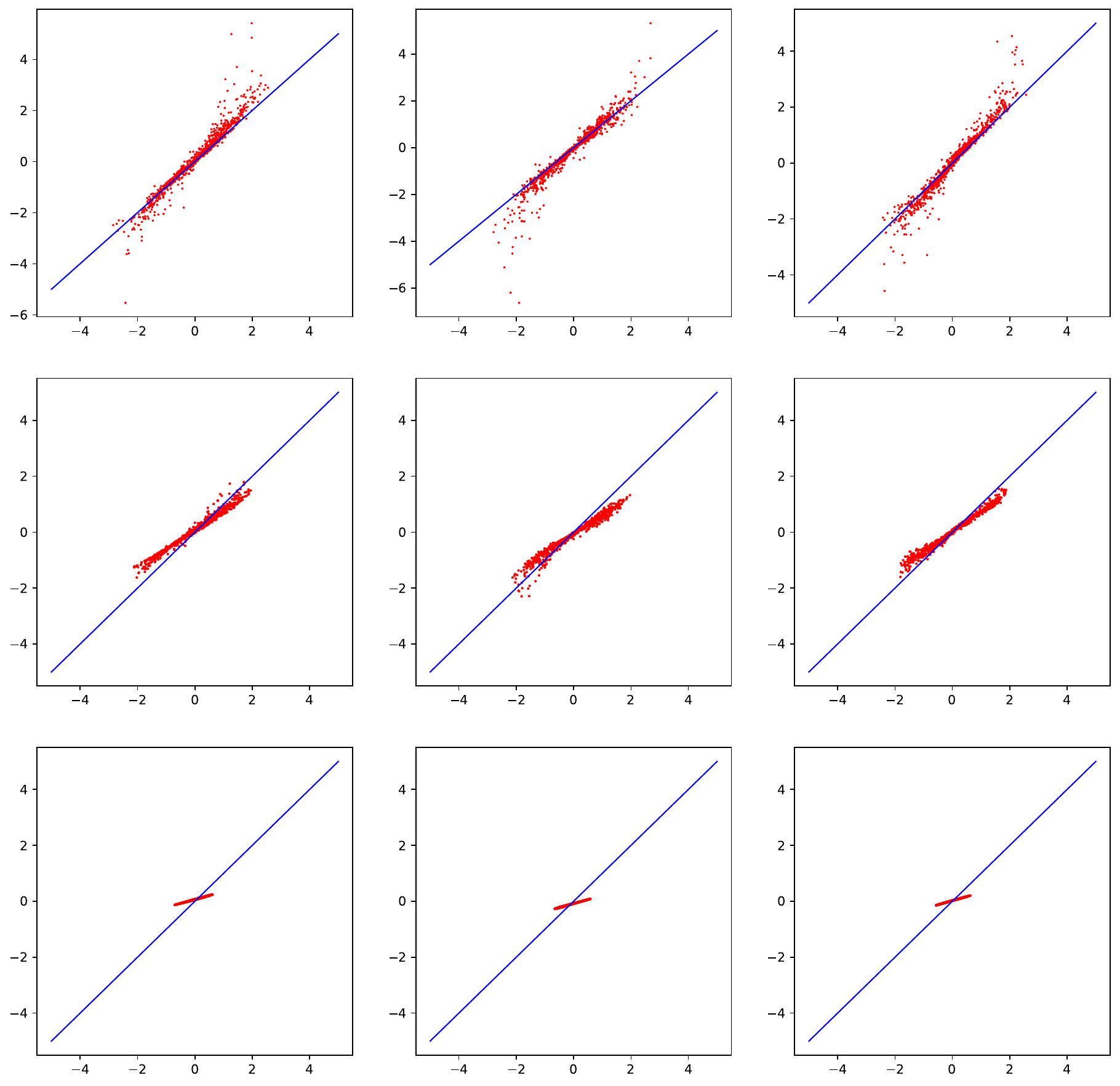}
    \caption{\small \sf EOT Q-Q plots for two samples. The samples are chosen as in Example~\ref{ssec:ex_light_vs_heavy}. The value of regularisation parameter $(\veps)$ is $10^{-3}$ for the first row, $10^{-2}$ for the second row and $10^{-1}$ for the third row. The blue lines represent $L$.}
    \label{fig:EOTQQ_sG_vs_mul_t_diff_par}
\end{figure}
Therefore, for a large value of $\veps$, while the EOT Q-Q plot can distinguish samples from different distributions, it may not reveal specific details such as the heaviness of the tails.\\[1ex]
Let us now discuss the EOT potential plots for different values of $\veps$. We first consider the setup in Example~\ref{ssec:ex1}, let both samples are drawn from a trivariate standard normal distribution. We display the EOT potential plots in Figure~\ref{fig:sG_vs_sG_potentials_diffparameters}, for three different values of $\veps$, $10^{-3}, 10^{-2}$ and $10^{-1}$, respectively.
Observe that, for large values of $\veps$, the plots are not convincing enough to infer that the two sets of samples are generated from the same  distribution.
\begin{figure}[H]
    \centering
    \includegraphics[width=.8\linewidth]{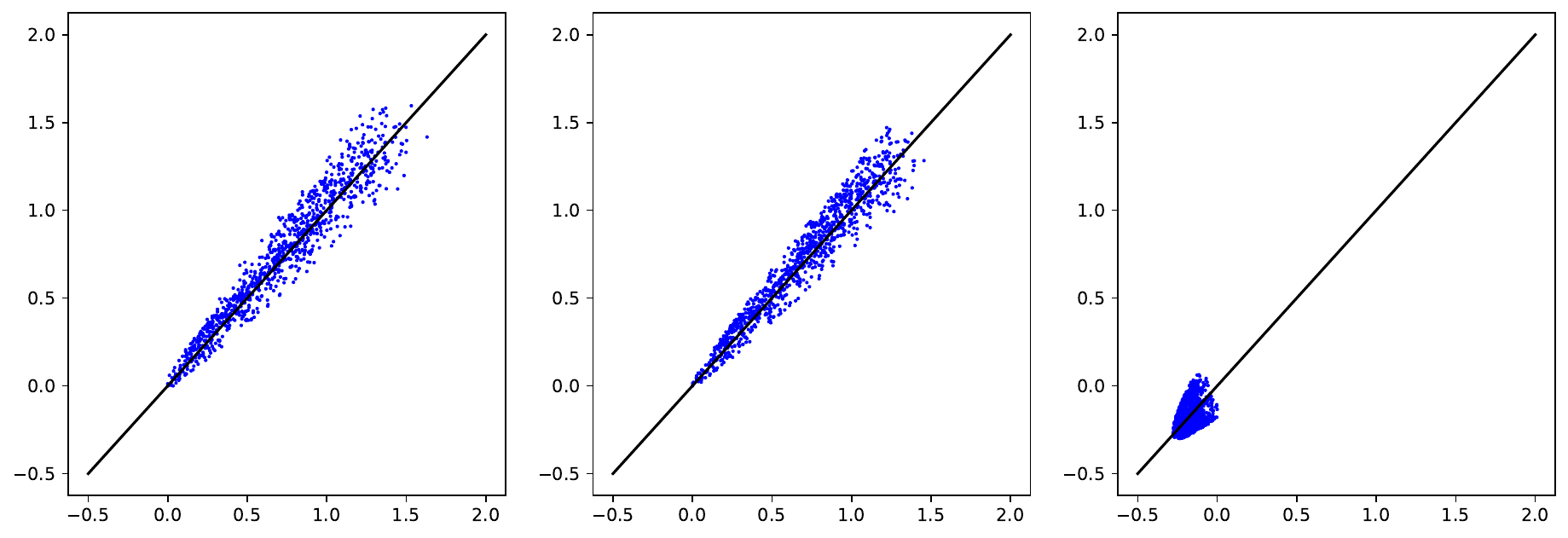}
    \caption{\small \sf EOT potential plots for two samples, both generated from a trivariate standard normal distribution. The value of the regularisation parameter $(\veps)$ for the left plot is $10^{-3}$, middle $10^{-2}$ and right $10^{-1}$. The blue lines represent $L$.}
    \label{fig:sG_vs_sG_potentials_diffparameters}
\end{figure}
\begin{figure}[H]
    \centering
    \includegraphics[width=.8\linewidth]{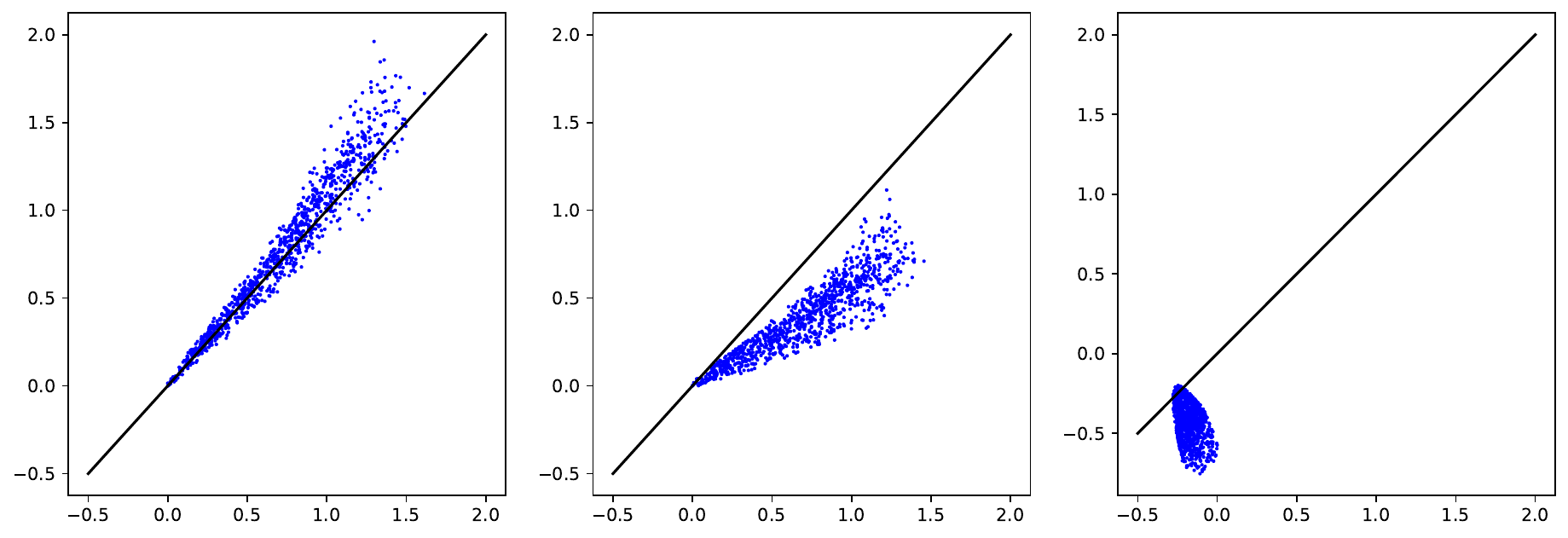}
    \caption{\small \sf EOT potential plots for two samples. We consider the same sample as in part $(b)$ of Example~\ref{ssec:ex_light_vs_heavy}. The value of the regularisation parameter $(\veps)$ is $10^{-3}$ for the left plot, $10^{-2}$ for the middle one and $10^{-1}$ for the right one. The blue lines represent $L$.}
    \label{fig:sG_vs_mult_potential_diffparam2}
\end{figure}
We show another example of potential plots in Figure~\ref{fig:sG_vs_mult_potential_diffparam2}. The samples are the same as in Example~\ref{ssec:ex_light_vs_heavy}. Observe that, although the plots for each value of $\veps$ suggest that the samples are not very similar, it is difficult to infer visually (especially when $\veps$ is large) that one of the samples has a heavier tail than the other (Unless the EOT plot is very close to the OT one as in Figure~\ref{fig:normal vs mul_t potential}).

%{\color{red}[we will need to comment a bit more on the dynamic behavior of quantiles of Y vs those of X, once we have the average fit in the plots]}

It seems therefore ideal to choose small values of $\veps$ for better visual analysis. Nevertheless, it is important to note that smaller the value of $\veps$, longer the algorithm takes to converge. Moreover, if $\veps$ is chosen too small, the algorithm (Sinkhorn) need more iteration to converge.

\subsection{Key takeaways}

\begin{itemize}
    \item[(i)] We observe that both (OT and EOT) Q-Q plots and potential plots can efficiently identify the similarity between distributions. We experiment this with i.i.d and non i.i.d Gaussian distributions.
    \item[(ii)] When the samples are generated from a same distribution (in our example, a Gaussian distribution) with one of them having a scaling or shifting, we can identify this effect with Q-Q plots. Potential plots do not identify the shift, but the scale. 
    \item[(iii)] We have shown with some examples that outliers can be detected with OT Q-Q plots and potential plots. For EOT, one needs to choose the value of the regularisation parameter to be small enough to observe them.
    \item[(iv)] Considering two examples, we showed that both Q-Q and potential plots can be used as a visual tool to compare tail distributions. For instance, comparing a light tail sample (e.g. Gaussian sample) with a heavy tail one (e.g. multivariate Student's t-distribution), we could identify on the plots the heaviest tail between the two distributions.
    \item[(v)] When comparing distributions in high dimension, potential plots give a $2$-dimensional representation, while a large number of componentwise plots is needed in the case of OT (and also geometric) Q-Q plots. As such, potential plots can be very useful in applied fields, e.g. risk management, as they offer:
    \begin{itemize}
        \item A discriminating tool between light and heavy tail, 
        \item A visual validation tool for multivariate modelling.
    \end{itemize} 
\end{itemize}
Given these characteristics, to compare visually two multivariate distributions, we would recommend to proceed as follows:
\begin{enumerate}
    \item Plot the potential function to detect if the two samples are drawn from a same multivariate distribution;
    \item To obtain the whole comparison with specific features, consider the (E)OT Q-Q plots. 
\end{enumerate}

As a last remark, if comparing two samples which one suspects to come from a same family of distributions, then one may standardize the data to compute the (E)OT potential and Q-Q plots.

\section{Application on real data}\label{sec:ex_real_data}

In the previous section, we considered many scenarios and experiments with simulated samples to have a better understanding of (E)OT Q-Q and potential plots as a visual tool to compare multivariate distributions. Now, we can turn to applications on real data. We consider two examples.

The first example is the Fisher's Iris dataset, which can be downloaded from\\ 
\href{https://archive.ics.uci.edu}{https://archive.ics.uci.edu}, a standard dataset used in statistics. It was also considered in \cite{Dhar2014} for analysing multivariate Q-Q plots based on geometric quantile, giving us a way to compare the results obtained when choosing two types of multivariate quantiles (see Section~\ref{sec:comparisonGeom}). The dataset has $4$ variables, and $50$ observations for each of the $4$ variables. Due to the relatively smaller size of this dataset, we consider another example offering a larger sample size. This second example is the Turkish rice Osmanic dataset, downloaded also from the link 
\href{https://archive.ics.uci.edu}{https://archive.ics.uci.edu}. It has $5$ variables, and $2180$ observations for each of the $5$ variables.

\subsection{Example 1: Fisher's Iris data}\label{applic:Iris}

This Fisher's Iris dataset is constituted of three multivariate samples corresponding to three varieties: Iris Setosa, Iris Versicolour and Iris Virginica. Each variety consists of $50$ observations and each observation contains the measurement of $4$ variables, namely: sepal length, sepal width, petal length and petal width. Therefore, each sample is $4$-dimensional with a size of $50$. Following the steps (ii) to (iv) described in Section~\ref{sec: Illustration} to build the (E)OT potential and Q-Q plots, we compare each of the three samples with a 4-dimensional Gaussian sample of size 50. This choice of comparison is also motivated by the fact that \cite{Dhar2014} showed a strong evidence of normality of those samples using the geometric Q-Q plots. For a fair comparison, we follow the procedure of standardising the dataset, like in \cite{Dhar2014}, before computing various quantiles and potentials.
%, but compute its first two empirical moments to consider the Gaussian distribution (to which we compare the data) with the same mean and covariance matrix as that of the data; we denote these two empirical moments by $\hat m$ and $\hat \Sigma$, respectively.}\\[1ex]
%
With the (E)OT approach, we obtain the potential plots displayed in Figure~\ref{fig:Iris_potential} for the three varieties of iris, and the Q-Q plots in Figure~\ref{fig:Iris_QQ}, for the Iris Setosa, Iris Versicolour, and Iris Virginica, respectively. Note that we also computed the associated test statistics and $p$-values but choose not to report them here, as they heavily rely on the asymptotics of the test statistics, which is not compatible with such a small sample size. 
Looking at the potential plots, we observe that the points are very much spread out, which would hint at non--Gaussianity of the samples (Iris Setosa, Iris Versicolour and Iris Virginica). However, as already pointed out, the sample size is very small for any reasonable inference. 
\begin{figure}[H]
\begin{minipage}{0.32\textwidth}
    \centering
    \includegraphics[width=.9\linewidth]{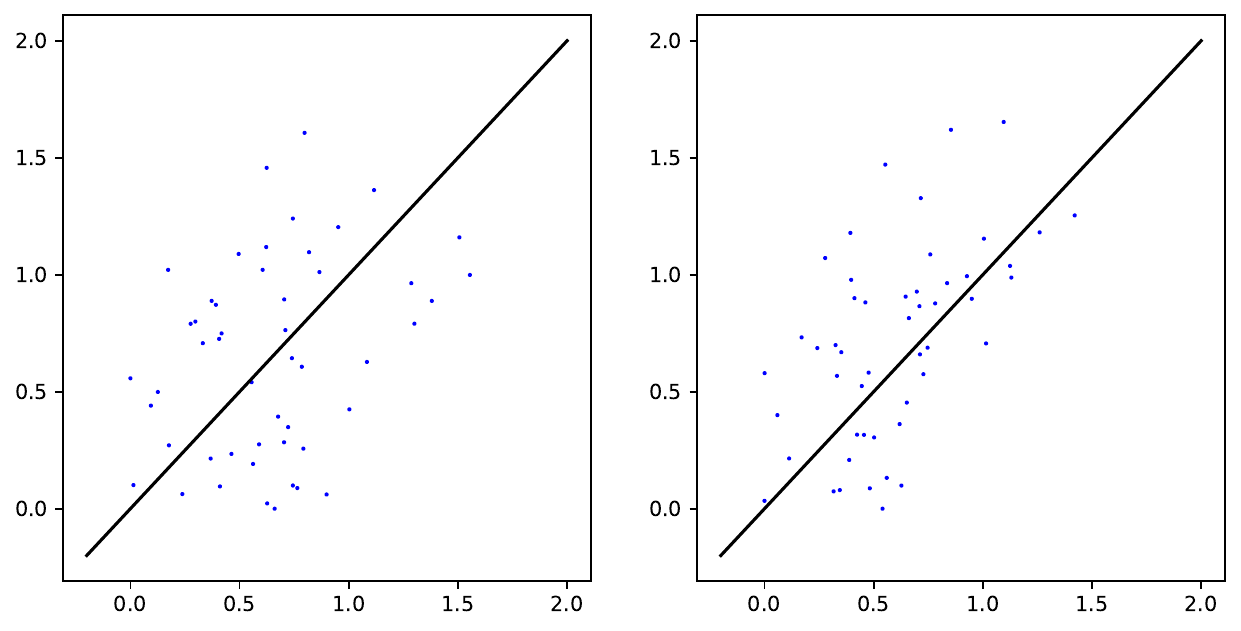}
   \\{\small\sf Iris Setosa data}
\end{minipage}
\hfill
\begin{minipage}{0.32\textwidth}
    \centering
    \includegraphics[width=.9\linewidth]{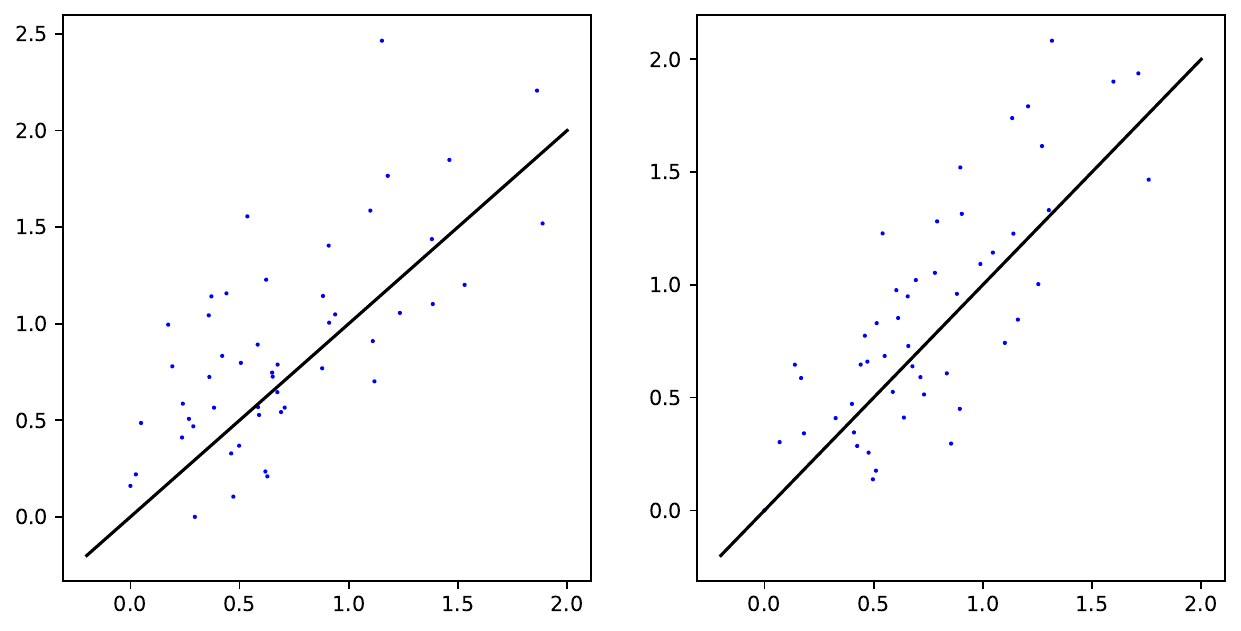}
    \\{\small\sf Iris Versicolour data}
\end{minipage}
\hfill
\begin{minipage}{0.32\textwidth}
        \centering
        \includegraphics[width=.9\linewidth]{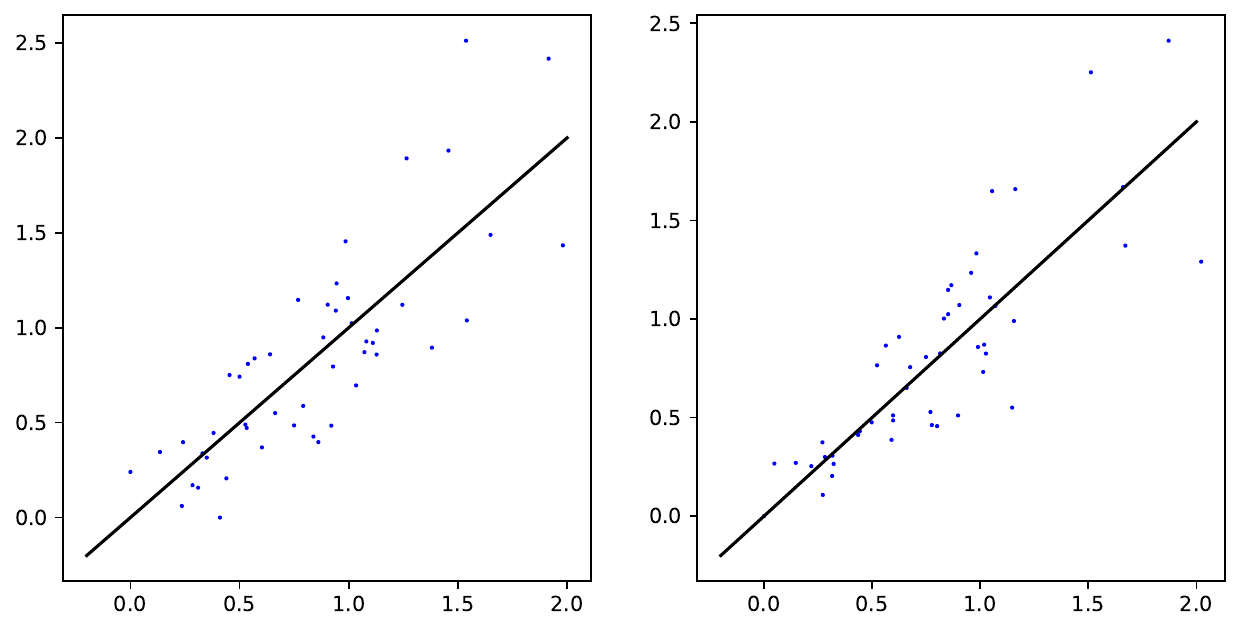}
        \\{\small\sf Iris Virginica data}
\end{minipage}
\caption{\footnotesize\sf Potential plots for two multivariate samples, each of size $50$; the first one is drawn from a $4$-dimensional standard normal  distribution and the second one is the standardised Iris Setosa data for the first (left) pair of plots, the standardised Iris Versicolour data for the second (middle) pair of plots, and the standardised Iris Virginica data for the third (right) pair of plots. For each pair, the left plot is OT potential plot and the right one is EOT potential, choosing $\veps=10^{-3}$. The straight (black) lines represent the line $L$.}
\label{fig:Iris_potential}
\end{figure}
We now look at the (E)OT Q-Q plots, to compare the data with a multivariate standard Gaussian distribution. In Figure~\ref{fig:Iris_QQ}, for each of the three samples, the points are loosely concentrated around the straight line $L$ in the OT Q-Q plots, while in EOT Q-Q plots, they appear a bit more concentrated around the line $L$.

\begin{figure}[H]
    \centering
    \includegraphics[width=.6\linewidth]{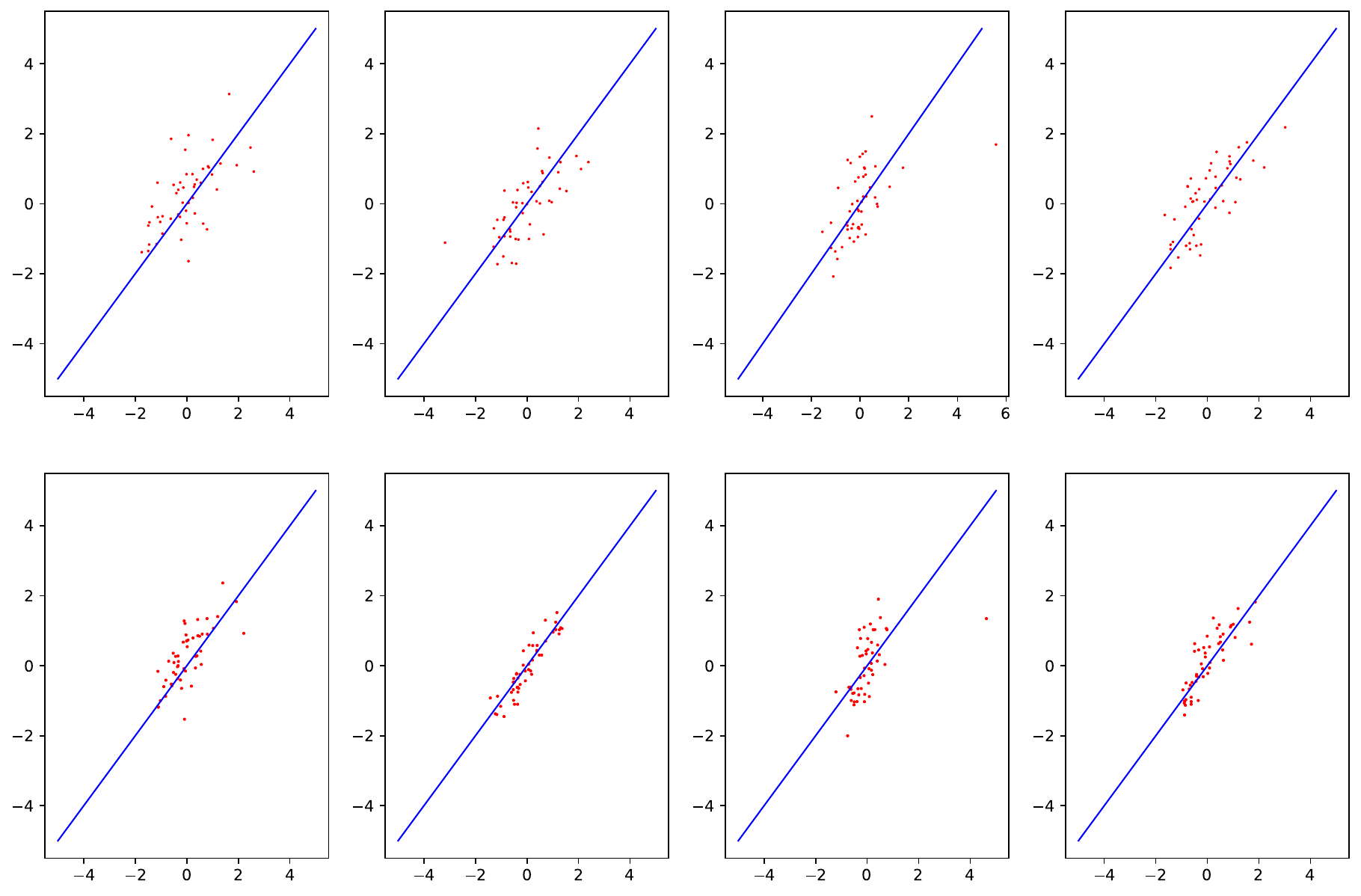}
        \\{\sf Iris Setosa data} 
        \\
        \includegraphics[width=.6\linewidth]{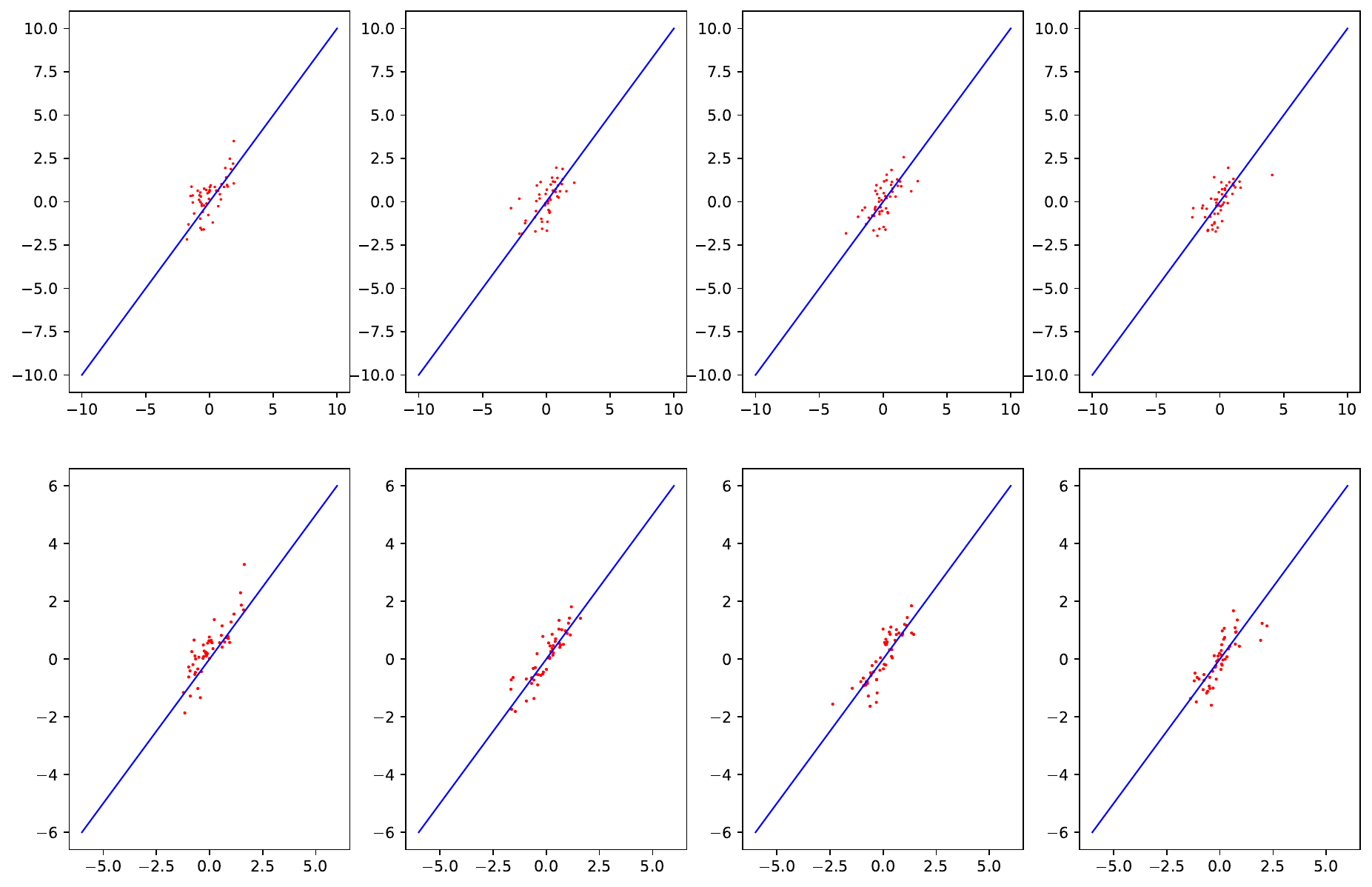}
        \\{\sf Iris Versicolour data}
        \\
        \includegraphics[width=.6\linewidth]{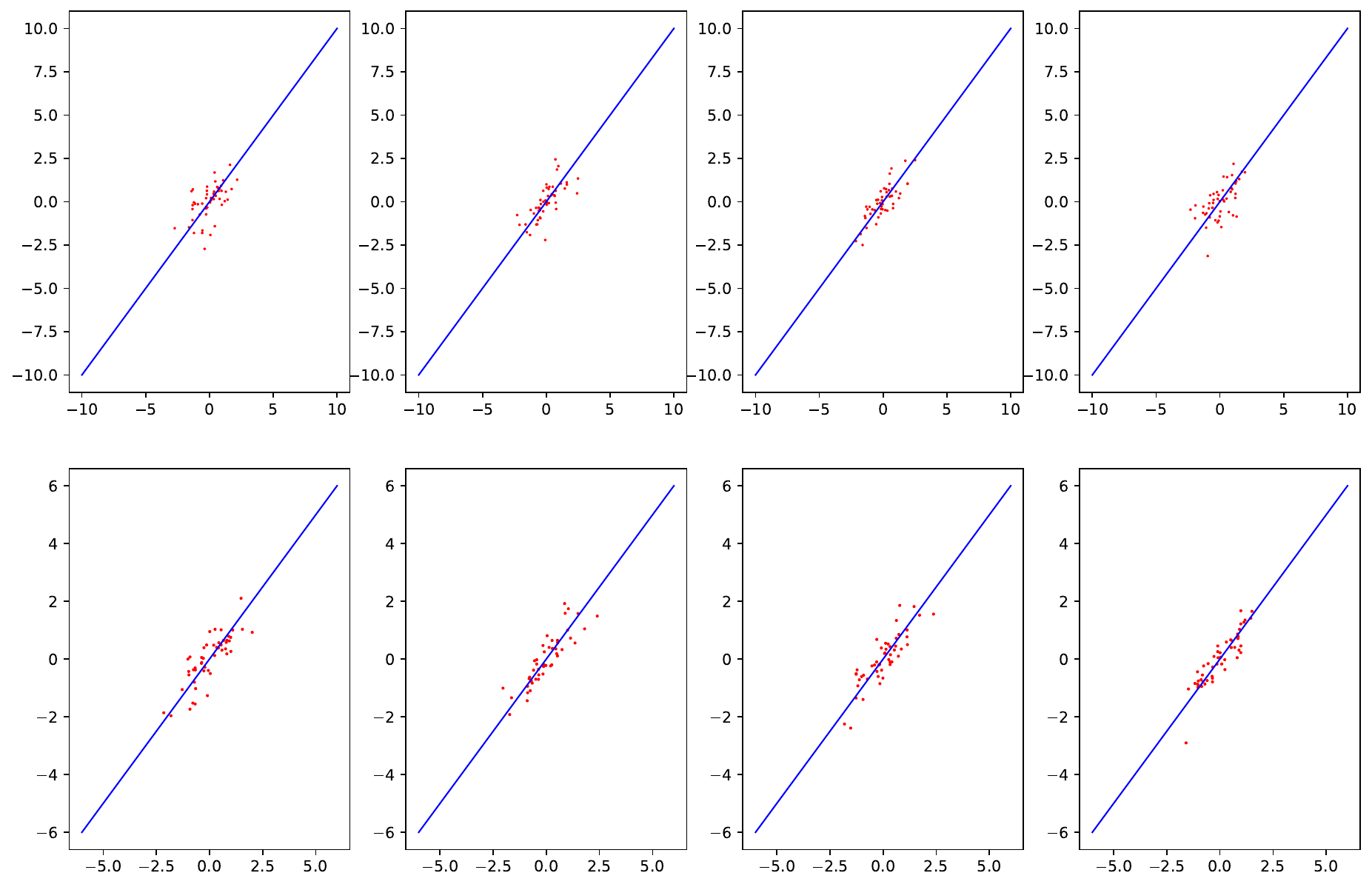}
        \\{\sf Iris Virginica data}
    \caption{\footnotesize\sf Q-Q plots for two samples, each of size $50$, the first one is generated from a $4$-dimensional standard Gaussian and the second one is, respectively, the standardised Iris Setosa, standardised Iris Versicolour and standardised Iris Virginica data. For each Iris variety considered, the first row displays OT Q-Q plots and the second one EOT Q-Q plots with $\veps=10^{-3}$.  The straight (blue) lines represent the line $L$.} 
    \label{fig:Iris_QQ}
\end{figure}

\subsection{Example 2: Turkish rice Osmanic data}\label{ex2:rice}
\vspace{-2ex}
\begin{figure}[H]
    \centering
    \includegraphics[width=.8\linewidth]{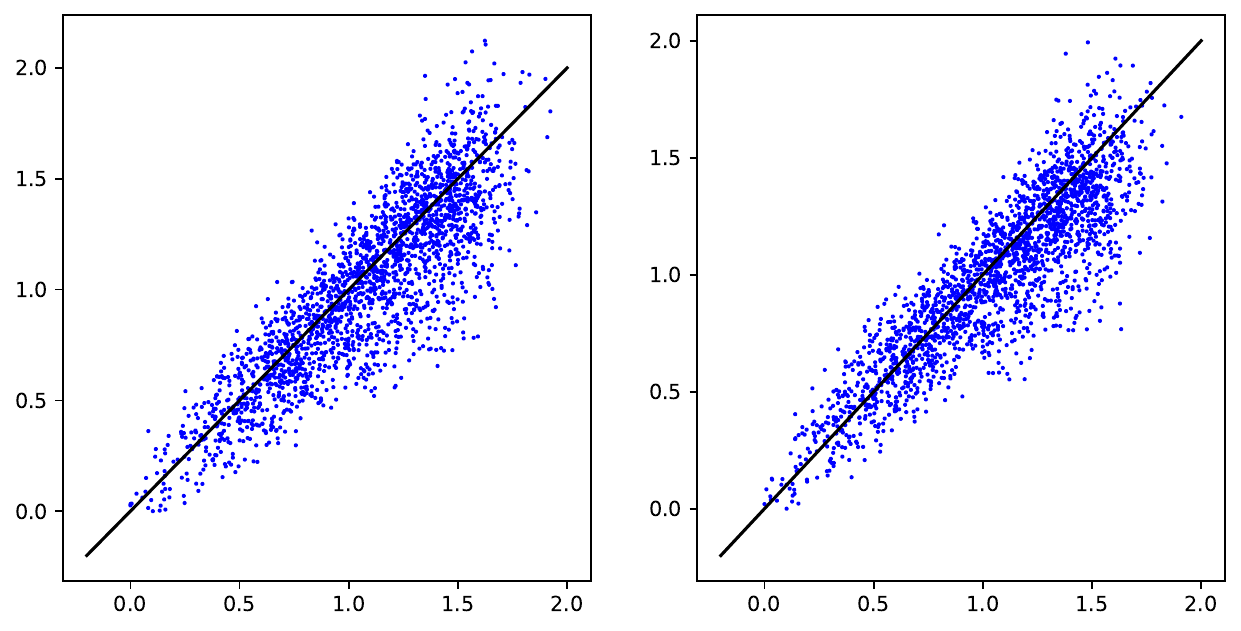}
    \vspace{-2ex}
    \caption{\small\sf Potential plots for two multivariate samples, each of size $2180$; the first one is drawn from a $5$-dimensional standard Gaussian distribution, and the second one is the standardised Turkish rice Osmanic data. The left plot is OT potential plot and the right one EOT potential plot. For EOT, we take $\veps=0.5\times 10^{-2}$. The straight (black) lines represent the line $L$.}
    \label{fig:rice_potential}
\end{figure}
\begin{figure}[H]
    \centering
    \includegraphics[width=.9\linewidth]{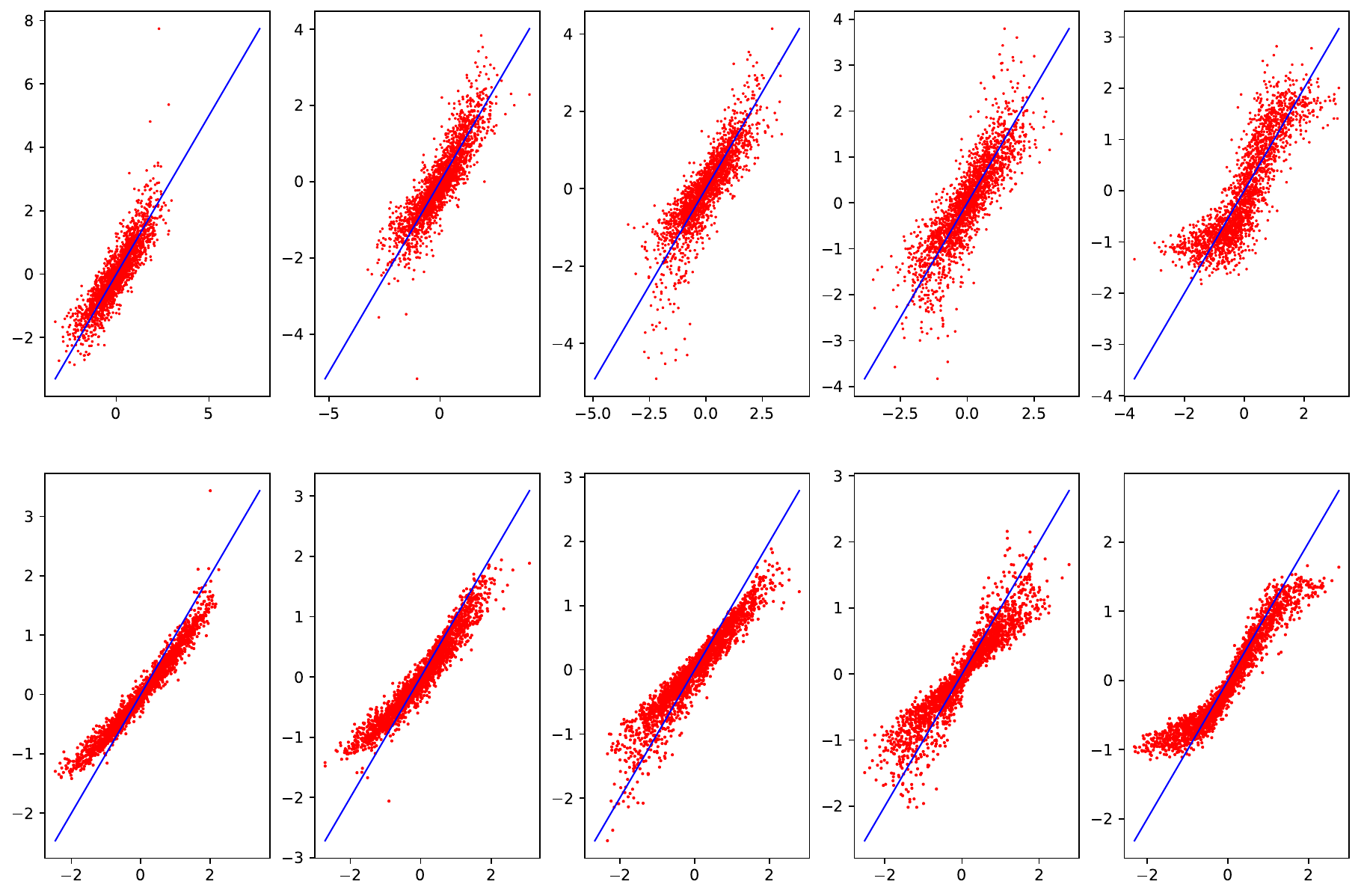}
    \vspace{-2ex}
    \caption{\small\sf Q-Q plots (left for OT, right for EOT with $\veps=0.5\times 10^{-2}$) for two multivariate samples, each of size $2180$; the first one is drawn from a $5$-dimensional standard Gaussian distribution and the second one is the standardised Turkish rice Osmanic data. The straight (black) lines represent the line $L$.}
    \label{fig:rice_QQ}
\end{figure}
In this second example, we consider the Turkish rice Osmanic dataset of size $2180$, with 5 variables that correspond to some features of the rice: Perimeter, Major Axis Length,	Minor Axis Length, Convex Area, and Extent. 
As previously, we want to compare the 5-dimensional standardised empirical distribution of this data with a $5$-dimensional standard Gaussian distribution.
% which parameters are chosen as the empirical mean (denoted $\hat m$) and empirical covariance  matrix (denoted $\hat \Sigma$) computed on the rice data considered.
%
We draw the (E)OT potential plots (see Figure~\ref{fig:rice_potential}), followed by the (E)OT Q-Q plots (see Figure~\ref{fig:rice_QQ}). For the EOT approach, we choose the regularisation parameter $\veps$ to be $0.5\times10^{-2}$.
\\[1ex]
Looking at Figure~\ref{fig:rice_potential}, we observe that the points are  spread out around the line $L$ (not as concentrated along $L$ as e.g. in Figure ~\ref{fig:sG_vs_sG_potentials_diffparameters}). So, we would deduce that the underlying distribution of the rice sample may not be well modelled by a multivariate standard Gaussian distribution. Let us move to the Q-Q plots in Figure~\ref{fig:rice_QQ} to gain more visual insights.
There, both OT and EOT Q-Q plots clearly indicate that the rice sample is drawn from a non-Gaussian distribution. Moreover, we observe a heavy tail behavior, for each component. 

Since the sample size is reasonable, let us complete this empirical analysis by computing the test statistics $E_n$ and $F_n$ (defined in \eqref{dfn:teststat}) and the corresponding $p$-values under the null hypothesis that the measure of the data (denoted $\nu_{\text{rice}}$) is Gaussian (denoted $\nu_G$), {\it i.e.} $H_0: \nu_{\text{rice}}|_{K_2}=\nu_G|_{K_2}$.  
\vspace{.3cm}
\begin{table}[h!]
    \begin{tabularx}{1\textwidth} { 
        | >{\centering\arraybackslash}X 
        | >{\centering\arraybackslash}X 
        | >{\centering\arraybackslash}X 
        | >{\centering\arraybackslash}X 
        | >{\centering\arraybackslash}X|}
       \hline
        $n$& $E_n$ & $p$-value of $E_n$ & $F_n$ & $p$-value of $F_n$\\
       \hline
       $250$  & $277.57$  & $0.44825$ & $ 76.73$ & $0.00025$ \\
      \hline
       $500$ & $469.02$ & $0$ & $180.31$ & $0$\\
      \hline
      $1000$  & $950.45$  & $0$ & $312.65$ & $0$ \\
      \hline
      \end{tabularx}      
      \caption{\small\sf Test statistics and $p$-values under the null hypothesis $H_0: \nu_{\text{rice}}|_{K_2}=\nu_G|_{K_2}$.}
  \label{table:pvalue_riceG}
\end{table}

The $p$--values reported in Table~\ref{table:pvalue_riceG} are small, as we could expect from what we observed on the various plots. This supports the assertion that the Turkish rice Osmanic data does not come from a Gaussian distribution.

\section{Comparison with geometric Q-Q plot}
\label{sec:comparisonGeom}

In this section, we aim at comparing the QQ-plots as graphical tools (to compare two distributions) when adopting two main approaches, the (E)OT one and the geometric one. To do so, we display the OT, EOT, and geometric Q-Q plots for various simulated and real datasets, then, focus on comparing if one of the methods provides more informative or stronger visual insights about the two datasets being compared. For the procedure regarding geometric Q-Q plots, we refer the reader to \cite{Dhar2014}.

\subsection{Simulated data}

We begin by considering two sets of samples, each containing 1000 observations. The first set is drawn from a 5-dimensional standard Gaussian distribution, while the second set is drawn from a 5-dimensional distribution, chosen as follows. The marginals are independent, with the first four following a univariate standard Gaussian distribution, and the fifth one following a univariate Pareto distribution with parameter $3.2$. 
%Before providing Q-Q plots, we first standardize the two samples so that sample mean and covariance matrix becomes zero and identity matrix. Standardizing the samples is a necessary requirement for the geometric Q-Q plot.
In Figure~\ref{fig:5d_sG_vs_4sG+1sP32}, we depict $3$ types of multivariate Q-Q plots: the first row represents the OT Q-Q plots, the second row shows the EOT Q-Q plots when taking $\veps=0.5\times10^{-2}$, and the last row displays the geometric Q-Q plots.
%
%{\color{red}[Fig to be replaced by non standardized one]}
\begin{figure}[H]
    \centering
    \includegraphics[width=.8\linewidth]{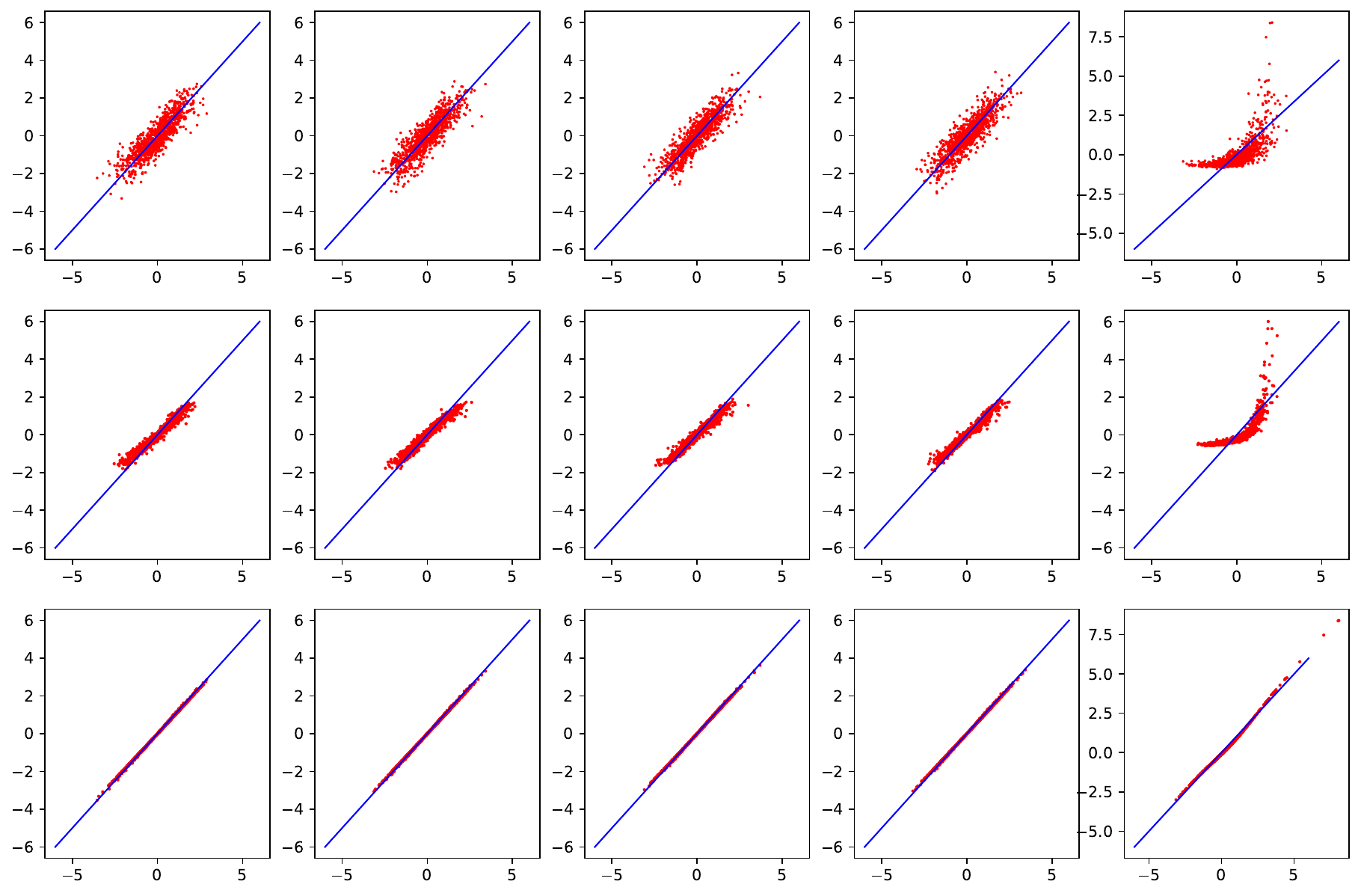}
        \caption{\small\sf Q-Q plots for two samples, the first one is drawn from a $5$-dimensional standard normal distribution, the second one from a $5$-dimensional distribution whose first $4$ marginals are standard normal and the last one is Pareto with parameter $3.2$, all the marginals being independent of each other. The first row displays OT Q-Q plots, the second EOT Q-Q plots with $\veps=0.5\times10^{-2}$, and the last one geometric Q-Q plots. The straight (blue) lines represent the line $L$.}
    \label{fig:5d_sG_vs_4sG+1sP32}
\end{figure}
We observe in Figure~\ref{fig:5d_sG_vs_4sG+1sP32} that OT and EOT Q-Q plots clearly reveal the presence of a heavy tail for the $5$th component, whereas it is not so obvious from the geometric one. In the 5th component of the geometric Q-Q plot, we see some slight movement away from the straight line $L$, but very mild in comparison with the OT and EOT approaches.

Next, we perform a similar experiment but replace the 5th marginal distribution for the second sample by a Student's $t$-distribution with $3.2$ degrees of freedom. %As before, we standardize the two samples to proceed,  
The OT, EOT, and geometric Q-Q plots are displayed in Figure~\ref{fig:5d_sG_vs_4sG+1Stud_t_32}.

We observe that, while the OT and EOT Q-Q plots reveal the presence of heavy tails in the second distribution, the geometric Q-Q plot does not. Moreover, in the latter case, the point cloud becomes highly aligned with the straight line $L$, suggesting that the two samples are drawn from the same distribution.

%{\color{red}[Fig to be replaced by non standardized one]}
\begin{figure}[H]
    \centering
    \includegraphics[width=.8\linewidth]{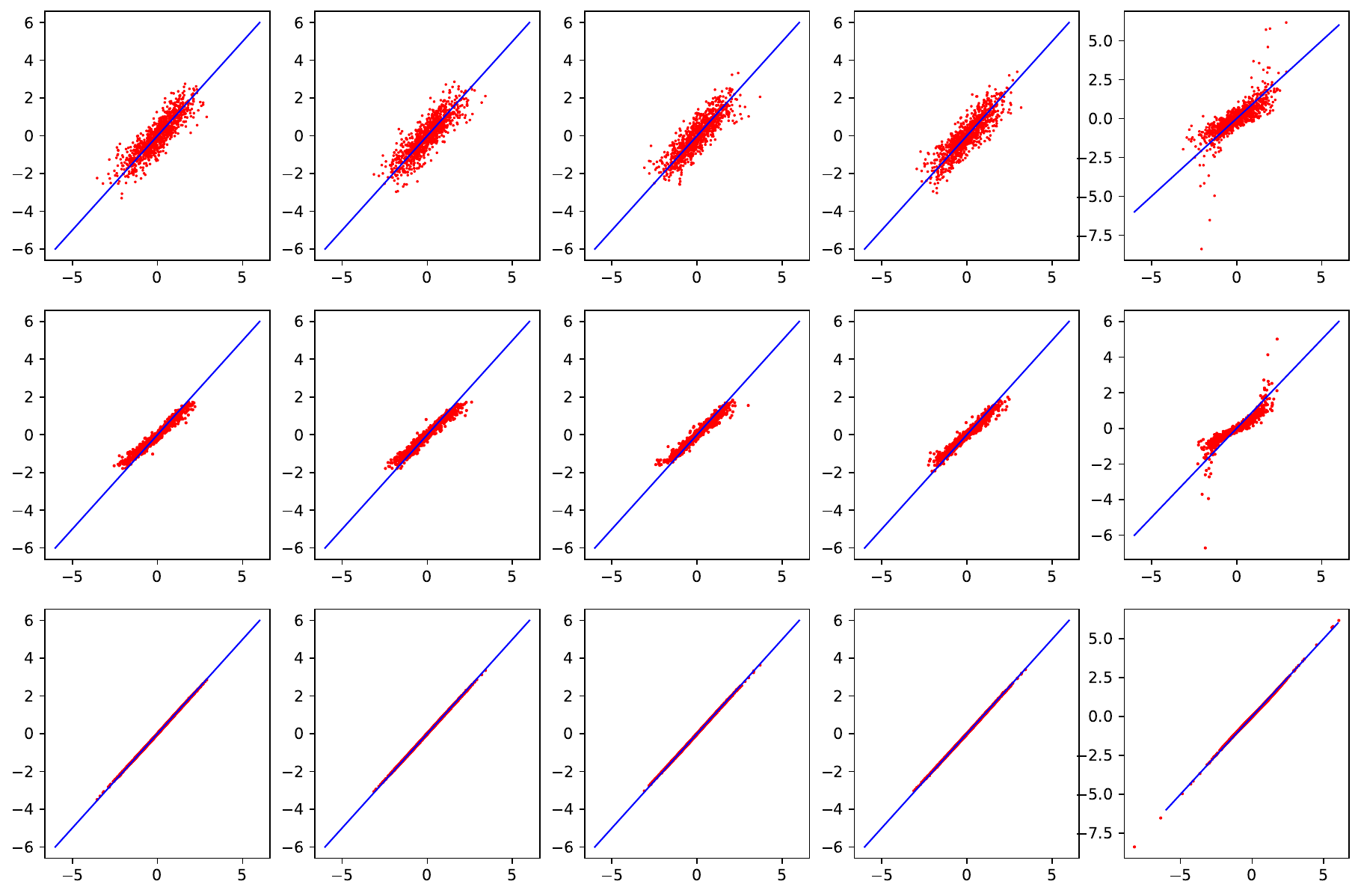}
        \caption{\small\sf Q-Q plots for two samples each of size $1000$, the first one is drawn from a $5$-dimensional standard normal distribution, the second one from a $5$-dimensional distribution whose first $4$ marginals are standard normal and the last one is a Student's t-distribution with $3.2$ degrees of freedom, all marginals being independent of each other. The first row displays OT Q-Q plots, the second one EOT Q-Q plots with $\veps=0.5\times10^{-2}$, and the last one geometric Q-Q plot. The straight (blue) lines represent the line $L$.}
    \label{fig:5d_sG_vs_4sG+1Stud_t_32}
\end{figure}

\subsection{Real data}
\label{subsec:osmanic}
Here, we go back to Example~\ref{ex2:rice}, choosing the Turkish rice Osmanic dataset as second sample, while the first one is drawn from the $5$-dimensional standard Gaussian distribution. 

We note that both OT and EOT Q-Q plots indicate that the second sample is drawn from a distribution that is different from the Gaussian distribution. Additionally, we observe distinct tail behavior in the third and fifth components. However, the geometric Q-Q plot fails to differentiate between the two samples.

In Appendix~\ref{App:comparisonGeom}, we perform a similar comparison, but now comparing the unaltered real data with a multivariate Gaussian distribution whose parameters are learned from the data.
%where the first one drawn from the $5$-dimensional standard Gaussian distribution. The second sample consists of $2180$ observations of five features of Turkish rice Osmanic which are Perimeter, Major Axis Length,	Minor Axis Length, Convex Area, and Extent.  The first sample is as before drawn from $5$ dimensional standard Gaussian distribution. We then standardize both the samples and the Q-Q plots are displayed in Figure~\ref{fig:sG_vs_Osmanic-Std}.

\begin{figure}[H]
    \centering
    \includegraphics[width=.8\linewidth]{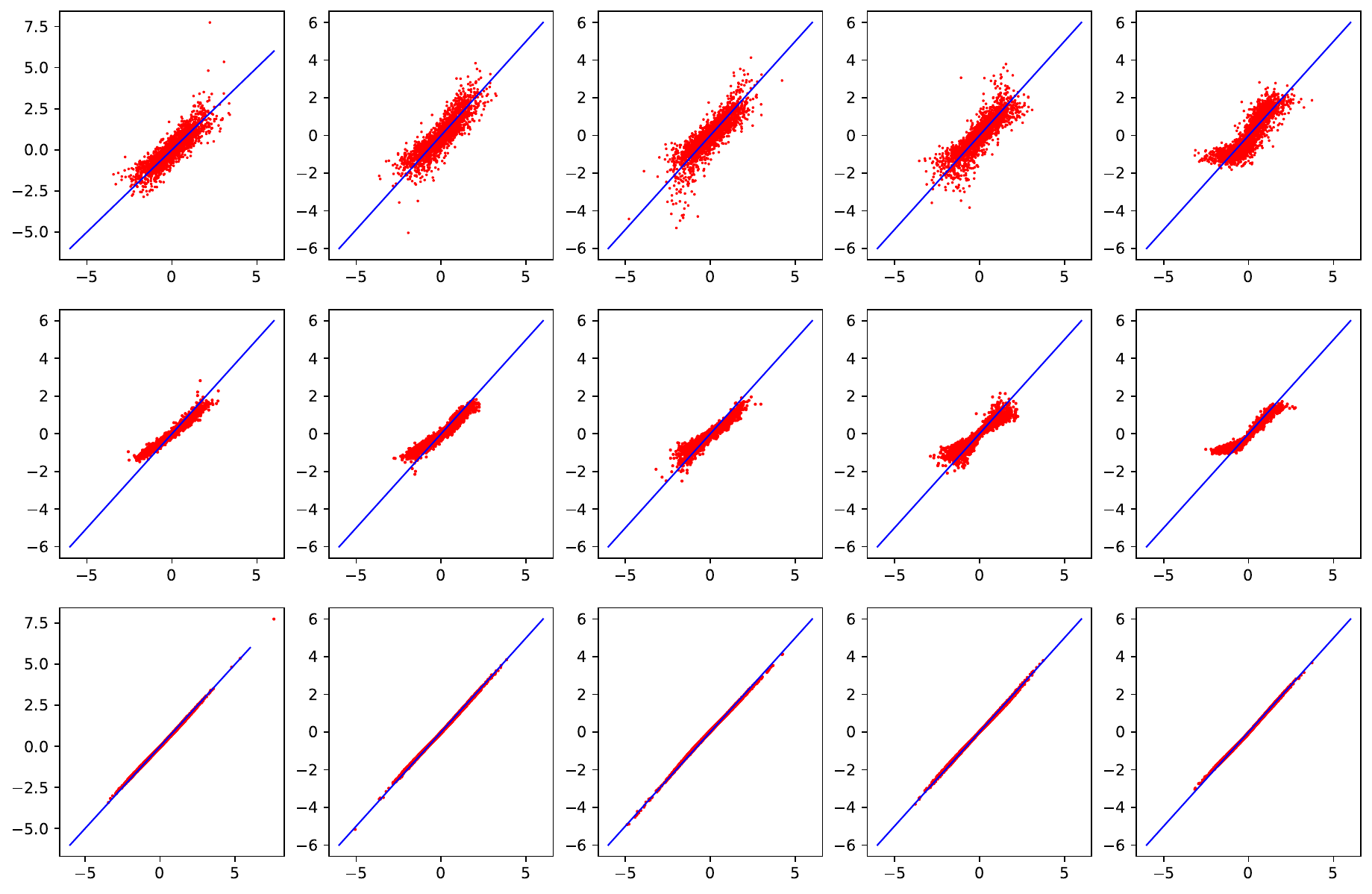}
        \caption{\small\sf Q-Q plots for two samples, each of size $2180$. The first sample is drawn from a $5$-dimensional standard normal distribution, while the second sample consists of $2180$ observations of standardised five features of Turkish rice Osmanic. The first row displays OT Q-Q plot, the second one EOT Q-Q plot with $\veps=0.5\times10^{-2}$ and the last one geometric Q-Q plot.}
    \label{fig:sG_vs_Osmanic-Std}
\end{figure}

\section{Conclusion} 
\label{sec:Q-Qconclusion}

In recent times, multivariate OT based quantiles have become a very important object of research because it reflects many useful theoretical properties of univariate quantiles. Although OT is well understood theoretically, it is computationally very challenging for a large sample; it lacks statistical stability. Many possible solutions have been proposed to circumvent these issues, entropy regularisation (EOT) being one of them. 

In this setting, considered both OT and EOT (maps and potentials) and obtained theoritical results that show that these (E)OT tools can be used to characterise distributions uniquely. Based on this property, we built multivariate Q-Q plots and potential plots by which one can visually compare two multivariate distributions. 

(E)OT potential plots are interesting visual tools when comparing multivariate distributions, as they provide a $2$-dimensional  representation, which is an advantage w.r.t. Q-Q plots, especially in high dimension. However, Q-Q plots provide a visual comparison for specific features (scaling, shifting, outliers, tail heaviness) which is more obvious and readable than potential do. 

Now, comparing OT versus EOT approach, we could observe that EOT Q-Q plots, with small regularization parameter, can determine well if two distributions are statistically similar or not, but they do not reveal informations about special features such as tail heaviness. Note that it is important to choose the value of $\veps$ sufficiently small so that the EOT map closely approximates the OT map, otherwise, for big values of $\eps$, the EOT map becomes close to a constant map. It is also observed that the smaller the values of $\veps$, the longer the algorithm takes to execute. Also, if $\veps$ is chosen very small, it happens sometimes that the code in \cite{flamary2021pot} breaks down because of encountering NaN value.

Finally, we proposed statistical tests associated with the EOT Q-Q and potential plots, based on the available statistical stability results (central limit theorems) for EOT. Once limit theorems will be available for OT, we will be able to construct test statistics for OT as well, in a similar way as for EOT. Note that limit theory for OT is an area of active research, for recent developments on this topic, we refer e.g. to \cite{sadhu2023stability, manole2023central}. 

Our future goal is to study the asymptotic of empirical estimators of the (E)OT quantile and potential function, specifically their tail behavior, as we did for geometric quantiles and half-space depths (see \cite{preprint1SKV2023}). It will be helpful to have a better understanding of such statistical tools in view of their applications.

\bibliographystyle{plainnat}
\bibliography{LitSibsankar2}

\begin{thebibliography}{23}
\providecommand{\natexlab}[1]{#1}
\providecommand{\url}[1]{\texttt{#1}}
\expandafter\ifx\csname urlstyle\endcsname\relax
  \providecommand{\doi}[1]{doi: #1}\else
  \providecommand{\doi}{doi: \begingroup \urlstyle{rm}\Url}\fi

\bibitem[Brenier(1991)]{Brenier1991PolarFA}
Y.~Brenier.
\newblock Polar factorization and monotone rearrangement of vector-valued
  functions.
\newblock \emph{Communications on Pure and Applied Mathematics}, 44:\penalty0
  375--417, 1991.

\bibitem[Chernozhukov et~al.(2017)Chernozhukov, Galichon, Hallin, and
  Henry]{Chernozhukov2017}
V.~Chernozhukov, A.~Galichon, M.~Hallin, and M.~Henry.
\newblock Monge-kantorovich depth, quantiles, ranks and signs.
\newblock \emph{Annals of Statistics}, 45\penalty0 (1):\penalty0 223--256,
  2017.

\bibitem[Cuturi(2013)]{Cuturi2013}
M.~Cuturi.
\newblock Sinkhorn distances: Lightspeed computation of optimal transportation
  distances.
\newblock \emph{Advances in Neural Information Processing Systems},
  26:\penalty0 2292--2300, 2013.

\bibitem[de~Valk and Segers(2018)]{Valk2018}
C.~de~Valk and J.~Segers.
\newblock Tails of optimal transport plans for regularly varying probability
  measures.
\newblock arXiv preprint:1811.12061, 2018.

\bibitem[Dhar et~al.(2014)Dhar, Chakraborty, and Chaudhuri]{Dhar2014}
S.~Dhar, B.~Chakraborty, and P.~Chaudhuri.
\newblock Comparison of multivariate distributions using quantile--quantile
  plots and related tests.
\newblock \emph{Bernoulli}, 20\penalty0 (3):\penalty0 1484--1506, 2014.

\bibitem[Easton and McCulloch(1990)]{Easton1990}
G.S. Easton and R.E. McCulloch.
\newblock A multivariate generalization of quantile-quantile plots.
\newblock \emph{Journal of the American Statistical Association}, 85\penalty0
  (410):\penalty0 376--386, 1990.

\bibitem[Flamary et~al.(2021)Flamary, Courty, Gramfort, Alaya, Boisbunon,
  Chambon, Chapel, Corenflos, Fatras, Fournier, Gautheron, Gayraud, Janati,
  Rakotomamonjy, Redko, Rolet, Schutz, Seguy, Sutherland, Tavenard, Tong, and
  Vayer]{flamary2021pot}
R.~Flamary, N.~Courty, A.~Gramfort, M.~Z. Alaya, A.~Boisbunon, S.~Chambon,
  L.~Chapel, A.~Corenflos, K.~Fatras, N.~Fournier, L.~Gautheron, N.T.H.
  Gayraud, H.~Janati, A.~Rakotomamonjy, L.~Redko, A.~Rolet, A.~Schutz,
  V.~Seguy, D.J. Sutherland, R.~Tavenard, A.~Tong, and T.~Vayer.
\newblock Pot: Python optimal transport.
\newblock \emph{Journal of Machine Learning Research}, 22\penalty0
  (78):\penalty0 1--8, 2021.
\newblock URL \url{http://jmlr.org/papers/v22/20-451.html}.

\bibitem[Genevay(2019)]{Genevay2019}
A.~Genevay.
\newblock Entropy-regularized optimal transport for machine learning.
\newblock Ph.D. thesis, 2019.
\newblock URL \url{https://audeg.github.io/publications/these_aude.pdf}.

\bibitem[Ghosal and Sen(2022)]{Ghosal2022}
P.~Ghosal and B.~Sen.
\newblock Multivariate ranks and quantiles using optimal transport:
  Consistency, rates and nonparametric testing.
\newblock \emph{The Annals of Statistics}, 50\penalty0 (2):\penalty0
  1012--1037, 2022.

\bibitem[Goldfeld et~al.(2023)Goldfeld, Kato, Rioux, and
  Sadhu]{goldfeld2023limit}
Z.~Goldfeld, K.~Kato, G.~Rioux, and R.~Sadhu.
\newblock Limit theorems for entropic optimal transport maps and the sinkhorn
  divergence.
\newblock \emph{arXiv preprint arXiv:2207.08683}, 2023.

\bibitem[Hallin et~al.(2021)Hallin, Del~Barrio, Cuesta-Albertos, and
  Matr{\'a}n]{Hallin2021}
M.~Hallin, E.~Del~Barrio, J.~Cuesta-Albertos, and C.~Matr{\'a}n.
\newblock Distribution and quantile functions, ranks and signs in dimension d:
  A measure transportation approach.
\newblock \emph{Annals of Statistics}, 49:\penalty0 1139--1165, 2021.

\bibitem[H{\"u}tter and Rigollet(2021)]{Huetter2021}
J.~C. H{\"u}tter and P.~Rigollet.
\newblock {Minimax estimation of smooth optimal transport maps}.
\newblock \emph{The Annals of Statistics}, 49\penalty0 (2):\penalty0 1166 --
  1194, 2021.

\bibitem[Manole et~al.(2023)Manole, Balakrishnan, Weed, and
  Wasserman]{manole2023central}
T.~Manole, S.~Balakrishnan, J.N. Weed, and L.~Wasserman.
\newblock Central limit theorems for smooth optimal transport maps.
\newblock \emph{arXiv:2312.12407}, 2023.

\bibitem[McCann(1995)]{McCann1995Existence}
R.~J. McCann.
\newblock Existence and uniqueness of monotone measure-preserving maps.
\newblock \emph{Duke Mathematical Journal}, 80:\penalty0 309--323, 1995.

\bibitem[Nutz(2022)]{Nutznote2022}
M.~Nutz.
\newblock Introduction to entropic optimal transport.
\newblock Lecture notes, available at
  \url{https://www.math.columbia.edu/~mnutz/docs/EOT_lecture_notes.pdf}, 2022.

\bibitem[Pooladian and Niles-Weed(2022)]{pooladian2022entropic}
A.-A. Pooladian and J.~Niles-Weed.
\newblock Entropic estimation of optimal transport maps.
\newblock arXiv:2109.12004, 2022.

\bibitem[Sadhu et~al.(2023)Sadhu, Goldfeld, and Kato]{sadhu2023stability}
R.~Sadhu, Z.~Goldfeld, and K.~Kato.
\newblock Stability and statistical inference for semidiscrete optimal
  transport maps.
\newblock arXiv:2303.10155, 2023.

\bibitem[Santambrogio(2015)]{Santambrogio2015}
F~Santambrogio.
\newblock \emph{Optimal Transport for Applied Mathematicians: Calculus of
  Variations, PDEs, and Modeling}.
\newblock Springer, 2015.

\bibitem[Shapiro and Wilk(1965)]{Shapiro1965}
S.~S. Shapiro and M.~B. Wilk.
\newblock An analysis of variance test for normality: {C}omplete samples.
\newblock \emph{Biometrika}, 52:\penalty0 591--611, 1965.

\bibitem[Singha et~al.(2023)Singha, Kratz, and Vadlamani]{preprint1SKV2023}
S.~Singha, M.~Kratz, and S.~Vadlamani.
\newblock From geometric quantiles to halfspace depths: A geometric approach
  for extremal behaviour.
\newblock arXiv:2306.10789 \& ESSEC WP2307, 2023.

\bibitem[van~der Vaart and Wellner(1996)]{Vaart1996}
A.~W. van~der Vaart and J.~A. Wellner.
\newblock \emph{Weak Convergence}.
\newblock Springer New York, 1996.

\bibitem[Villani(2003)]{villani2003topics}
C.~Villani.
\newblock \emph{Topics in Optimal Transportation}.
\newblock American Mathematical Society, 2003.

\bibitem[Villani(2009)]{villani2009oldnew}
C.~Villani.
\newblock \emph{Optimal Transport Old and New}.
\newblock Springer, 2009.

\end{thebibliography}

\newpage 

\appendix

{\bf \Large Appendix - Additional examples}

\section{Another example for the comparison of heavy versus light tail}
\label{App:LightHT}

Let $\nu_X$ be the push forward measure of the trivariate standard normal distribution (denoted $\nu$) by the function $f(x_1, x_2, x_3)=(1+|x_1|, 1+|x_2|, 1+|x_3|)$, {\it i.e.} $\nu_X=f \sharp \nu$. Let $\nu_Y$ be a probability distribution on $\real^3$, with  i.i.d. Pareto($3$) marginals (with density function given by $p(x)=3x^{-4}, \, x \geq 1$). Note that both the distributions $\nu_X$ and $\nu_Y$ are supported on the first orthant. We consider two samples, each of size $1000$, drawn from $\nu_X$ and $\nu_Y$, respectively. Then, we follow the steps $(ii)-(iii)$ to provide the Q-Q plots corresponding to these samples, which are displayed in Figure \ref{fig:normal vs pareto}. Clearly, the scatter plots do not cluster around the straight line $L$ (in blue). We further notice that in each plot, the points exhibit a nonlinear behaviour which shows significant deviation from $L$. Observe that, in each componentwise plots, the high (or extreme) quantiles corresponding to ${\cal Y}^n$ grow faster than those of ${\cal X}^n$, implying that the sample ${\cal Y}^n$ represent a distribution with relatively heavier tail.
\begin{figure}[H] 
    \centering
    \includegraphics[width=.8\linewidth]{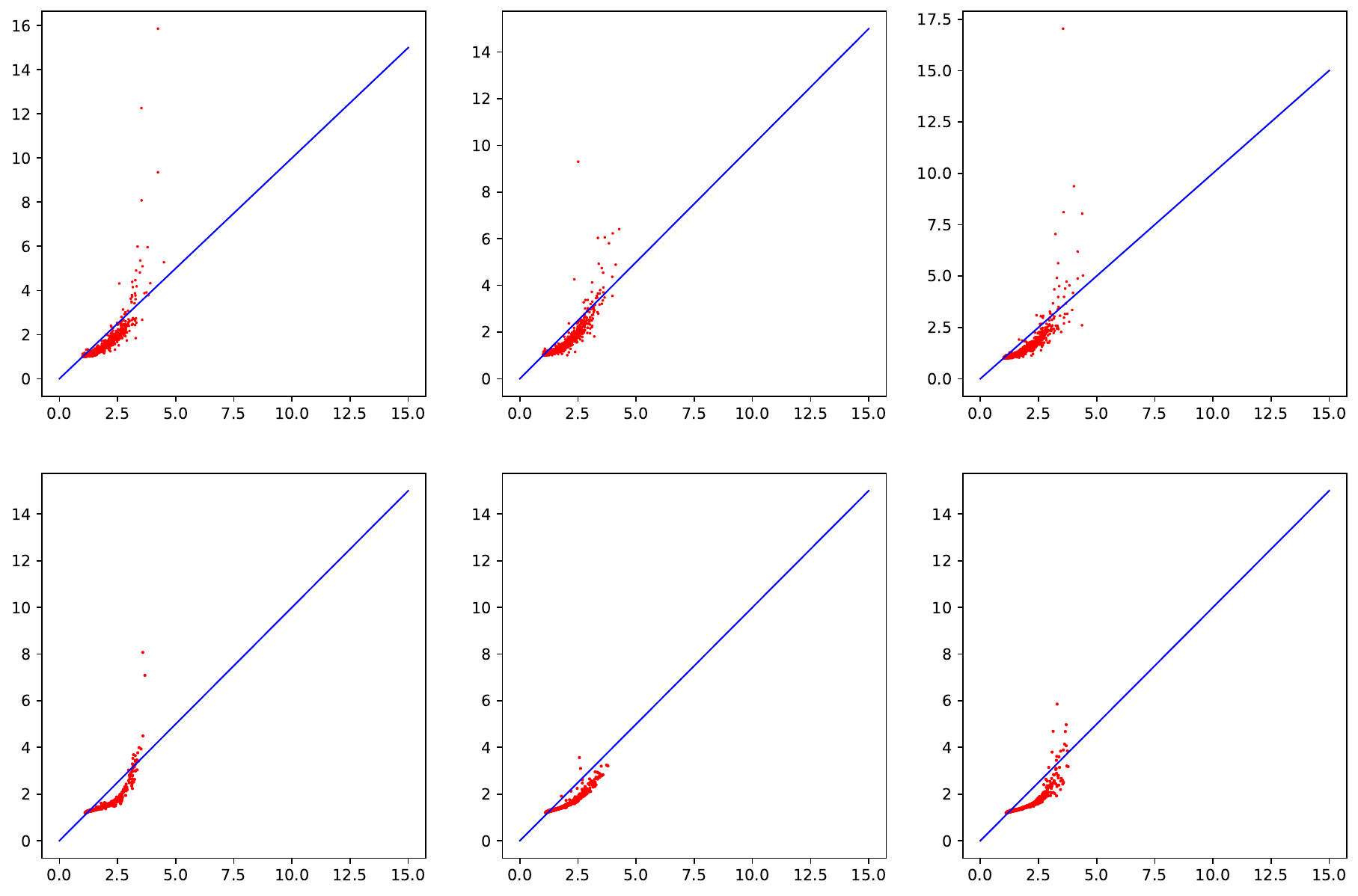}
    \caption{\small \sf Q-Q plots for two samples, the first one is drawn from $\nu_X$ and the second from $\nu_Y$. The first row displays OT Q-Q plot and the second EOT Q-Q plot with $\veps=10^{-3}$.  The straight (blue) lines represent the line $L$.}
    \label{fig:normal vs pareto}
\end{figure}
Next, the potential plots (step (iv)) for this particular example are displayed in Figure~\ref{fig:sG_vs_sP_par3_potential}. It can be observed that the points are scattered around a nonlinear (increasing) curve above the line $L$. This implies that the potential function associated with the second sample has a higher growth rate. Since the quantile is the potential gradient, higher growth rate of potential implies that, in absolute term, the quantiles are bigger, thereby suggesting that the corresponding distribution is heavy tailed.
\begin{figure}[H]
    \centering
    \includegraphics[width=.8\linewidth]{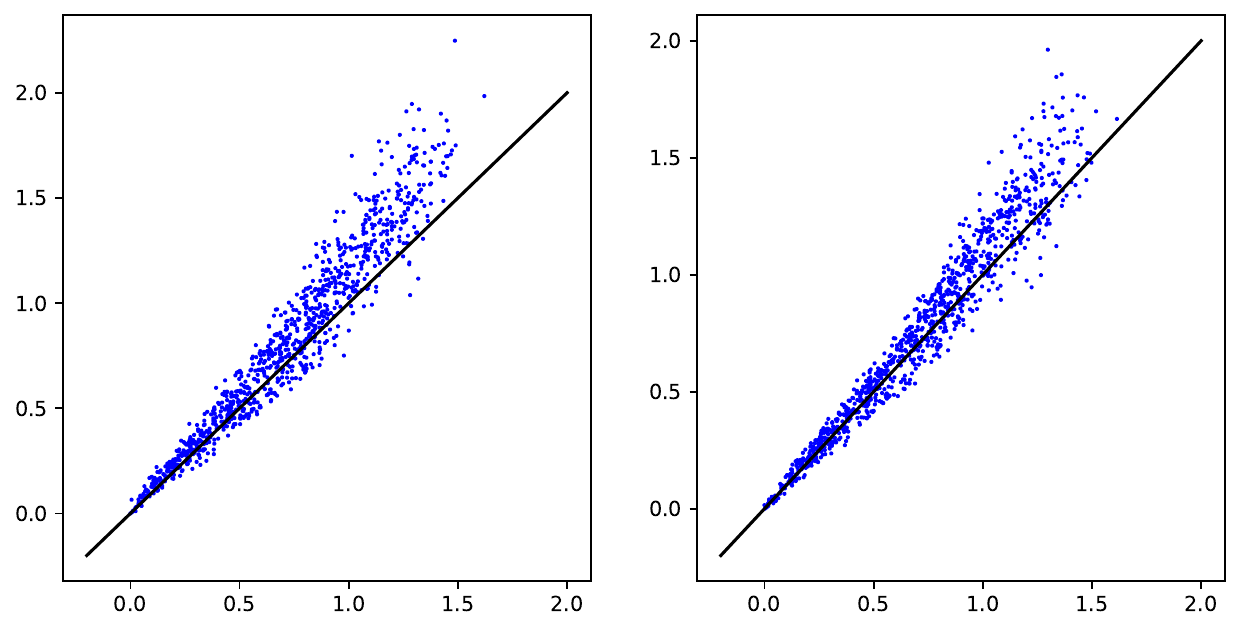}
    \caption{\small\sf Potential plots (left for OT, right for EOT with $\veps=10^{-3}$) for two multivariate samples, each of size $1000$; the first sample is drawn from $\nu_X$ and the second from $\nu_Y$. The straight (black) lines represent the line $L$.}
    \label{fig:sG_vs_sP_par3_potential}
\end{figure}

\section{Comparison among OT, EOT and geometric Q-Q plots }
\label{App:comparisonGeom}

We consider the same examples on simuated and real data as in Section~\ref{sec:comparisonGeom}, but now we perform the comparison  of the raw data with an appropriate multivariate Gaussian distribution whose parameters (mean and covariance for instance) are learned from the raw data.

\subsection{Simulated data}

We begin by considering two samples, each containing 1000 observations.

Using the notation set forth in previous sections, $\calY^n$ is the sample to be compared; it is drawn from a $5$-dimensional distribution with independent marginals, such that the first four marginals follow a univariate standard Gaussian distribution, and the fifth marginal follows a univariate Pareto distribution with parameter $3.2$. The sample $\calX^n$ is drawn from a $5$-dimensional Gaussian distribution with mean ${\hat m}$ and covariance ${\hat \Sigma}$, where ${\hat m}$ and ${\hat \Sigma}$ are the sample mean and the sample covariance, respectively, of $\calY^n$. 

In Figure ~\ref{fig:5d_sG_estpar_vs_4sG+1sP32}, we depict various types of multivariate Q-Q plots: the first row represents the OT Q-Q plot, the second row shows the EOT Q-Q plot, and the last row displays the geometric Q-Q plot.

\begin{figure}[H]
    \centering
    \includegraphics[width=.8\linewidth]{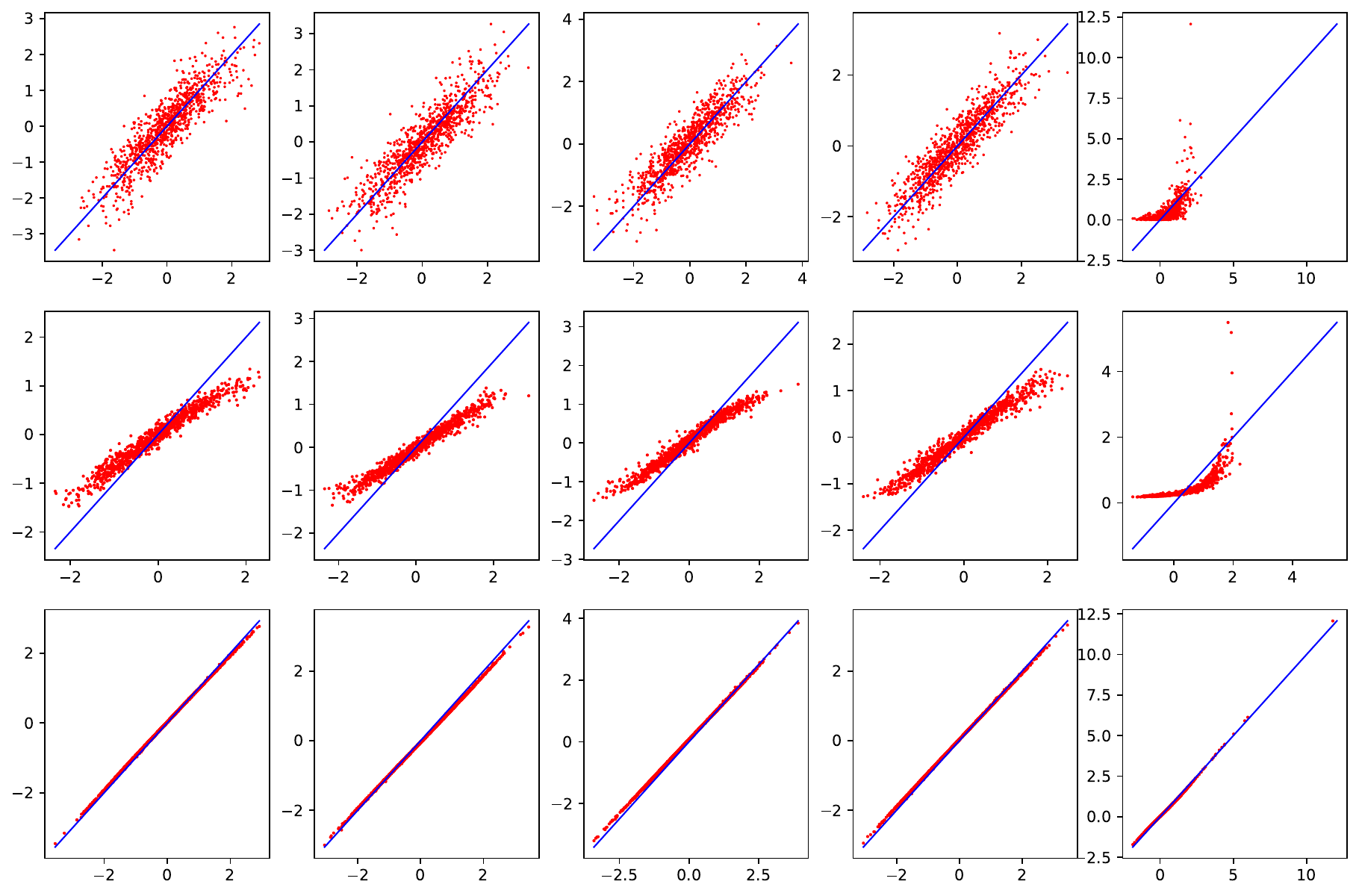}
        \caption{\small\sf Q-Q plots for two samples, the first one is drawn from a $5$-dimensional normal $(\hat m, \hat\sigma)$ distribution, the second one from a $5$-dimensional distribution whose first $4$ marginals are standard normal and the last one is Pareto with parameter $3.2$, all the marginals being independent of each other. The first row displays OT Q-Q plots, the second one EOT Q-Q plot with $\varepsilon = 0.5\times 10^{-2}$, and the last one geometric Q-Q plots. The straight (blue) lines represent $L$.}
    \label{fig:5d_sG_estpar_vs_4sG+1sP32}
\end{figure}
We observe from Figure~\ref{fig:5d_sG_estpar_vs_4sG+1sP32}, that the OT and EOT Q-Q plots clearly reveal about the presence of a heavy tail in the fifth component, whereas it is not so obvious from the geometric one. In the fifth component of the geometric Q-Q plot, we see some slight deviation from the straight line $L$, but it is really mild in comparison to the OT and EOT.\\ [1ex]

Next, we perform a similar experiment but with different marginal distribution. $\calY^n$ is drawn from a $5$--dimensional distribution, where the marginals are independent of each other. The first four marginals follow a univariate standard Gaussian distribution, while the last marginal follows a Student's $t$-distribution with degree of freedom $3.2$. $\calX^n$ is drawn from the $5$--dimensional multivariate Gaussian distribution with mean ${\hat m}$ and covariance ${\hat \Sigma}$, where ${\hat m}$ and ${\hat \Sigma}$ are the sample mean and sample covariance of $\calY^n$, respectively. The OT, EOT, and geometric Q-Q plots are displayed in Figure~\ref{fig:5d_sG_estpar_vs_4sG+1st32}.

\begin{figure}[H]
    \centering
    \includegraphics[width=.8\linewidth]{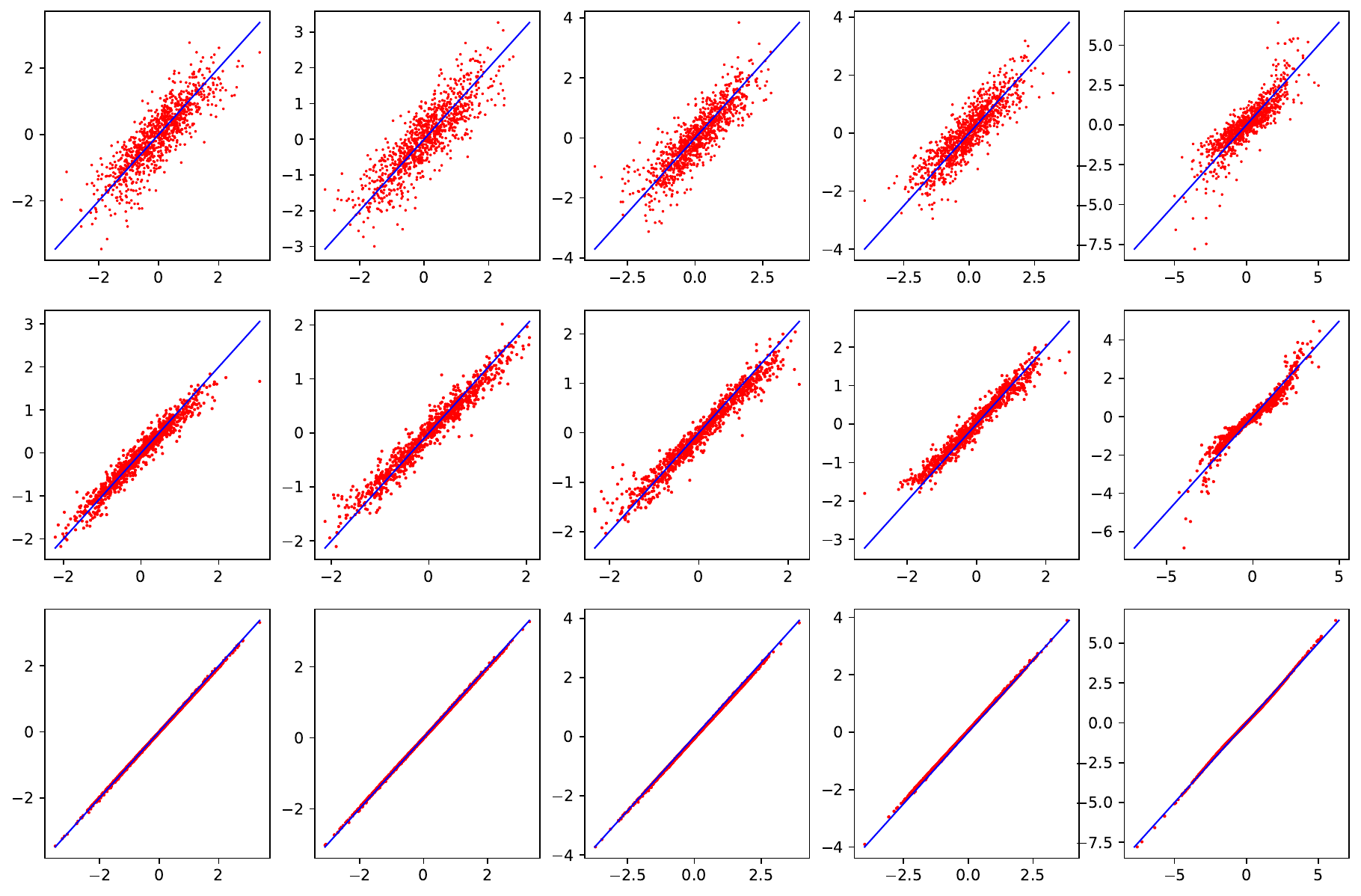}
    \caption{\small\sf Q-Q plots for two samples each of size $1000$, the first one is drawn from a $5$-dimensional normal $(\hat m, \hat\Sigma)$ distribution, the second one from a $5$-dimensional distribution whose first $4$ marginals are standard normal and the last one is Studen's t distribution with parameter $3.2$ and all the marginals are independent of each other. The first row displays OT Q-Q plot, the second one EOT  Q-Q plot with $\varepsilon=0.5\times 10^{-2}$, and the last one geometric Q-Q plot. The straight lines in blue represent the line $L$. }
    \label{fig:5d_sG_estpar_vs_4sG+1st32}
\end{figure}

We observe that, while the OT and EOT Q-Q plots reveal the presence of heavy tails in the $5$ th margnal, the geometric Q-Q plot does not. Moreover, in the latter case, the point cloud becomes highly aligned with the straight line $L$, suggesting that the two samples are drawn from the same distribution.

\subsection{Real data}

Now let us move to real data, considering the same data as in Section~\ref{subsec:osmanic}. The sample $\calY^n$ consists of $2180$ observations of five features of Turkish rice Osmanic (Perimeter, Major Axis Length,	Minor Axis Length, Convex Area, and Extent). $\calX^n$ is the sample drawn from the $5$--dimensional Gaussian distribution with mean ${\hat m}$ and covariance ${\hat \Sigma}$, where ${\hat m}$ and ${\hat \Sigma}$ are the sample mean and the sample covariance, respectively, of $\calY^n$. The Q-Q plots are displayed in Figure~\ref{fig:RiceEstimatedMoments-Gauss}.

\begin{figure}[H]
    \centering
    \includegraphics[width=.8\linewidth]{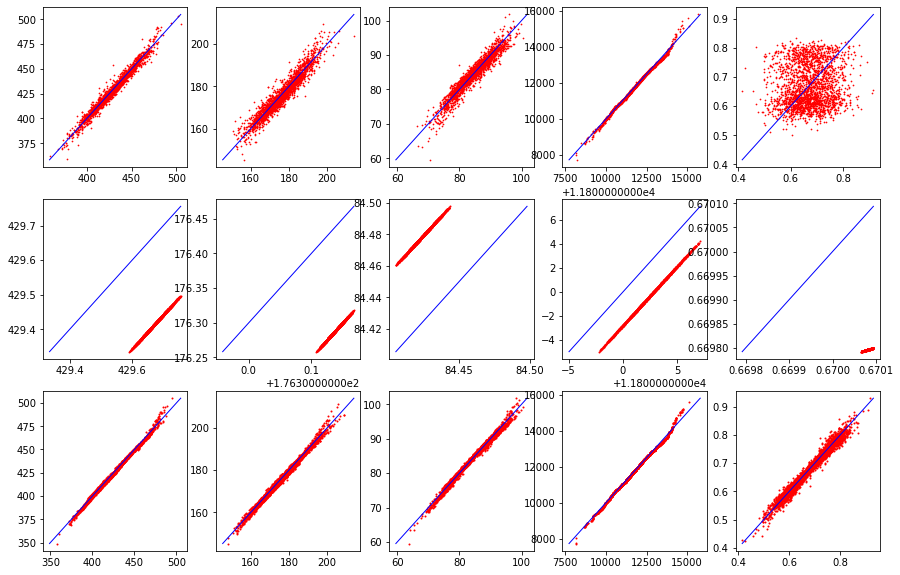}
        \caption{\small\sf Q-Q plots for two samples, each of size $2180$. The first sample is drawn from the $5$-dimensional Gaussian distribution with mean $\hat{m}$ and covariance matrix $\hat{\Sigma}$, while the second sample consists of $2180$ observations of five features of Turkish rice Osmanic. The first row displays OT Q-Q plots, the second one EOT Q-Q plots with $\veps=0.5\times10^{-2}$ and the last one geometric Q-Q plots. The straight (blue) lines represent $L$.}
    \label{fig:RiceEstimatedMoments-Gauss}
\end{figure}

We observe that all the Q-Q plots indicate that the second sample is drawn from a distribution different from the Gaussian distribution with mean $\hat{m}$ and covariance matrix $\hat{\Sigma}$. Notably, the EOT Q-Q plots look very different as compared to the OT and geometric Q-Q plots, possibly due to the regularization parameter $\veps$ not being sufficiently small. Due to computational limitations, we were unable to select $\veps$ smaller than $0.5 \times 10^{-2}$. On the other hand, the OT Q-Q plots more prominently highlight the dissimilarities between the two samples, particularly evident in the $5$-th component (note that the scaling is very different when comparing with other components) of the plots. It is important to note that none of the plots exhibit clear features of the distribution, like heavy tails. Compared to the Q-Q plots conducted after standardizing both samples in Example~\ref{ex2:rice} and Section~\ref{subsec:osmanic}, this drawback indicates that standardizing the samples when working with real data might be a better approach.

%This could perhaps be attributed to (1) the high dimensionality or (2) standardizing the samples or (3) an insufficient number of sample points. In the following, we will justify the aforementioned points through some examples. There is possibility that 

\end{document}